\documentclass{scrartcl}
\usepackage{amsmath,amssymb,amsthm,bm,color,enumerate,fullpage,natbib,url}

\setcitestyle{round}
\usepackage{xspace}
\usepackage{algorithm2e}
\usepackage{graphicx}
\usepackage[hypcap=false]{caption}
\usepackage{subcaption}
\usepackage[utf8]{inputenc}
\usepackage{lscape}
\usepackage{centernot}
\usepackage{hyperref}
\usepackage{appendix}

\newtheorem{theorem}{Theorem}
\newtheorem*{theorem*}{Theorem}

\newtheorem{proposition}[theorem]{Proposition}
\theoremstyle{definition}
\newtheorem{definition}[theorem]{Definition}
\newtheorem{example}[theorem]{Example}

\newtheorem{assumption}[theorem]{Assumption}

\numberwithin{theorem}{section}
\numberwithin{equation}{section} 

\usepackage{mathtools}
\newcommand\numberthis{\addtocounter{equation}{1}\tag{\theequation}}
\newcommand*{\defeq}{\coloneqq}

\newcommand{\wrt}{\mathrm{d}}
\newcommand{\rd}{\mathrm{d}}
\newcommand{\set}[1]{\left\{ #1 \right\}}

\newcommand{\norm}[1]{\left\| #1 \right\|}

\newcommand{\radonnikodym}{Radon--Nikod{\'y}m\xspace}

\newcommand{\normal}{\mathcal{N}}
\newcommand{\reals}{\mathbb{R}}

\newcommand{\expected}{\mathbb{E}}
\newcommand{\cauchy}{\textup{Cauchy}}

\newcommand*{\arXiv}[1]{\bgroup\color{blue}\href{http://arxiv.org/abs/#1}{arXiv:#1}\egroup}
\newcommand*{\doi}[1]{\bgroup\color{blue}\href{http://dx.doi.org/#1}{doi:#1}\egroup}
\newcommand*{\email}[1]{\bgroup\color{blue}\href{mailto:#1}{#1}\egroup}
\renewcommand*{\url}[1]{\bgroup\color{blue}\href{#1}{#1}\egroup}

\DeclareMathOperator*{\argmin}{arg\,min}
\DeclareMathOperator*{\arginf}{arg\,inf}

\newcommand{\ci}{\mathbin{\perp\!\!\!\perp}} 

\usepackage{tikz}
\usetikzlibrary{shapes,backgrounds,shadings,intersections}

\usepackage{pdflscape} 
\usepackage{afterpage}
\usepackage{capt-of}

\begin{document}

\title{Bayesian Probabilistic Numerical Methods}

\author{Jon Cockayne\footnote{University of Warwick, j.cockayne@warwick.ac.uk} \and Chris Oates\footnote{Newcastle University and Alan Turing Institute, chris.oates@ncl.ac.uk} \and Tim Sullivan\footnote{Free University of Berlin and Zuse Institute Berlin, sullivan@zib.de} \and Mark Girolami\footnote{Imperial College London and Alan Turing Institute, m.girolami@imperial.ac.uk}}

\date{\today}

\maketitle

\begin{abstract}
The emergent field of probabilistic numerics has thus far lacked clear statistical principals.
This paper establishes \emph{Bayesian} probabilistic numerical methods as those which can be cast as solutions to certain inverse problems within the Bayesian framework.
This allows us to establish general conditions under which Bayesian probabilistic numerical methods are well-defined, encompassing both non-linear and non-Gaussian models.
For general computation, a numerical approximation scheme is proposed and its asymptotic convergence established. 
The theoretical development is then extended to pipelines of computation, wherein probabilistic numerical methods are composed to solve more challenging numerical tasks.
The contribution highlights an important research frontier at the interface of numerical analysis and uncertainty quantification, with a challenging industrial application presented.
\end{abstract}

\section{Introduction}

Numerical computation underpins almost all of modern scientific and industrial research and development.
The impact of a finite computational budget is that problems whose solutions are high- or infinite-dimensional, such as the solution of differential equations, must be discretised in order to be solved.
The result is an approximation to the object of interest. 
The declining rate of processor improvement as physical limits are reached is in contrast to the surge in complexity of modern inference problems, and as a result the error incurred by discretisation is attracting increased interest \citep[e.g.][]{Capistran2013}.

The situation is epitomised in modern climate models, where use of single-precision arithmetic has been explored to permit finer temporal resolution.
However, when computing in single-precision, a detailed time discretisation can \emph{increase} total error, due to the increased number of single precision computations, and in practice some form of \emph{ad-hoc} trade-off is sought \citep{Harvey2015}.
It has been argued that statistical considerations can permit more principled error control strategies for such models \citep{Hennig:2015jf}.

Numerical methods are designed to mitigate discretisation errors of all forms \citep{Press2007}.
Nonetheless, the introduction of error is unavoidable and it is the role of the numerical analyst to provide control of this error \citep{Oberkampf:2013hi}.
The central theoretical results of numerical analysis have in general not been obtained through statistical considerations.
More recently, the connection of discretisation error to statistics was noted as far back as \citet{Henrici:1963vy,Hull1966}, who argued that discretisation error can be modelled using a series of independent random perturbations to standard numerical methods. 
However, numerical analysts have cast doubt on this approach, since discretisation error can be highly structured; see \cite{Kahan:1996uo} and \citet[Section 2.8]{Higham:2002hj}.
To address these objections, the field of \emph{probabilistic numerics} has emerged with the aim to properly quantify the uncertainty introduced through discretisation in numerical methods.

The foundations of probabilistic numerics were laid in the 1970s and 1980s, where an important shift in emphasis occurred from the descriptive statistical models of the 1960s to the use of formal inference modalities that generalise across classes of numerical tasks.
In a remarkable series of papers, \cite{Larkin1969,Larkin1970,Larkin1972,Kuelbs1972,Larkin1974,Larkin1979,Larkin1979a}, Mike Larkin presented now classical results in probabilistic numerics, in particular establishing the correspondence between Gaussian measures on Hilbert spaces and optimal numerical methods.
Re-discovered and re-emphasised on a number of occasions, the role for statisticians in this new outlook was clearly captured in \cite{Kadane1985}:
\begin{quote}
Statistics can be thought of as a set of tools used in making decisions and inferences in the face of uncertainty.
Algorithms typically operate in such an environment.
Perhaps then, statisticians might join the teams of scholars addressing algorithmic issues.
\end{quote}
The 1980s culminated in development of Bayesian optimisation methods \citep{Mockus1989, Torn1989},
as well as the relation of smoothing splines to Bayesian estimation \citep{Kimeldorf1970a, Diaconis1983}. 

The modern notion of a probabilistic numerical method (henceforth PNM) was described in \cite{Hennig:2015jf}; these are algorithms whose output is a distribution over an unknown, deterministic quantity of interest, such as the numerical value of an integral.
Recent research in this field includes PNMs for numerical linear algebra \citep{Hennig:2015hf, Bartels:2016eh}, numerical solution of ordinary differential equations \citep[ODEs;][]{Schober:2014wt, Kersting2016, Schober:2016uh, Conrad:2016gv, Chkrebtii:2013ux}, numerical solution of partial differential equations \citep[PDEs;][]{Owhadi2015,Cockayne:2016ts, Conrad:2016gv} and numerical integration \citep{OHagan1991,Briol2016}. 

\paragraph{Open Problems}

Despite numerous recent successes and achievements, there is currently no general statistical foundation for PNMs,
due to the infinite-dimensional nature of the problems being solved.
For instance, at present it is not clear under what conditions a PNM is well-defined, except for in the standard conjugate Gaussian framework considered in \citep{Larkin1972}.
This limits the extent to which domain-specific knowledge, such as boundedness of an integrand or monotonicity of a solution to a differential equation, can be encoded in PNMs.
In contrast, classical numerical methods often exploit such information to achieve substantial reduction in discretisation error.
For instance, finite element methods for solution of PDEs proceed based on a mesh that is designed to be more refined in areas of the domain where greater variation of the solution is anticipated \citep{Strang1973}.

Furthermore, although PNMs have been proposed for many standard numerical tasks (see Section \ref{sec:state_of_art}), the lack of common theoretical foundations makes comparison of these methods difficult.
Again taking PDEs as an example, \citet{Cockayne:2016ts} placed a probability distribution on the unknown solution of the PDE, whereas \citet{Conrad:2016gv} placed a probability distribution on the unknown discretisation error of a numerical method.
The uncertainty modelled in each case is fundamentally different, but at present there is no framework in which to articulate the relationship between the two approaches.
Furthermore, though PNMs are often reported as being ``Bayesian'' there is no clear definition of what this ought to entail.

A more profound consequence of the lack of common foundation occurs when we seek to compose multiple PNMs.
For example, multi-physics cardiac models involve coupled ODEs and PDEs which must each be discretised and approximately solved to estimate a clinical quantity of interest \citep{Niederer:2011gg}.
The composition of successive discretisations leads to non-trivial error propagation and accumulation that could be quantified, in a statistical sense, with PNMs. 
However, proper composition of multiple PNMs for solutions of ODEs and PDEs requires that these PNMs share common statistical foundations that ensure coherence of the overall statistical output.
These foundations remain to be established.

\paragraph{Contributions}

The main contribution of this paper is to establish rigorous foundations for PNMs:

The first contribution is to argue for an explicit definition of a ``Bayesian'' PNM.
Our framework generalises the seminal work of \cite{Larkin1972} and builds on the modern and popular mathematical framework of \cite{Stuart:2010gja}.
This illuminates subtle distinctions among existing methods and clarifies the sense in which non-Bayesian methods are approximations to Bayesian PNMs.

The second contribution is to establish when PNMs are well-defined outside of the conjugate Gaussian context.
For exploration of non-linear, non-Gaussian models, a numerical approximation scheme is developed and shown to asymptotically approach the posterior distribution of interest.
Our aim here is not to develop new or more computationally efficient PNMs, but to understand when such development can be well-defined.

The third contribution is to discuss pipelines of composed PNMs.
This is a critical area of development for probabilistic numerics;
in isolation, the error of a numerical method can often be studied and understood, but when composed into a pipeline the resulting error structure may be non-trivial and its analysis becomes more difficult.
The real power of probabilistic numerics lies in its application to pipelines of numerical methods, where the probabilistic formulation permits analysis of variance (ANOVA) to understand the contribution of each discretisation to the overall numerical error.
This paper introduces conditions under which a composition of PNMs can be considered to provide meaningful output, so that ANOVA can be justified.

\paragraph{Structure of the Paper}

In Section~\ref{sec:pnm} we argue for an explicit definition of Bayesian PNM and establish when such methods are well-defined.
Section~\ref{sec:decision} establishes connections to other related fields, in particular with relation to evaluating the performance of PNMs.
In Section~\ref{sec:sampling} we develop useful numerical approximations to the output of Bayesian PNMs.
Section~\ref{sec:pipelines} develops the theory of composition for multiple PNMs.
Finally, in Section~\ref{sec:results} we present applications of the techniques discussed in this paper.

All proofs can be found in either the Appendix or the Electronic Supplement.

\section{Probabilistic Numerical Methods}
\label{sec:pnm}

The aim of this section is to provide rigorous statistical foundations for PNMs.

\subsection{Notation}

For a measurable space $(\mathcal{X},\Sigma_{\mathcal{X}})$, the shorthand $\mathcal{P}_{\mathcal{X}}$ will be used to denote the set of all distributions on $(\mathcal{X}, \Sigma_{\mathcal{X}})$.
For $\mu,\mu' \in \mathcal{P}_{\mathcal{X}}$ we write $\mu \ll \mu'$ when $\mu$ is absolutely continuous with respect to $\mu$.
The notation $\delta(x)$ will be used to denote a Dirac measure on $x \in \mathcal{X}$, so that $\delta(x) \in \mathcal{P}_{\mathcal{X}}$.
Let $1[S]$ denote the indicator function of an event $S \in \Sigma_{\mathcal{X}}$.
For a measurable function $f:\mathcal{X} \to \reals$ and a distribution $\mu \in \mathcal{P}_\mathcal{X}$, we will on occasion use the notation $\mu(f) = \int f(x) \mu(\wrt x)$ and $\|f\|_\infty = \sup_{x \in \mathcal{X}} |f(x)|$.
The point-wise product of two functions $f$ and $g$ is denoted $f \cdot g$.
For a function or operator $T$, $T_{\#}$ denotes the associated push-forward operator\footnote{Recall that, for measurable $T:\mathcal{X} \to \mathcal{A}$, the pushforward $T_\# \mu$ of a distribution $\mu \in \mathcal{P}_\mathcal{X}$ is defined as $T_\# \mu(A) = \mu(T^{-1}(A))$ for all $A \in \Sigma_{\mathcal{A}}$.} that acts on measures on the domain of $T$.
Let $\ci$ denote conditional independence.
The subset $\ell^p \subset \mathbb{R}^\infty$ is defined to consist of sequences $(u_i)$ for which $\sum_{i=1}^\infty |u_i|^p$ is convergent.
$C(0,1)$ will be used to denote the set of continuous functions on $(0,1)$.

\subsection{Definition of a PNM} \label{subsec:definition}

To first build intuition, consider numerical approximation of the Lebesgue integral 
\begin{equation*}
	\int x(t) \nu(\wrt t)
\end{equation*}
for some integrable function $x \colon D \to \reals$, with respect to a measure $\nu$ on $D$.
Here we may directly interrogate the integrand $x(t)$ at any $t \in D$, but unless $D$ is finite we cannot evaluate $x$ at all $t \in D$ with a finite computational budget.
Nonetheless, there are many algorithms for approximation of this integral based on information $\{x(t_i)\}_{i=1}^n$ at some collection of locations $\{t_i\}_{i=1}^n$.

To see the abstract structure of this problem, assume the state variable $x$ exists in a measurable space $(\mathcal{X}, \Sigma_{\mathcal{X}})$.
Information about $x$ is provided through an \emph{information operator} $A \colon \mathcal{X} \to \mathcal{A}$ whose range is a measurable space $(\mathcal{A}, \Sigma_\mathcal{A})$.
Thus, for the Lebesgue integration problem, the information operator is
\begin{equation}
	\label{eq:system_finite}
	A(x) = \begin{bmatrix}
		x(t_1) \\
		\vdots \\
		x(t_n)
	\end{bmatrix} = a \in \mathcal{A} .
\end{equation}
The space $\mathcal{X}$, in this case a space of functions, can be high- or infinite-dimensional, but the space $\mathcal{A}$ of information is assumed to be finite-dimensional in accordance with our finite computational budget. 
In this paper we make explicit a quantity of interest (QoI) $Q(x)$, defined by a map $Q \colon \mathcal{X} \rightarrow \mathcal{Q}$ into a measurable space $(\mathcal{Q}, \Sigma_{\mathcal{Q}})$. 
This captures that $x$ itself may not be the object of interest for the numerical problem; for the Lebesgue integration illustration, the QoI is not $x$ itself but $Q(x) = \int x(t) \nu(\wrt t)$.

The standard approach to such computational problems is to construct an algorithm which, when applied, produces some approximation $\hat{q}(a)$ of $Q(x)$ based on the information $a$, whose theoretical convergence order can be studied. 
A successful algorithm will often tailor the information operator $A$ to the QoI $Q$. 
For example, classical Gaussian cubature specifies \emph{sigma points} $\{t_i^*\}_{i=1}^n$ at which the integrand must be evaluated, based on exact integration of certain polynomial test functions. 

The probabilistic numerical approach, instead, begins with the introduction of a random variable $X$ on $(\mathcal{X}, \Sigma_\mathcal{X})$.
The true state $X = x$ is fixed but unknown; the randomness is used an abstract device used to represent epistemic uncertainty about $x$ prior to evolution of the information operator \citep{Hennig:2015jf}.
This is now formalised:

\begin{definition}[Belief Distribution]
An element $\mu \in \mathcal{P}_{\mathcal{X}}$ is a \emph{belief distribution}\footnote{Two remarks are in order: First, we have avoided the use of ``prior'' as this abstract framework encompasses both Bayesian and non-Bayesian PNMs (to be defined). Second, the use of ``belief'' differs to the set-valued \emph{belief functions} in Dempster–Shafer theory, which do not require that $\mu(E) + \mu(E^\text{c}) = 1$ \citep{Shafer1976}.} for $x$ if it carries the formal semantics of belief about the true, unknown state variable $x$.
\end{definition}

\noindent
Thus we may consider $\mu$ to be the law of $X$.
The construction of an appropriate belief distribution $\mu$ for a specific numerical task is not the focus of this research and has been considered in detail in previous work; see the Electronic Supplement for an overview of this material.
Rather we consider the problem of how one \emph{updates} the belief distribution $\mu$ in response to the information $A(x) = a$ obtained about the unknown $x$.
Generic approaches to update belief distributions, which generalise Bayesian inference beyond the unique update demanded in Bayes theorem, were formalised in \cite{Bissiri2016,Carvalho2017}. 

\begin{definition}[Probabilistic Numerical Method]
	Let $(\mathcal{X},\Sigma_{\mathcal{X}})$, $(\mathcal{A}, \Sigma_{\mathcal{A}})$ and $(\mathcal{Q},\Sigma_{\mathcal{Q}})$ be measurable spaces and let $A \colon \mathcal{X} \to \mathcal{A}$, $Q \colon \mathcal{X} \rightarrow \mathcal{Q}$ and $B \colon \mathcal{P}_\mathcal{X} \times \mathcal{A} \to \mathcal{P}_\mathcal{Q}$ where $A$ and $Q$ are measurable functions.
	The pair $M = (A , B)$ is called a \emph{probabilistic numerical method} for estimation of a quantity of interest $Q$.
	The map $A$ is called an \emph{information operator}, and the map $B$ is called a \emph{belief update operator}.
\end{definition}

The output of a PNM is a distribution $B(\mu,a) \in \mathcal{P}_{\mathcal{Q}}$.
This holds the formal status of a belief distribution for the value of $Q(x)$, based on both the initial belief $\mu$ about the value of $x$ and the information $a$ that are input to the PNM.

An objection sometimes raised to this construction is that $x$ itself is not random.
We emphasise that this work does not propose that $x$ should be considered as such; the random variable $X$ is a formal statistical device used to represent epistemic uncertainty \citep{Kadane2011,Lindley2014}.
Thus, there is no distinction from traditional statistics, in which $x$ represents a fixed but unknown parameter and $X$ encodes epistemic uncertainty about this parameter. 

Before presenting specific instances of this general framework, we comment on the potential analogy between $A$ and the likelihood function, and between $B$ and Bayes' theorem.
Whilst intuitively correct, the mathematical developments in this paper are not well-suited to these terms;
in Section~\ref{sec:conditioning} we show that Bayes formula is not well-defined, as the posterior distribution is not absolutely continuous with respect to the prior.

To strengthen intuition we now give specific examples of established PNMs:

\begin{example}[Probabilistic Integration] \label{ex:PI}

Consider the numerical integration problem earlier discussed.
Take $D \subseteq \reals^d$, $\mathcal{X}$ a separable Banach space of real-valued functions on $D$, and $\Sigma_\mathcal{X}$ the Borel $\sigma$-algebra for $\mathcal{X}$.
The space $(\mathcal{X},\Sigma_{\mathcal{X}})$ is endowed with a Gaussian belief distribution $\mu \in \mathcal{P}_{\mathcal{X}}$.
Given information $A(x) = a$, define $\mu^a$ to be the restriction of $\mu$ to those functions which interpolate $x$ at the points $\{t_i\}_{i=1}^n$; that $\mu^a$ is again Gaussian follows from linearity of the information operator \citep[see][for details]{Bogachev1998b}.
The QoI $Q$ remains $Q(x) = \int x(t) \nu(\wrt t)$.

This problem was first considered by \cite{Larkin1972}. The belief update operator proposed therein, and later considered in \cite{Diaconis1988,OHagan1991} and others, was $B(\mu,a) = Q_{\#} \mu^a$.
Since Gaussians are closed under linear projection, the PNM output $B(\mu, a)$ is a univariate Gaussian whose mean and variance can be expressed in closed-form for certain choices of Gaussian covariance function and reference measure $\nu$ on $D$.
Specifically, if $\mu$ has mean function $m \colon \mathcal{X} \rightarrow \mathbb{R}$ and covariance function $k:\mathcal{X} \times \mathcal{X} \rightarrow \mathbb{R}$, then
\begin{equation}
	B(\mu,a) = \text{N}(z^\top \mathrm{K}^{-1} (a-\bar{m}) , z_0 - z^\top \mathrm{K}^{-1} z)
\end{equation}
where $\bar{m}, z \in \mathbb{R}^n$ are defined as $\bar{m}_i = m(t_i)$, $z_i = \int k(t,t_i) \nu(\mathrm{d}t)$, $\mathrm{K} \in \mathbb{R}^{n \times n}$ is defined as $\mathrm{K}_{i,j} = k(t_i,t_j)$ and $z_0 = \iint k(t,t') (\nu \times \nu)\mathrm{d} (t \times t') \in \mathbb{R}$. 
This method was extensively studied in \citet{Briol2016}, who provided a listing of $(\nu,k)$ combinations for which $z$ and $z_0$ possess a closed-form.

An interesting fact is that the mean of $B(\mu, a)$ coincides with classical cubature rules for different choices of $\mu$ and $A$ \citep{Diaconis1988,Saerkkae2015}.
In Section~\ref{sec:decision} we will show that this is a typical feature of PNMs.
The crucial distinction between PNMs and classical numerical methods is the \emph{distributional} nature of $B(\mu,a)$, which carries the formal semantics of belief about the QoI.
The full distribution $B(\mu,a)$ was examined in \cite{Briol2016}, who established contraction to the exact value of the integral under certain smoothness conditions on the Gaussian covariance function and on the integrand.
See also \cite{Kanagawa2016,Karvonen2017}.

\end{example}

\begin{example}[Probabilistic Meshless Method] \label{ex:PMM}

As a canonical example of a PDE, take the following elliptic problem with Dirichlet boundary conditions
\begin{align*}
	-\nabla \cdot \left( \kappa \nabla x \right) &= f && \text{in } D\\
	x &= g && \text{on } \partial D \numberthis \label{eq:elliptic_pde}
\end{align*}
where we assume $D \subset \reals^d$ and $\kappa \colon D \rightarrow \mathbb{R}^{d \times d}$ is a known coefficient. 
Let $\mathcal{X}$ be a separable Banach space of appropriately differentiable real-valued functions and take $\Sigma_\mathcal{X}$ to be the Borel $\sigma$-algebra for $\mathcal{X}$.
In contrast to the first illustration, the QoI here is $Q(x) = x$, as the goal is to make inferences about the solution of the PDE itself. 

Such problems were considered in \cite{Cockayne:2016ts} wherein $\mu$ was restricted to be a Gaussian distribution on $\mathcal{X}$. 
The information operator was constructed by choosing finite sets of locations 
$T_1 = \set{t_{1,1}, \dots, t_{1,n_1}} \subset D$ 
and 
$T_2 = \set{t_{2,1}, \dots, t_{2,n_2}} \subset \partial D$ at which the system defined in Eq.~\eqref{eq:elliptic_pde} was evaluated, so that
\begin{equation*}
	A (x) = \begin{bmatrix}
		-\nabla \cdot \left( \kappa(t_{1,1}) \nabla x(t_{1,1})\right) \\
		\vdots \\
		-\nabla \cdot \left( \kappa(t_{1,n_1}) \nabla x(t_{1,n_1})\right) \\
		x(t_{2,1}) \\
		\vdots \\
		x(t_{2,n_2})
	\end{bmatrix}
	\qquad
	a = \begin{bmatrix}
		f(t_{1,1}) \\
		\vdots \\
		f(t_{1,n_1}) \\
		g(t_{2,1}) \\
		\vdots \\
		g(t_{2,n_2})
	\end{bmatrix} \;.
\end{equation*}
The belief update operator was chosen to be $B(\mu,a) = \mu^a$, where $\mu^a$ is the restriction of $\mu$ to those functions for which $A(x) = a$ is satisfied. 
In the setting of a linear system of PDEs such as that in Eq.~\eqref{eq:elliptic_pde}, the distribution $B(\mu,a)$ is again Gaussian \citep{Bogachev1998b}.
Full details are provided in \cite{Cockayne:2016ts}.

As in the previous example, we note that the mean of $B(\mu,a)$ coincides with the numerical solution to the PDE provided by a classical method \citep[the \emph{symmetric collocation} method;][]{Fasshauer1999}.
The full distribution $B(\mu,a)$ provides uncertainty quantification for the unknown exact solution and can again be shown to contract to the exact solution under certain smoothness conditions \citep{Cockayne:2016ts}.
This method was further analysed for a specific choice of covariance operator in the belief distribution $\mu$, in an impressive contribution from \cite{Owhadi2015a}.

\end{example}

\subsubsection{Classical Numerical Methods}

Standard numerical methods fit into the above framework, as can be seen by taking 
\begin{equation}
	B(\mu , a)  =  \delta \circ b(a) \label{eq:classic_NM}
\end{equation}
independent of the distribution $\mu$, where a function $b \colon \mathcal{A} \rightarrow \mathcal{Q}$ gives the output of some classical numerical method for solving the problem of interest.
Here $\delta \colon \mathcal{Q} \rightarrow \mathcal{P}_{\mathcal{Q}}$ maps $b(a) \in \mathcal{Q}$ to a Dirac measure centred on $b(a)$.
Thus, information in $a \in \mathcal{A}$ is used to construct a point estimate $b(a) \in \mathcal{Q}$ for the QoI.

The formal language of probabilities is not used in classical numerical analysis to describe numerical error. 
However, in many cases the classical and probabilistic analyses are mathematically equivalent. 
For instance, there is an equivalence between the standard deviation of $B(\mu,a)$ for probabilistic integration and the worst-case error for numerical cubature rules from numerical analysis \citep{Novak2010}.
The explanation for this phenomenon will be given in Section \ref{sec:decision}.

\subsection{Bayesian PNMs}

Having defined a PNM, we now state the central definition of this paper, that is of a \emph{Bayesian} PNM.
Define $\mu^a$ to be the conditional distribution of the random variable $X$, given the event $A(X) = a$.
For now we assume that this can be defined without ambiguity and reserve a more technical treatment of conditional probabilities for Section~\ref{sec:conditioning}.

In this work we followed \cite{Larkin1972} and cast the problem of determining $x$ in Eq.~\eqref{eq:system_finite} as a problem of Bayesian inversion, a framework now popular in applied mathematics and uncertainty quantification research \citep{Stuart:2010gja}.
However, in a standard Bayesian inverse problem the observed quantity $a$ is assumed to be corrupted with measurement error, which is described by a ``likelihood''.
This leads, under mild assumptions, to general versions of Bayes' theorem \citep[see][Section 2.2]{Stuart:2010gja}

For PNM, however, the information is \emph{not} corrupted with measurement error.
As a result, the support of the likelihood is a null set under the prior, making the standard approaches to such problems, including Bayes' theorem, ill-defined outside of the conjugate Gaussian case when unknowns are infinite-dimensional. 
This necessitates a new definition:

\begin{definition}[Bayesian Probabilistic Numerical Method] \label{def:bayesian_pnm}
	A probabilistic numerical method $M = (A,B)$ is said to be \emph{Bayesian}\footnote{The use of ``Bayesian'' contrasts with \cite{Bissiri2016}, for whom all belief update operators represent Bayesian learning algorithms to some greater or lesser extent. An alternative term could be ``lossless'', since all the information in $a$ is conditioned upon in $\mu^a$.} for a quantity of interest $Q$ if, for all $\mu \in \mathcal{P}_{\mathcal{X}}$, the output 
	\begin{equation*}
		B(\mu , a) = Q_{\#} \mu^a, 
		\quad  \text{ for } A_{\#} \mu \text{-almost-all } a \in \mathcal{A} .
	\end{equation*}
\end{definition}
That is, a PNM is Bayesian if the output of the PNM is the push-forward of the conditional distribution $\mu^a$ through $Q$.
This definition is familiar from the examples in Section \ref{subsec:definition}, which are both examples of Bayesian PNMs.

For Bayesian PNMs we adopt the traditional terminology in which $\mu$ is the \emph{prior} for $x$ and the output $Q_{\#} \mu^a$ the \emph{posterior} for $Q(x)$.
Note that, for fixed $A$ and $\mu$, the Bayesian choice of belief update operator $B$ (if it exists) is uniquely defined.

It is emphasised that the class of Bayesian PNMs is a subclass of all PNMs; examples of non-Bayesian PNMs are provided in Section~\ref{sec:state_of_art}.
Our analysis is focussed on Bayesian PNMs due to their appealing Bayesian interpretation and ease of generalisation to pipelines of computation in Section~\ref{sec:pipelines}.
For non-Bayesian PNMs, careful definition and analysis of the belief update operator is necessary to enable proper interpretation of the uncertainty quantification being provided.
In particular, the analysis of non-Bayesian PNMs may present considerable challenges in the context of computational pipelines, whereas for Bayesian PNMs this is shown in Section~\ref{sec:pipelines} to be straight-forward.

\subsection{Model Evidence}

A cornerstone of the Bayesian framework is the model evidence, or marginal likelihood \citep{MacKay1992}.
Let $\mathcal{A} \subseteq \mathbb{R}^n$ be equipped with the Lebesgue reference measure $\lambda$, such that $A_{\#}\mu$ admits a density $p_A = \wrt A_{\#} \mu / \wrt \lambda$.
Then the \emph{model evidence} $p_A(a)$, based on the information that $A(x) = a$, can be used as the basis for Bayesian model comparison.
In particular, two prior distributions $\mu$, $\tilde{\mu}$, can be compared through the Bayes factor 
\begin{align}
\mathrm{BF} := \frac{\tilde{p}_A(a)}{p_A(a)} &= \frac{\wrt A_{\#}\tilde{\mu}}{\wrt A_{\#}\mu}(a), 
\end{align}
where $\tilde{p}_A = \wrt A_{\#} \tilde{\mu} / \wrt \lambda$.
Here the second expression is independent of the choice of reference measure $\lambda$ and is thus valid for general $\mathcal{A}$.
The model evidence has been explored in connection with the design of Bayesian PNM.
For the integration and PDE examples \ref{ex:PI} and \ref{ex:PMM}, the model evidence has a closed form and was investigated in \cite{Briol2016,Cockayne:2016ts}.
In Section \ref{sec:results} we investigate the model evidence in the context of non-linear ODEs and PDEs for which it must be approximated.

\subsection{The Disintegration Theorem} \label{sec:conditioning}

The purpose of this section is to formalise $\mu^a$ and to determine conditions under which $\mu^a$ exists and is well-defined.
From Definition \ref{def:bayesian_pnm}, the output of a Bayesian PNM is $B(\mu, a) = Q_\# \mu^a$.
If $\mu^a$ exists, the pushforward $Q_\# \mu^a$ exists as $Q$ is assumed to be measurable; thus, in this section, we focus on the rigorous definition of $\mu^a$.

Unlike many problems of Bayesian inversion, proceeding by an analogue of Bayes' theorem is not possible.
Let $\mathcal{X}^a = \set{x \in \mathcal{X} : A (x) = a}$. Then we observe that, if it is  measurable, $\mathcal{X}^a$ may be a set of zero measure under $\mu$. Standard techniques for infinite-dimensional Bayesian inversion rely on constructing a posterior distribution based on its \radonnikodym derivative with respect to the prior \citep{Stuart:2010gja}. However, when $\mu^a \centernot\ll \mu$ no \radonnikodym derivative exists and we must turn to other approaches to establish when a Bayesian PNM is well-defined.

Conditioning on null sets is technical and was formalised in the celebrated construction of measure-theoretic probability by \cite{Kolmogorov:1933fe}. 
The central challenge is to establish uniqueness of conditional probabilities.
For this work we exploit the \emph{disintegration theorem} to ensure our constructions are well-defined.
The definition below is due to \citet[][p.78]{Dellacherie1978}, and a statistical introduction to disintegration can be found in \cite{Chang1997}.

\begin{definition}[Disintegration]
For $\mu \in \mathcal{P}_{\mathcal{X}}$, a collection $\{\mu^a\}_{a \in \mathcal{A}} \subset \mathcal{P}_{\mathcal{X}}$ is a \emph{disintegration} of $\mu$ with respect to the (measurable) map $A \colon \mathcal{X} \rightarrow \mathcal{A}$ if:
\begin{enumerate}
\item[1] (Concentration:) $\mu^a(\mathcal{X} \setminus \mathcal{X}^a) = 0$ for $A_\# \mu$-almost all $a \in \mathcal{A}$;
\end{enumerate}
and for each measurable $f \colon \mathcal{X} \rightarrow [0,\infty)$ it holds that
\begin{enumerate}
\item[2] (Measurability:) $a \mapsto \mu^a(f)$ is measurable;
\item[3] (Conditioning:) $\mu(f) = \int \mu^a(f) A_\# \mu(\mathrm{d}a)$.
\end{enumerate}
\end{definition}

The concept of disintegration extends the usual concept of conditioning of random variables to the case where $\mathcal{X}^a$ is a null set, in a way closely related to regular conditional distributions \citep{Kolmogorov:1933fe}.
Existence of disintegrations is guaranteed under general weak conditions:
\begin{theorem}[Disintegration Theorem; Thm. 1 of \cite{Chang1997}] \label{thm:Chang}
	Let $\mathcal{X}$ be a metric space, $\Sigma_{\mathcal{X}}$ be the Borel $\sigma$-algebra and $\mu\in \mathcal{P}_\mathcal{X}$ be Radon.
	Let $\Sigma_{\mathcal{A}}$ be countably generated and contain all singletons $\{a\}$ for $a \in \mathcal{A}$. 
	Then there exists a disintegration $\{\mu^a\}_{a \in \mathcal{A}}$ of $\mu$ with respect to $A$.
	Moreover, if $\{\nu^a\}_{a \in \mathcal{A}}$ is another such disintegration, then $\{a \in \mathcal{A} : \mu^a \neq \nu^a\}$ is a $A_\# \mu$ null set.
\end{theorem}

The requirement that $\mu$ is Radon is weak and is implied when $\mathcal{X}$ is a \emph{Radon space}, which encompasses, for example, separable complete metric spaces. 
The requirement that $\Sigma_{\mathcal{A}}$ is countably generated is also weak and includes the standard case where $\mathcal{A} = \mathbb{R}^n$ with the Borel $\sigma$-algebra.
From Theorem \ref{thm:Chang} it follows that $\{\mu^a\}_{a \in \mathcal{A}}$ exists and is essentially unique for all of the examples considered in this paper.
Thus, under mild conditions, we have established that Bayesian PNMs are well-defined, in that an essentially unique disintegration $\{\mu^a\}_{a \in \mathcal{A}}$ exists.
It is noted that a variational definition of $\mu^a$ has been posited as an alternative approach, for when the existence of a disintegration is difficult to establish \citep[p3 of][]{Trillos2017}.

\subsection{Prior Construction} \label{sec:rcp_prior}

The Gaussian distribution is popular as a prior in the PNM literature for its tractability, both in the fact that finite-dimensional distributions take a closed-form and that an explicit conditioning formula exists. 
More general priors, such as Besov priors \citep{Dashti:2012el} and Cauchy priors \citep{Sullivan:2016wp} are less easily accessed. 
In this section we summarise a common construction for these prior distributions, designed to ensure that a disintegration will exist.

Let $\set{\phi_i}_{i=0}^\infty$ denote an orthogonal Schauder basis for $\mathcal{X}$, assumed to be a separable Banach space in this section.
Then any $x \in \mathcal{X}$ can be represented through an expansion
\begin{equation}
	\label{eq:generic_random_field}
	x = x_0 + \sum_{i=0}^\infty u_i \phi_i
\end{equation}
for some fixed element $x_0 \in \mathcal{X}$ and a sequence $u \in \reals^\infty$.
Construction of measures $\mu \in \mathcal{P}_{\mathcal{X}}$ is then reduced to construction of almost-surely convergent measures on $\reals^\infty$ and studying the pushforward of such measures into $\mathcal{X}$.
In particular, this will ensure that $\mu \in \mathcal{P}_{\mathcal{X}}$ is Radon (as $\mathcal{X}$ is a separable complete metric space), a key requirement for existence of a disintegration $\{\mu^a\}_{a \in \mathcal{A}}$.

To this end it is common to split $u$ into a stochastic and deterministic component;
let $\xi \in \mathbb{R}^\infty$ represent an i.i.d sequence of random variables, and $\gamma \in \ell^p$ for some $p \in (1,\infty)$.
Then with $u_i = \gamma_i \xi_i$, for the prior distribution to be well-posed we require that almost-surely $u \in \ell^1$. 
Different choices of $(\xi, \gamma)$ give rise to different distributions on $\mathcal{X}$.
For instance, $\xi_i \sim \text{Uniform}(-1, 1)$, $\gamma \in \ell^1$ is termed a \emph{uniform} prior and $\xi_i \sim \normal(0, 1)$ gives a \emph{Gaussian} prior, where $\gamma$ determines the regularity of the covariance operator $\mathcal{C}$ \citep{Bogachev1998b}. The choice of $\xi_i \sim \cauchy(0, 1)$ gives a \emph{Cauchy} prior in the sense of \cite{Sullivan:2016wp}; here we require $\gamma \in \ell^{1} \cap \ell \log \ell$ for $\mathcal{X}$ a separable Banach space, or $\gamma \in \ell^2$ for when $\mathcal{X}$ is a Hilbert space.

A range of prior specifications will be explored in Section \ref{sec:results}, including non-Gaussian prior distributions for numerical solution of nonlinear ODEs.

\subsubsection{Dichotomy of Existing PNMs} \label{sec:state_of_art}

This section concludes with an overview of existing PNMs with respect to our definition of a Bayesian PNM.
This serves to clarify some subtle distinctions in existing literature, as well as to highlight the generality of our framework.
To maintain brevity we have summarised our findings in Table \ref{table:comparison}.

\afterpage{%
\clearpage
\begin{landscape}
\footnotesize
\begin{tabular}{|p{2cm}|p{2cm}|p{4cm}|p{7cm}|p{7cm}|} \hline
\textbf{Method} & \textbf{QoI} $Q(x)$ & \textbf{Information} $A(x)$ & \textbf{Non- (or Approximate) Bayesian PNMs} & \textbf{Bayesian PNMs} \\ \hline\hline
Integrator & $\int x(t) \nu(\mathrm{d} t)$ & $\{x(t_i)\}_{i=1}^n$ & \cite{Osborne2012,Osborne2012a,Gunter2014} & Bayesian Quadrature \citep{Larkin1974,Diaconis1988,OHagan1991} \\
 & $\int f(t) x(\mathrm{d}t)$ & $\{t_i\}_{i=1}^n$ s.t. $t_i \sim x$ & \cite{Kong2003,Tan2004,Kong2007} & \\
 & $\int x_1(t) x_2(\mathrm{d}t)$ & $\{(t_i,x_1(t_i))\}_{i=1}^n$ s.t. $t_i \sim x_2$ & & \cite{Oates2016a} \\ \hline\hline
Optimiser & $\arg\min x(t)$ & $\{x(t_i)\}_{i=1}^n$ & & Bayesian Optimisation \citep{Mockus1989} \\
 & & $\{\nabla x(t_i)\}_{i=1}^n$ & & \cite{Hennig2013} \\
 & & $\{(x(t_i), \nabla x(t_i)\}_{i=1}^n$ & & Probabilistic Line Search \citep{Mahsereci2015} \\
 & & $\{\mathbb{I}[t_{\min} < t_i]\}_{i=1}^n$ & & Probabilistic Bisection Algorithm \citep{Horstein1963} \\
 & & $\{\mathbb{I}[t_{\min} < t_i] \text{ + error} \}_{i=1}^n$ & \cite{Waeber2013} & \\ \hline\hline 
Linear Solver & $x^{-1}b$ & $\{xt_i\}_{i=1}^n$ & & Probabilistic Linear Solvers \citep{Hennig:2015hf,Bartels:2016eh} \\ \hline \hline
ODE Solver & $x$ & $\{\nabla x(t_i)\}_{i=1}^n$ & \citep{Skilling:1992fb} & \\
& & & Filtering Methods for IVPs \citep{Schober:2014wt,Chkrebtii:2013ux,Kersting2016,Teymur2016,Schober:2016uh} & \\
& & & Finite Difference Methods \citep{John2017} & \\
 & & $\nabla x$ + rounding error & \cite{Hull1966,Mosbach:2009kq} & \\ 
 & $x(t_{\text{end}})$ & $\{\nabla x(t_i)\}_{i=1}^n$ & Stochastic Euler \citep{Krebs2016} & \\ \hline\hline
PDE Solver & $x$ & $\{Dx(t_i)\}_{i=1}^n$ & \cite{Chkrebtii:2013ux,Raissi:2017} & Probabilistic Meshless Methods \citep{Owhadi2015,Owhadi2015a,Cockayne:2016ts,Raissi2016} \\
 & & $Dx$ + discretisation error & \cite{Conrad:2016gv} & \\ \hline
\end{tabular}
\captionof{table}{Comparison of several existing Probabilistic Numerical Methods (PNMs).}
\label{table:comparison}
\end{landscape}
\clearpage
}

\section{Decision-Theoretic Treatment} \label{sec:decision}

Next we assess the performance of PNMs from a decision-theoretic perspective \citep{Berger1985} and explore connections to average-case analysis of classical numerical methods \citep{Ritter2000}. Note that the treatment here is agnostic to whether the PNM in question is Bayesian, and also encompasses classical numerical methods.
Throughout, the existence of a disintegration $\{\mu^a\}_{a \in \mathcal{A}}$ will be assumed.

\subsection{Loss and Risk} \label{subsec:lossrisk}
Consider a generic loss function $L \colon \mathcal{Q} \times \mathcal{Q} \rightarrow \mathbb{R}$
where $L(q^\dagger, q)$ describes the loss incurred when the true QoI $q^\dagger = Q(x)$ is estimated with $q \in \mathcal{Q}$. 
Integrability of $L$ is assumed.

The belief update operator $B$ returns a distribution over $\mathcal{Q}$ which can be cast as a randomised decision rule for estimation of $q^\dagger$.
For randomised decision rules, the \emph{risk function} $r \colon \mathcal{Q} \times \mathcal{P}_{\mathcal{Q}} \rightarrow \mathbb{R}$ is defined as
\begin{equation*}
	r(q^\dagger , \nu) = \int L(q^\dagger, q) \nu(\wrt q) \;.
\end{equation*}
The \emph{average risk} of the PNM $M=(A,B)$ with respect to $\mu \in \mathcal{P}_{\mathcal{X}}$ is defined as
\begin{equation} \label{eq:bayes_risk}
	R(\mu, M) = \int r (Q(x) , B(\mu, A(x))) \mu(\wrt x)  .
\end{equation}
Here a state $x \sim \mu$ is drawn at random and the risk of the PNM output $B(\mu,A(x))$ is computed.
We follow the convention of terming $R(\mu, M)$ the \emph{Bayes risk} of the PNM, though the usual objection that a frequentist expectation enters into the definition of the Bayes risk could be raised.

Next, we consider a sequence $A^{(n)}$ of information operators indexed such that $A^{(n)}(x)$ is $n$-dimensional (i.e. $n$ pieces of information are provided about $x$).

\begin{definition}[Contraction]
	A sequence $M^{(n)} = (A^{(n)} , B^{(n)})$ of PNMs is said to \emph{contract} at a \emph{rate} $r_n$ under a belief distribution $\mu$ if $R(\mu,M^{(n)}) = O(r_n)$.
\end{definition}

This definition allows for comparison of classical and probabilistic numerical methods \citep{Kadane:1983ww,Diaconis1988}.
In each case an important goal is to determine methods $M^{(n)}$ that contract as quickly as possible for a given distribution $\mu$ that defines the Bayes risk.
This is the approach taken in average-case analysis \citep[ACA;][]{Ritter2000} and will be discussed in Section~\ref{sec:ACA}.
For Examples \ref{ex:PI} and \ref{ex:PMM} of Bayesian PNMs, \cite{Briol2016} and \cite{Cockayne:2016ts} established rates of contraction for particular prior distributions $\mu$; we refer the reader to those papers for details.

\subsection{Bayes Decision Rules}

A (possibly randomised) decision rule is said to be a \emph{Bayes rule} if it achieves the minimum Bayes risk among all decision rules.
In the context of (not necessarily Bayesian) PNMs, let $M = (A, B)$ and let 
\begin{equation*}
	\mathfrak{B}(A) = \set{B : R(\mu, (A, B)) = \inf_{B'} R(\mu, (A, B'))} .
\end{equation*}
That is, for fixed $A$, $\mathfrak{B}(A)$ is the set of all belief update operators that achieve minimum Bayes risk.

This raises the natural question of which belief update operators yield Bayes rules. Although the definition of a Bayes rule applies generically to both probabilistic and deterministic numerical methods, it can be shown\footnote{The proof is included in the Electronic Supplement.} that if $\mathfrak{B}(A)$ is non-empty, then there exists a $B \in \mathfrak{B}(A)$ which takes the form of a classical numerical method, as expressed in Eq.~\eqref{eq:classic_NM}.
Thus in general, \emph{Bayesian PNMs do not constitute Bayes rules}, as the extra uncertainty inflates the Bayes risk, so that such methods are not optimal.

Nonetheless, there is a natural connection between Bayesian PNMs and Bayes rules, as exposed in \cite{Kadane:1983ww}:

\begin{theorem} \label{thm:bayes_mean}
Let $M=(A, B)$ be a Bayesian probabilistic numerical method for the QoI $Q$.
Let $(\mathcal{Q},\langle \cdot , \cdot \rangle_{\mathcal{Q}})$ be an inner-product space and let the loss function $L$ have the form $L(q^\dagger, q) = \|q^\dagger - q\|_\mathcal{Q}^2$, where $\|\cdot\|_{\mathcal{Q}}$ is the norm induced by the inner product.
Then the decision rule that returns the mean of the distribution $B(\mu, a)$ is a Bayes rule for estimation of $q^\dagger$.
\end{theorem}

This well-known fact from Bayesian decision theory\footnote{This is the fact that the Bayes act is the posterior mean under squared-error loss \citep{Berger1985}.} is interesting in light of recent research in constructing PNMs whose mean functions correspond to classical numerical methods \citep{Schober:2014wt,Hennig:2015hf,Saerkkae2015,Teymur2016,Schober:2016uh}.
Theorem \ref{thm:bayes_mean} explains the results in Examples \ref{ex:PI} and \ref{ex:PMM}, in which both instances of Bayesian PNMs were demonstrated to be centred on an established classical method. 
 
\subsection{Optimal Information}

The previous section considered selection of the belief update operator $B$, but not of the information operator $A$.
The choice of $A$ determines the Bayes risk for a PNM, which leads to a problem of experimental design to minimise that risk.

The theoretical study of optimal information is the focus of the information complexity literature \citep{Traub1988,Novak2010}, while other fields such as quasi-Monte Carlo \citep[QMC, ][]{Dick2010} attempt to develop asymptotically optimal information operators for specific numerical tasks, such as the choice of evaluation points for numerical approximation of integrals in the case of QMC.
Here we characterise optimal information for Bayesian PNMs.

Consider the choice of $A$ from a fixed subset $\Lambda$ of the set of all possible information operators.
To build intuition, for the task of numerical integration, $\Lambda$ could represent all possible choices of locations $\{t_i\}_{i=1}^n$ where the integrand is evaluated.
For Bayesian PNM, one can ask for optimal information:
\begin{equation*}
	A_{\mu} \in \arginf_{A \in \Lambda} \left\{ R(\mu,M) \; \text{s.t.} \; M = (A,B), \; B = Q_{\#} \mu^A \right\}
\end{equation*}
where we have made explicit the fact that the optimal information depends on the choice of prior $\mu$.
Next we characterise $A_\mu$, while an explicit example of optimal information for a Bayesian PNM is detailed in Example \ref{ex:optimalinfo}.

\subsection{Connection to Average Case Analysis} \label{sec:ACA}

The decision theoretic framework in Section~\ref{subsec:lossrisk} is closely related to average-case analysis (ACA) of classical numerical methods \citep{Ritter2000}.
In ACA the performance of a classical numerical method $b \colon \mathcal{A} \rightarrow \mathcal{Q}$ is studied in terms of the Bayes risk $R(\mu,M)$ given in Eq.~\eqref{eq:bayes_risk}, for the PNM $M = (A,B)$ with belief operator $B(\mu, a) = \delta \circ b(a)$ as in Eq.~\eqref{eq:classic_NM}.
ACA is concerned with the study of optimal information:
\begin{equation*}
A_{\mu}^{*} \in \arginf_{A \in \Lambda} \left\{ \inf_b R(\mu,M) \; \text{s.t.} \; M = (A,B), \; B = \delta \circ b \right\} .
\end{equation*}
In general there is no reason to expect $A_\mu$ and $A_\mu^*$ to coincide, since Bayesian PNM are not Bayes rules\footnote{The distribution $Q_{\#}\mu^a$ will in general not be supported on the set of Bayes acts.}.
Indeed, an explicit example where $A_\mu \neq A_\mu^*$ is presented in Appendix~\ref{sec:optimal_counterexample}.
However, we can establish sufficient conditions under which optimal information for a Bayesian PNM is the same as optimal information for ACA:

\begin{theorem} \label{thm:optimal_information}
Let $(\mathcal{Q},\langle \cdot , \cdot \rangle_{\mathcal{Q}})$ be an inner product space and the loss function $L$ have the form $L(q^\dagger,q) = \|q^\dagger-q\|_{\mathcal{Q}}^2$ where $\|\cdot\|_{\mathcal{Q}}$ is the norm induced by the inner product.
Then the optimal information $A_\mu$ for a Bayesian PNM and $A_\mu^*$ for ACA are identical.
\end{theorem}

It is emphasised that this result is \emph{not} a trivial consequence of the correspondance between Bayes rules and worst case optimal methods, as exposed in \cite{Kadane:1983ww}. 
To the best of our knowledge, information-based complexity research has studied $A_\mu^*$ but not $A_\mu$.

Theorem \ref{thm:optimal_information} establishes that, for the squared norm loss, we can extract results on optimal average case information from the ACA literature and use them to construct optimal Bayesian PNMs.
An example is provided next.

\begin{example}[Optimal Information for Probabilistic Integration] \label{ex:optimalinfo}

To illustrate optimal information for Bayesian PNMs, we revisit the first worked example of ACA, due to \cite{Suldin1959,Suldin1960}.
Set $\mathcal{X} = \{x \in C(0,1) : x(0) = 0\}$ and take the belief distribution $\mu$ to be induced from the Weiner process on $\mathcal{X}$, i.e.\ a Gaussian process with mean $0$ and covariance function $k(t,t') = \min(t,t')$.
Our QoI is $Q(x) = \int_0^1 x(t) \rd t$ and the loss function is $L(q,q') = (q-q')^2$.

Consider standard information $A(x) = (x(t_1),\dots,x(t_n))$ for $n$ fixed knots $0 \leq t_1 < \dots < t_n \leq 1$.
Our aim is to determine knots $t_i$ that represent optimal information for a Bayesian PNM with respect to $\mu$ and $L$. 

Motivated by Theorem \ref{thm:optimal_information} we first solve the optimal information problem for ACA and then derive the associated PNM.
It will be sufficient to restrict attention to linear methods $b(a) = \sum_{i=1}^n w_i x(t_i)$ with $w_i \in \mathbb{R}$.
This allows a closed-form expression for the average error:
\begin{equation}
\label{eq:WCE1}
R(\mu,(A , \delta \circ b)) = \frac{1}{3} - 2\sum_{i=1}^n w_i \left(t_i - \frac{1}{2} t_i^2 \right) + \sum_{i,j=1}^n w_i w_j \min(t_i,t_j) .
\end{equation}
Standard calculus can be used to minimise Eq.~\eqref{eq:WCE1} over both the weights $\{w_i\}_{i=1}^n$ and the locations $\{t_i\}_{i=1}^n$; the full calculation can be found in Chapter~2, Section~3.3 of \cite{Ritter2000}.
The result is an ACA optimal method
\begin{equation*}
b(A(x)) = \frac{2}{2n+1} \sum_{i=1}^n x( t_i^* ), \quad t_i^* = \frac{2i}{2n+1}
\end{equation*}
which is recognised as the trapezium rule with equally spaced knots.
The associated contraction rate $r_n$ is $n^{-1}$ \citep{Lee1986}.

From Theorem \ref{thm:optimal_information} we have that ACA optimal information is also optimal information for the Bayesian PNM.
Thus the optimal Bayesian PNM $M = (A,B)$ for the belief distribution $\mu$ is uniquely determined:
\begin{align*}
	A(x) & = \begin{bmatrix} x( t_1^*) \\ \vdots \\ x( t_n^* ) \end{bmatrix} ,
	& B(\mu, a) & = \mathrm{N}\left( \frac{2}{2n+1} \sum_{i=1}^n a_i \; , \; \frac{1}{3 (2n+1)^2} \right) .
\end{align*}
Note how the PNM is centred on the ACA optimal method.
However the PNM itself is not a Bayes rule; it in fact carries twice the Bayes risk as the ACA method.

This illustration can be generalised.
It is known that for $\mu$ induced from the Weiner process on $\partial^s x$, $Q$ a linear functional and $\phi$ a loss function that is convex and symmetric, equi-spaced evaluation points are essentially optimal information, the Bayes rule is the natural spline of degree $2s + 1$, and the contraction rate $r_n$ is essentially $n^{-(s+1)}$; see \cite{Lee1986} for a complete treatment.

\end{example}

This completes our performance assessment for PNMs; next we turn to computational matters.

\section{Numerical Disintegration} \label{sec:sampling}

In this section we discuss algorithms to access the output from a Bayesian PNM.
The approach considered in this paper is to form an explicit approximation to $\mu^a$ that can be sampled.
The construction of a sampling scheme can exploit sophisticated Monte Carlo methods and allow probing $B(\mu,a)$ at a computational cost that is de-coupled from the potentially substantial cost of obtaining the information $a$ itself.

The construction of an approximation to $\mu^a$ is non-trivial on a technical level.
As shown in Section~\ref{sec:conditioning}, under weak conditions on the space $\mathcal{X}$ and the operator $A$, the disintegration $\mu^a$ is well-defined for $A_{\#}\mu$-almost all $a \in \mathcal{A}$.
The approach considered in this work is based on sampling from an approximate distribution $\mu_\delta^a$ which converges in an appropriate sense to $\mu^a$ in the $\delta \downarrow 0$ limit.
This follows in a similar spirit to \cite{Ackerman2017}.


\subsection{Sequential Approximation of a Disintegration} \label{sec:rcp_approximation}

Suppose that $\mathcal{A}$ is an open subset of $\mathbb{R}^n$ and that the distribution $A_\# \mu \in \mathcal{P}_{\mathcal{A}}$, admits a continuous and positive density $p_A$ with respect to Lebesgue measure on $\mathcal{A}$.
Further endow $\mathcal{A}$ with the structure of a Hilbert space, with norm $\|\cdot\|_{\mathcal{A}}$. 

Let $\phi \colon \mathbb{R}_+ \rightarrow \mathbb{R}_+$ denote a decreasing function, to be specified, that is continuous at $0$, with $\phi(0) = 1$ and $\lim_{r \rightarrow \infty} \phi(r) = 0$.
Consider
\begin{equation*}
	\mu_\delta^a(\mathrm{d}x) \defeq \frac{1}{Z_\delta^a} \phi\left(\frac{\|A(x) - a\|_{\mathcal{A}}}{\delta} \right) \mu(\mathrm{d}x)
\end{equation*}
where the normalisation constant 
\begin{equation*}
	Z_\delta^a \defeq \int \phi\left(\frac{\|\tilde{a} - a\|_{\mathcal{A}}}{\delta} \right) p_A(\mathrm{d} \tilde{a})
\end{equation*}
is non-zero since $p_A$ is bounded away from 0 on a neighbourhood of $a \in \mathcal{A}$ and $\phi$ is bounded away from 0 on a sufficiently small interval $[0,\gamma]$.
Our aim is to approximate $\mu^a$ with $\mu^a_\delta$ for small bandwidth parameter $\delta$.
The construction, which can be considered a mathematical generalisation of approximate Bayesian computation \citep{DelMoral2012}, ensures that $\mu^a_\delta \ll \mu$.
The role of $\phi$ is to admit states $x \in \mathcal{X}$ for which $A(x)$ is close to $a$ but not necessarily equal.
It is assumed to be sufficiently regular:

\begin{assumption} \label{varphi_assumption}
There exists $\alpha > 0$ such that $C_\phi^\alpha \defeq \int r^{\alpha + n - 1} \phi(r) \mathrm{d}r < \infty$.
\end{assumption}

To discuss the convergence of $\mu_\delta^a$ to $\mu^a$ we must first select a metric on $\mathcal{P}_{\mathcal{X}}$.
Let $\mathcal{F}$ be a normed space of (measurable) functions $f \colon \mathcal{X} \to \reals$ with norm $\norm{\cdot}_{\mathcal{F}}$.
For measures $\nu,\nu' \in \mathcal{P}_\mathcal{X}$, define
\begin{equation*}
	d_{\mathcal{F}}(\nu,\nu') = \sup_{\|f\|_{\mathcal{F}} \leq 1} |\nu(f) - \nu'(f)| .
\end{equation*}
This formulation encompasses many common probability metrics such as the total variation distance and Wasserstein distance \citep{Mueller1997}.
However, not all spaces of functions $\mathcal{F}$ lead to useful theory. In particular the total variation distance between $\mu^a$ and $\mu^{a'}$ for $a \neq a'$ will be one in general. 
Furthermore depending on the choice of $\mathcal{F}$, $d_\mathcal{F}$ may be merely a pseudometric\footnote{For a pseudometric, $d_\mathcal{F}(x, y) = 0 \implies x = y$ need not hold.}.
Sufficient conditions for weak convergence with respect to $\mathcal{F}$ are now established:

\begin{assumption} \label{assumption:lipschitz_rcp}
The map $a \mapsto \mu^a$ is almost everywhere $\alpha$-H\"{o}lder continuous in $d_{\mathcal{F}}$, i.e.
	\begin{equation*}
		d_{\mathcal{F}}(\mu^a , \mu^{a'}) \leq C_\mu^\alpha \|a - a'\|_{\mathcal{A}}^\alpha
	\end{equation*}
	for some constant $C_\mu^\alpha > 0$ and for $A_{\#}\mu$ almost all $a, a' \in \mathcal{A}$.
\end{assumption}
Sufficient conditions for Assumption \ref{assumption:lipschitz_rcp} are discussed in \cite{Ackerman2017}, but are somewhat technical.

\begin{theorem} \label{thm:rcp_contraction}
Let $\bar{C}_\phi^\alpha \defeq C_\phi^\alpha / C_\phi^0$.
Then, for $\delta>0$ sufficiently small,
	\begin{equation*}
		d_{\mathcal{F}}(\mu^a_\delta , \mu^a ) \leq C_\mu^\alpha (1 + \bar{C}_\phi^\alpha) \delta^\alpha 
	\end{equation*}
	for $A_{\#}\mu$ almost all $a \in \mathcal{A}$.
\end{theorem}
This result justifies the approximation of $\mu^a$ by $\mu_\delta^a$ when the QoI can be well-approximated by integrals with respect to $\mathcal{F}$.
This result is stronger than that of earlier work, such as \cite{Pfanzagl:1979hd}, in that it holds for infinite-dimensional $\mathcal{X}$, though it also relies upon the stronger H\"{o}lder continuity assumption.

The specific form for $\phi$ is not fundamental, but can impact upon rate constants.
For the choice $\phi(r) = 1[r < 1]$ we have $\bar{C}_\phi^\alpha = \frac{n}{\alpha + n}$, which can be bounded independent of the dimension $n$ of $\mathcal{A}$.
On the other hand, for $\phi(r) = \exp(-\frac{1}{2}r^2)$ it can be shown that, for $\alpha \in \mathbb{N}$,
\begin{equation}
	\bar{C}_\phi^\alpha = \frac{(\alpha + n - 1)!!}{(n-1)!!}
\end{equation}
so that the constant $\bar{C}_\phi^\alpha$ might not be bounded.
In general this necessitates effective Monte Carlo methods that are able to sample from the regime where $\delta$ can be extremely small, in order to control the overall approximation error.

\subsection{Computation for Series Priors} \label{sec: non-Gaussian comp}

The series representation of $\mu$ in Eq.~\eqref{eq:generic_random_field} of Section~\ref{sec:rcp_prior} is infinite-dimensional and thus cannot, in general, be instantiated.
To this end, define $\mathcal{X}_N = x_0 + \text{span}\{\phi_0 , \dots , \phi_N\}$ and define the associated projection operator $P_N \colon \mathcal{X} \rightarrow \mathcal{X}_N$ as
\begin{equation*}
	P_N \left( x_0 + \sum_{i=0}^\infty u_i \phi_i \right) \defeq x_0 + \sum_{i=0}^N u_i \phi_i .
\end{equation*}
A natural approach is to compute with the modified information operator $A \circ P_N$ instead of $A$.
This has the effect of updating the distribution of the first $N+1$ coefficients and leaving the tail unchanged, to produce an output $\mu_{\delta,N}^a$.
Then computation performed in the Bayesian \emph{update} step is finite-dimensional, whilst instantiation of the posterior itself remains infinite-dimensional.
A ``likelihood-informed'' choice of basis $\set{\phi_i}$ in such problems was considered in \cite{Cui:2016ev}.

Inspired by this approach, we next considered convergence of the output $\mu_{\delta,N}^a$ to $\mu_\delta^a$ in the limit $N \rightarrow \infty$.
In this section it is additionally required that $\phi$ be everywhere continuous with $\phi > 0$.
Let $\varphi = - \log \phi$, so that $\varphi$ is a continuous bijection of $\mathbb{R}_+$ to itself. 
The following are also assumed:

\begin{assumption} \label{assumption:local_Lipschitz}
For each $R > 0$, it holds that $|\varphi(r) - \varphi(r')| \leq C_R |r-r'|$ for some constant $C_R$ and all $r , r' < R$.
\end{assumption}

\begin{assumption} \label{assumption:prior_truncation}
$\|A(x) - A\circ P_N(x)\|_{\mathcal{A}} \leq \exp(m(\|x\|_{\mathcal{X}})) \Psi(N)$ for all $x \in \mathcal{X}$, where $m$ is measurable and satisfies $\mathbb{E}_{X \sim \mu}[\exp(2 m(\|X\|_{\mathcal{X}}))] < \infty$ and $\Psi(N)$ vanishes as $N$ is increased.
\end{assumption}

\begin{assumption} \label{assumption:bounded_info}
	$\sup_{x \in \mathcal{X}} \|A(x)\|_{\mathcal{A}} < \infty$.
\end{assumption}

\begin{assumption} \label{assumption:norm_bounds}
	$\norm{f}_\infty \leq C_{\mathcal{F}} \norm{f}_{\mathcal{F}}$ for some constant $C_{\mathcal{F}}$ and all $f \in \mathcal{F}$.
\end{assumption}

Assumption \ref{assumption:local_Lipschitz} holds for the case $\varphi(r) = \frac{1}{2}r^2$ with constant $C_R = R$.
Assumption \ref{assumption:prior_truncation} is standard in the inverse problem literature; for instance it is shown to hold for certain series priors in Theorem 3.4 of \cite{Cotter2010}.
Assumption \ref{assumption:bounded_info} is, in essence, a compactness assumption, in that it is implied by compactness of the state space $\mathcal{X}$ when $A$ is linear.
In this sense it is a strong assumption; however it can be enforced in our experiments, where $\mathcal{X}$ is unbounded, through a threshold map 
\begin{equation*}
	\tilde{A}(x) \defeq \begin{cases} A(x) & \text{if $\|A(x)\|_{\mathcal{A}} \leq \lambda_{\max}$,} \\ \lambda_{\max} \frac{A(x)}{\|A(x)\|_{\mathcal{A}}} & \text{if $\|A(x)\|_{\mathcal{A}} > \lambda_{\max}$,} \end{cases}
\end{equation*}
where $\lambda_{\max}$ is a large pre-defined constant.
Assumption \ref{assumption:norm_bounds} places a restriction on the probability metric $d_{\mathcal{F}}$ in which our result is stated.

The following theorem has its proof in the Electronic Supplement:
\begin{theorem} \label{thm:truncation error}
For some constant $C_\delta$, dependent on $\delta$, it holds that $d_{\mathcal{F}}(\mu_{\delta,N}^a , \mu_\delta^a) \leq C_\delta \Psi(N)$.
\end{theorem}

An immediate consequence of Theorems \ref{thm:rcp_contraction} and \ref{thm:truncation error} is that the total approximation error can be bounded by applying the triangle inequality:
\begin{equation*}
	d_{\mathcal{F}} (\mu^a, \mu_{\delta, N}^a ) \leq C_\mu^\alpha (1 + \bar{C}_\phi^\alpha) \delta^\alpha + C_\delta \Psi(N) .
\end{equation*}
In particular, we have convergence of $\mu_{\delta,N}^a$ to $\mu^a$ in the $\delta \downarrow 0$ limit provided that the number of basis functions satisfies $C_\delta \Psi(N) = o(1)$.

The approximate posterior $\mu_{\delta,N}^a$ analysed above can be sampled when $\mu$ is Gaussian, since the first $N+1$ coefficients can be handled with MCMC and the tail $\sum_{i=N+1}^\infty u_i \phi_i$, being Gaussian, can be sampled.
However, when $\mu$ is non-Gaussian the tail is not recognised in a form that can be sampled.
For the experiments in Section \ref{sec:results}, in which both Gaussian and non-Gaussian priors $\mu$ are considered, the series in Eq.~\eqref{eq:generic_random_field} was truncated at level $N+1$, with the resultant prior denoted $\mu_N$.
The associated posterior was then entirely supported on the finite-dimensional subspace $\mathcal{X}_N$; this is mathematically equivalent to working with the \emph{projected output} $P_N \mu_{\delta,N}^a$.
Analysis of prior truncation, as opposed to modification of the information operator just reported, is known to be difficult.
Indeed, while $\mu_{N}$ converges to $\mu$ weakly, it does not do so in total variation, and this deficiency generally transfers to the associated posteriors.
In general the impact of prior perturbation is a subtle topic --- see e.g.\ \cite{Owhadi2015b} and the references therein --- and we therefore defer theoretical analysis of this approximation to future work.

\subsection{Monte Carlo Methods for Numerical Disintegration} \label{sec: MC for ND}

The previous sections established a sequence of well-defined distributions $\mu_\delta^a$ (or $\mu_{\delta,N}^a$ for non-Gaussian models) which converge (in a specific weak sense) to the exact disintegration $\mu^a$.
From construction, $\mu_\delta^a \ll \mu^a$ and this is sufficient to allow standard Monte Carlo methods to be used.
The construction of Monte Carlo methods is de-coupled from the core material in the main text and the main methodological considerations are well-documented \citep[e.g.][]{Girolami:2011hw}.

For the experiments reported in subsequent sections two approaches were explored; a Sequential Monte Carlo (SMC) method \citep{Doucet2001a} and a parallel tempering method \citep{Geyer:1991ws}.
This provided a transparent sampling scheme, whose non-asymptotic approximation error can be theoretically understood.
In particular, they provide robust estimators of model evidence that can be used for Bayesian model comparison.
Full details of the Monte-Carlo methods used for this work, along with associated theoretical analysis for the SMC method, are contained in Section \ref{sec:rcp_smc} of the Electronic Supplement.

\section{Computational Pipelines and PNM} \label{sec:pipelines}

The last theoretical development in this paper concerns composition of several PNMs.
Most analysis of numerical methods focuses on the error incurred by an individual method.
However, real-world computational procedures typically rely on the composition of several numerical methods.
The manner in which accumulated discretisation error affects computational output may be highly non-trivial \citep{Roy2010,Anderson2011,Babuska2016}.
An extreme example occurs when one of the numerical methods in a pipeline is charged with integration of a chaotic dynamical system \citep{Strogatz2014}.

In recent work, \cite{Chkrebtii:2013ux}, \cite{Conrad:2016gv} and \cite{Cockayne:2016ts} each used PNMs within a broader statistical procedure to estimate unknown parameters in systems of ODEs and PDEs.
The probabilistic description of discretisation error was incorporated into the data-likelihood, resulting in posterior distributions for parameters with inflated uncertainty to properly account for the inferential impact of discretisation error.
However, beyond these limited works, no examination of the composition of PNMs has been performed.
In particular, the question of which PNMs can be composed, and when the output of such a composition is meaningful, has not been addressed.
This is important; for instance, if the output of a composition of PNMs is to be used for analysis of variance to elucidate the main sources of discretisation error, then it is important that such output is meaningful.

This section defines a \emph{pipeline} as an abstract graphical object that may be combined with a collection of \emph{compatible} PNMs.
It is proven that when compatible Bayesian PNMs are employed in the pipeline, the distributional output of the pipeline carries a Bayesian interpretation under an explicit conditional independence condition on the prior $\mu$.

To build intuition, for the simple case where two Bayesian PNMs are composed in series, our results provide conditions for when, informally, the output $B_2(B_1(\mu,a_1),a_2)$ corresponds to a single Bayesian procedure $B(\mu,(a_1,a_2))$.
To reduce the notational and technical burden, in this section we will not provide rigorous measure theoretic details; however we note that those details broadly follow the same pattern as in Section \ref{sec:conditioning}.

\subsection{Computational Pipelines}

To analyse pipelines of PNMs, we consider $n$ such methods $M_1,\dots,M_n$, where each method $M_i = (A_i,B_i)$ is defined on a common\footnote{This is without loss of generality, since $\mathcal{X}$ can be taken as the union of all state spaces required by the individual methods.} state space $\mathcal{X}$ and targets a QoI $Q_i \in \mathcal{Q}_i$.
A pipeline will be represented as a directed graphical model, wherein the QoIs $Q_i$ from parent methods constitute information operators for child methods. 
It may be that a method will take quantities from multiple parents as input.
To allow for this, we suppose that the information operator $A_i \colon \mathcal{X} \rightarrow \mathcal{A}_i$ can be decomposed into components $A_{i,j} \colon \mathcal{X} \rightarrow \mathcal{A}_{i,j}$ such that $A_i = (A_{i,1},\dots,A_{i,m(i)})$ and $\mathcal{A}_i = \mathcal{A}_{i,1} \times \dots \times \mathcal{A}_{i,m(i)}$. 
Thus, each component $A_{i,j}$ can be thought of as the QoI output by one of the parents of the method $M_i$.

Without loss of generality we designate the $n$th QoI $Q_n$ to be the \emph{principal} QoI.
That is, the purpose of the computational pipeline is to estimate $Q_n$.
The case of multiple principal QoI is a simple extension not described herein.
Nodes with no immediate children are called \emph{terminal} nodes, while nodes with no immediate parents are called \emph{source nodes}. We denote by $A$ the set of all source nodes.

\begin{definition}[Pipeline]
A \emph{pipeline} $P$ is a directed acyclic graph defined as follows:
\begin{itemize}
	\item Nodes are of two kinds:
	\emph{Information} nodes are depicted by $\square$, and \emph{method} nodes  are depicted by $\blacksquare$.
	\item The graph is bipartite, so that edges connect a method node to an information node or vice-versa.
	That is, edges are of the form $\square \rightarrow \blacksquare$ or $\blacksquare \rightarrow \square$.
	\item There are $n$ method nodes, each with a unique label in $\{1,\dots,n\}$.
	\item The method node labelled $i$ has $m(i)$ parents and one child. 
	Its in-edges are assigned a unique label in $\{1,\dots,m(i)\}$.
	\item There is a unique terminal node and it is the child of method node $n$. This represents the principal QoI $Q_n$.
\end{itemize}
\end{definition}

\begin{example}[Distributed Integration] \label{ex:divide}

Recall the numerical integration problem of Example \ref{ex:optimalinfo} and, as a thought experiment, consider partitioning the domain of integration in order to distribute computation:
\begin{equation}
\underbrace{\int_0^1 x(t) \rd t}_{\text{(c)}} = \underbrace{\int_0^{0.5} x(t) \rd t}_{\text{(a)}} + \underbrace{\int_{0.5}^1 x(t) \rd t}_{\text{(b)}}
\label{eq:split_integral}
\end{equation}
To keep presentation simple we consider an integral over $[0,1]$ with $2m+1$ equidistant knots $t_i = i / 2m$.
Let $M_1$ be a Bayesian PNM for estimating $Q_1(x) =$ (a) and $M_2$ be a Bayesian PNM for estimating $Q_2(x) =$ (b).

In terms of our notational convention, we divide the information operator into four components;
$A_{i,j}$, for $i,j \in \set{1,2}$. $A_{1,1}$ and $A_{2,2}$ contain the information unique to $M_1$ and $M_2$.
Specifically
\begin{equation*}
	A_{1,1}(x) = \left[ \begin{array}{c} x(t_1) \\ \vdots \\ x(t_{m-1}) \end{array} \right], \qquad A_{2,2}(x) = \left[ \begin{array}{c} x(t_{m+1}) \\ \vdots \\ x(t_{2m}) \end{array} \right] .
\end{equation*}
$A_{1,2}$ and $A_{2,1}$ contain the information that is shared between the two methods; that is $A_{1,2} = A_{2,1} = \set{ x(t_m) }$.
To complete the specification we need a third PNM for estimation of $Q_3(x) =$ (c) which we denote $M_3$ and which combines the outputs of $M_1$ and $M_2$ by simply adding them together.
Formally this has information operator $A_3(x) = (A_{3,1}(x) , A_{3,2}(x))$ where $A_{3,1}(x) =$ (a) and $A_{3,2}(x) =$ (b). Its belief update operator is given by:
\begin{equation*}
B_3(\mu,(a_{3,1},a_{3,2})) = \delta(a_{3,1} + a_{3,2} )
\end{equation*}
An intuitive graphical representation of this set-up is shown in Figure \ref{fig:distributed_integration}.
The pipeline $P$ itself, which is identical to Figure \ref{fig:distributed_integration} but with additional node and edge labels, is shown in Figure \ref{fig:PN_pipeline}.
\end{example}

\begin{figure}[t!]
\centering
\resizebox{0.8\textwidth}{!}{
\begin{tikzpicture}

\tikzstyle{square}=[regular polygon,regular polygon sides=4];
\tikzstyle{info}=[draw,rectangle,fill = black!0,minimum width=1.2cm];
\tikzstyle{meth}=[draw,rectangle,fill = black!100,minimum width=1.2cm,text=white];
\tikzstyle{arrow}=[very thick,->];

\node[info] at (-1,0) (I1) {$x(t_1),\dots,x(t_{m-1})$};
\node[info] at (-0.5,-1) (I2) {$x(t_m)$};
\node[info] at (-1,-2) (I3) {$x(t_{m+1}),\dots,x(t_{2m})$};

\node[meth] at (2,0) (M1) {$B_1(\mu,\cdot)$};
\node[meth] at (2,-2) (M2) {$B_2(\mu,\cdot)$};

\node[info] at (4,0) (I5) {$\int_0^{0.5} x(t) \mathrm{d} t$};
\node[info] at (4,-2) (I6) {$\int_{0.5}^1 x(t) \mathrm{d} t$};

\node[meth] at (6,-1) (M3) {$B_3(\mu,\cdot)$};

\node[info] at (8,-1) (I8) {$\int_0^1 x(t) \mathrm{d} t$};

\path[arrow] (I1) edge node [above] {} (M1);
\path[arrow] (I2) edge node [below] {} (M1);
\path[arrow] (I2) edge node [above] {} (M2);
\path[arrow] (I3) edge node [below] {} (M2);

\path[arrow] (M1) edge (I5);
\path[arrow] (M2) edge (I6);

\path[arrow] (I5) edge node [above] {} (M3);
\path[arrow] (I6) edge node [below] {} (M3);

\path[arrow] (M3) edge (I8);
\end{tikzpicture}}
\caption{An intuitive representation of Example \ref{ex:divide}.}
\label{fig:distributed_integration}
\end{figure}
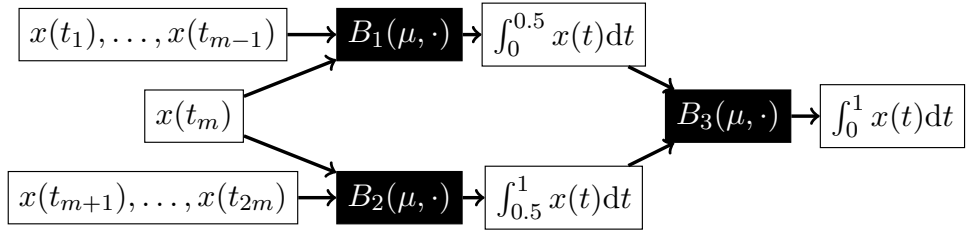

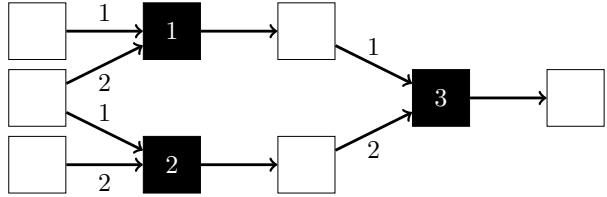
\begin{figure}[t!]
\centering
\resizebox{0.5\textwidth}{!}{
\begin{tikzpicture}

\tikzstyle{square}=[regular polygon,regular polygon sides=4];
\tikzstyle{info}=[draw,square,fill = black!0,minimum width=1.2cm];
\tikzstyle{meth}=[draw,square,fill = black!100,minimum width=1.2cm,text=white];
\tikzstyle{arrow}=[very thick,->];

\node[info] at (0,0) (I1) {};
\node[info] at (0,-1) (I2) {};
\node[info] at (0,-2) (I3) {};

\node[meth] at (2,0) (M1) {1};
\node[meth] at (2,-2) (M2) {2};

\node[info] at (4,0) (I5) {};
\node[info] at (4,-2) (I6) {};

\node[meth] at (6,-1) (M3) {3};

\node[info] at (8,-1) (I8) {};

\path[arrow] (I1) edge node [above] {$1$} (M1);
\path[arrow] (I2) edge node [below] {$2$} (M1);
\path[arrow] (I2) edge node [above] {$1$} (M2);
\path[arrow] (I3) edge node [below] {$2$} (M2);

\path[arrow] (M1) edge (I5);
\path[arrow] (M2) edge (I6);

\path[arrow] (I5) edge node [above] {$1$} (M3);
\path[arrow] (I6) edge node [below] {$2$} (M3);

\path[arrow] (M3) edge (I8);
\end{tikzpicture}}
\caption{The pipeline $P$ corresponding to Figure \ref{fig:distributed_integration}.}
\label{fig:PN_pipeline}
\end{figure}

In general, the method node labelled $i$ is taken to represent the method $M_i$.
The in-edge to this node labelled $j$ is taken to represent the information provided by the relationship $A_{i,j} (x) = a_{i,j}$.
Here $a_{i,j}$ can either be deterministic information provided to the pipeline, or statistical information derived from the output of another PNM.
To make this formal and to ``match the input-output spaces'' we next define what it means for the collection of methods $M_i$ to be compatible with the pipeline $P$.
Informally, this describes the conditions that must be satisfied for method nodes in a pipeline to be able to connect to each other.

\begin{definition}[Compatible]
The collection $(M_1,\dots, M_n)$ of PNMs is \emph{compatible} with the pipeline $P$ if the following two requirements are satisfied:
\begin{enumerate}[(i)]
\item (Method nodes which share an information node must have consistent information spaces and information operators.)
For a motif 
\begin{center}
\resizebox{0.25\textwidth}{!}{
\begin{tikzpicture}
\tikzstyle{square}=[regular polygon,regular polygon sides=4];
\tikzstyle{info}=[draw,square,fill = black!0,minimum width=1.2cm];
\tikzstyle{meth}=[draw,square,fill = black!100,minimum width=1.2cm,text=white];
\tikzstyle{arrow}=[very thick,->];

\node[meth] at (-2,0) (M1) {$i$};
\node[info] at (0,0) (I) {};
\node[meth] at (2,0) (M2) {$j$};

\path[arrow] (I) edge node [above] {$i'$} (M1);
\path[arrow] (I) edge node [above] {$j'$} (M2);
\end{tikzpicture}}
\end{center}
we have that $A_{i,i'} = A_{j,j'}$ and $\mathcal{A}_{i,i'} = \mathcal{A}_{j,j'}$.
\item (The space $\mathcal{Q}_i$ for the output of a previous method must be consistent with the information space of the next method.)
For a motif
\begin{center}
\resizebox{0.25\textwidth}{!}{
\begin{tikzpicture}
\tikzstyle{square}=[regular polygon,regular polygon sides=4];
\tikzstyle{info}=[draw,square,fill = black!0,minimum width=1.2cm];
\tikzstyle{meth}=[draw,square,fill = black!100,minimum width=1.2cm,text=white];
\tikzstyle{arrow}=[very thick,->];

\node[meth] at (0,0) (M1) {$i$};
\node[info] at (2,0) (I) {};
\node[meth] at (4,0) (M2) {$j$};

\path[arrow] (M1) edge node [above] {} (I);
\path[arrow] (I) edge node [above] {$j'$} (M2);
\end{tikzpicture}}
\end{center}
we have that $\mathcal{Q}_i = \mathcal{A}_{j,j'}$.
\end{enumerate}
\end{definition}

Note that we do not require the converse of (i) at this stage;
that is, the same information can be represented by more than one node in the pipeline. 
This permits redundancy in the pipeline, in that information is not recycled.
It will transpire that pipelines with such redundancy are non-Bayesian.

The role of the pipeline $P$ is to specify the order in which information, either deterministic of statistical, is propagated through the collection of PNMs.
This is illustrated next:

\begin{example}[Propagation of Information] \label{ex:propaga}
For the pipeline in Figure \ref{fig:PN_pipeline}, the propagation of information proceeds as follows::
\begin{enumerate}
	\item The source nodes, representing $A(x) = \{ A_{1,1}(x), A_{1,2}(x) = A_{2,1}(x), A_{2,2}(x) \}$ are evaluated as $\{a_{1,1},a_{1,2}=a_{2,1},a_{2,2}\}$.
	This represents all the information on $x$ that is provided.
	\item The distributions 
	\begin{align*}
		\mu^{(1)} & \defeq B_1(\mu, (a_{1,1},a_{1,2}) ) \\
		\mu^{(2)} & \defeq B_2(\mu, (a_{2,1},a_{2,2}) ) 
	\end{align*}
	are computed.
	\item The push-forward distribution
	\begin{equation*}
		\mu^{(3)} \defeq (B_3)_{\#}(\mu , \mu^{(1)} \times \mu^{(2)}) 
	\end{equation*}
	is computed.
\end{enumerate}
Here $\mu^{(1)} \times \mu^{(2)}$ is defined on the Cartesian product $\Sigma_{\mathcal{A}_{3,1}} \times \Sigma_{\mathcal{A}_{3,2}}$ with independent components $\mu^{(1)}$ and $\mu^{(2)}$.
The notation $(B_3)_\#$ refers to the push-forward of the function $B_3(\mu,\cdot)$ over its second argument.
The distribution $\mu^{(3)}$ is the output of the pipeline and is a distribution over the principal QoI $Q_3(x)$.
\end{example}

The procedure in Example \ref{ex:propaga} can be formalised, but to keep the presentation and notation succinct, we leave this implicit:

\begin{definition}[Computation] \label{def:bayes_computation}
For a collection $(M_1,\dots, M_n)$ of PNMs that are compatible with a pipeline $P$, the \emph{computation} $P(M_1,\dots,M_n)$ is defined as the PNM with information operator $A$ and belief update operator $B$ that takes $\mu$ and $A(x) = a$ as input and returns the distribution $\mu^{(n)}$ as its output $B(\mu,a)$, obtained through the procedure outlined in Example \ref{ex:propaga}.
\end{definition}
\noindent
That is, the \emph{computation} $P(M_1,\dots,M_n)$ is a PNM for the principal QoI $Q_n$.
Note that this definition includes a classical numerical work-flow just as a PNM encompasses a standard numerical method.

\subsection{Bayesian Computational Pipelines}

Noting that $P(M_1,\dots,M_n)$ is itself a PNM, there is a natural definition for when such a computation can be called Bayesian:

\begin{definition}[Bayesian Computation]
Denote by $(A, B)$ the information and belief operators associated with the computation $P(M_1,\dots,M_n)$ and let $\{\mu^a\}_{a \in \mathcal{A}}$ be a disintegration of $\mu$ with respect to the information operator $A$. 
The computation $P(M_1,\dots,M_n)$ is said to be \emph{Bayesian} for the QoI $Q_n$ if 
\begin{equation*}
B(\mu,a) = (Q_n)_{\#} \mu^a \quad  \text{ for } A_{\#} \mu \text{-almost-all } a \in \mathcal{A} .
\end{equation*}
\end{definition}
\noindent
This is clearly an appealing property; the output of a Bayesian computation can be interpreted as a posterior distribution over the QoI $Q_n(x)$ given the prior $\mu$ and the information $A(x)$.
Or, more informally, the ``pipeline is lossless with information''.
However, at face value it seems difficult to verify whether a given computation $P(M_1,\dots,M_n)$ is Bayesian, since it depends on both the individual PNMs $M_i$ and the pipeline $P$ that combines them. 
Our next aim is to establish verifiable sufficient conditions, for which we require another definition:

\begin{definition}[Dependence Graph]
The \emph{dependence graph} of a pipeline $P$ is the directed acyclic graph $G(P)$ obtained by taking the pipeline $P$, removing the method nodes and replacing all $\square \rightarrow \blacksquare \rightarrow \square$ motifs with direct edges $\square \rightarrow \square$.
\end{definition}
The dependency graph for Example \ref{ex:divide} is shown in Figure \ref{fig:PN_dependence_graph}.

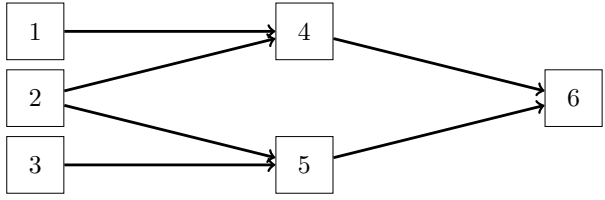
\begin{figure}[t!]
\centering
\resizebox{0.5\textwidth}{!}{
\begin{tikzpicture}

\tikzstyle{square}=[regular polygon,regular polygon sides=4];
\tikzstyle{info}=[draw,square,fill = black!0,minimum width=1.2cm];
\tikzstyle{arrow}=[very thick,->];

\node[info] at (0,0) (I1) {1};
\node[info] at (0,-1) (I2) {2};
\node[info] at (0,-2) (I3) {3};

\node[info] at (4,0) (I4) {4};
\node[info] at (4,-2) (I5) {5};

\node[info] at (8,-1) (I6) {6};

\path[arrow] (I1) edge (I4);
\path[arrow] (I2) edge (I4);
\path[arrow] (I2) edge (I5);
\path[arrow] (I3) edge (I5);

\path[arrow] (I4) edge (I6);
\path[arrow] (I5) edge (I6);
\end{tikzpicture}}
\caption{Dependence graph $G(P)$ corresponding to the pipeline $P$ in Figure \ref{fig:PN_pipeline}.
The nodes are indexed with a topological ordering (shown).}
\label{fig:PN_dependence_graph}
\end{figure}

For a computation $P(M_1,\dots,M_n)$, each of the $J$ distinct nodes in $G(P)$ can be associated with a random variable $Y_j$ where either $Y_j = A_{k,l}(X)$ for some $k,l$, when the node is a source, or otherwise $Y_j = Q_k(X)$, for some $k$.
Randomness here is understood to be due to $X \sim \mu$, so that the distribution of the $\{Y_j\}_{j=1}^J$ is a function of $\mu$.
The convention used here is that the $Y_j$ are indexed according to a topological ordering on $G(P)$, which has the properties that (i) the source nodes correspond to indices $1,\dots,I$, and (ii) the final random variable is $Y_J = Q_n(X)$.

\begin{definition}[Coherence] \label{def:markov_condition}
	Consider a computation $P(M_1,\dots,M_n)$.
	Denote by $\pi(j) \subseteq \{1,\dots,j-1\}$ the parent set of node $j$ in the dependence graph $G(P)$.
	Then we say that $\mu \in \mathcal{P}_{\mathcal{X}}$ is \emph{coherent} for the computation $P(M_1,\dots,M_n)$ if the implied joint distribution of the random variables $Y_1, \dots, Y_J$ satisfies:
	\begin{equation*}
		Y_j \ci Y_{\{1,\dots,j-1\} \setminus \pi(j)} \; | \; Y_{\pi(j)}
	\end{equation*}
	for all $j = I+1,\dots,J$.
\end{definition}
Note that this is weaker than the Markov condition for directed acyclic graphs \citep[see][]{Lauritzen1991}, since we do not insist that the variables represented by the source nodes are independent.
It is emphasised that, for a given $\mu \in \mathcal{P}_{\mathcal{X}}$, the coherence condition can in general be checked and verified.

The following result provides sufficient and verifiable conditions which ensure that a computation composed of individual Bayesian PNMs is a Bayesian computation:

\begin{theorem} \label{thm:markov}
	Let $M_1,\dots,M_n$ be Bayesian PNMs and let $\mu \in \mathcal{P}_{\mathcal{X}}$ be coherent for the computation $P(M_1,\dots,M_n)$.
	Then it holds that the computation $P(M_1,\dots,M_n)$ is Bayesian for the QoI $Q_n$.
\end{theorem}
Conversely, if non-Bayesian PNM are combined then the computation $P(M_1,\dots,M_n)$ need not be Bayesian in general.

\begin{example}[Example \ref{ex:divide}, continued]
The random variables $Y_i$ in this example are:
$$
Y_1 = \{X(t_i)\}_{i=1}^{m-1}, \quad Y_2 = X(t_m), \quad Y_3 = \{X(t_i)\}_{i=m+1}^{2m}, \quad Y_4 = \int_0^{0.5} X(t) \mathrm{d}t, \quad Y_5 = \int_{0.5}^1 X(t) \mathrm{d}t.
$$
From $G(P)$ in Figure \ref{fig:PN_dependence_graph}, coherence condition in Definition \ref{def:markov_condition} requires that the non-trivial conditional independences $Y_4 \ci Y_3 \; | \; \{Y_1,Y_2\}$ and $Y_5 \ci Y_1 \; | \; \{Y_2,Y_3\}$ hold.
Thus the distribution $\mu$ is coherent for the computation $P(M_1,M_2,M_3)$ if and only if, for $X \sim \mu$, the associated information variables satisfy $\int_0^{0.5} X(t) \wrt t \ci \{X(t_i)\}_{i=m+1}^{2m} | \{X(t_i)\}_{i=1}^{m}$ and $\int_{0.5}^1 X(t) \wrt t \ci \{X(t_i)\}_{i=1}^{m-1} | \{X(t_i)\}_{i=m}^{2m}$.

The distribution $\mu$ induced by the Weiner process on $x$ in Example \ref{ex:optimalinfo} satisfies these conditions.
Indeed, under $\mu$ the stochastic process $\{x(t) : t > t_m\}$ is conditionally independent of its history $\{x(t) : t < t_m\}$ given the current state $x(t_m)$.
Thus for this choice of $\mu$, from Theorem \ref{thm:markov} we have that $P(M_1,M_2,M_3)$ is Bayesian and parallel computation of $(a)$ and $(b)$ in Eq.~\eqref{eq:split_integral} can be justified from a Bayesian statistical standpoint.

However, for the alternative of belief distributions induced by the Weiner process on $\partial^s x$, this condition is not satisfied and the computation $P(M_1,M_2,M_3)$ is not Bayesian.
To turn this into a Bayesian procedure for these alternative belief distributions it would be required that $A_{1,2}(x)$ provides information about the derivatives $\partial^k x(t_m)$ for all orders $k \leq s$.
\end{example}

\subsection{Monte Carlo Methods for Probabilistic Computation}

The most direct approach to access $\mu^{(n)}$ is to sample from each Bayesian PNM and treat the output samples as inputs to subsequent PNM. 
This is sometimes known as \emph{ancestral sampling} in the Bayesian network literature \citep[e.g.][]{Paige2016}, and is illustrated in the following example:

\begin{example}[Ancestral Sampling for PNM]

For Example \ref{ex:divide}, ancestral sampling proceeds as follows:
\begin{enumerate}
\item Draw initial samples
\begin{align*}
q_1 & \sim B_1(\mu, (a_{1,1} , a_{1,2}) ) \\
q_2 & \sim B_2(\mu, (a_{2,1} , a_{2,2}) ) 
\end{align*}
\item Draw a final sample
\[
q_3 \sim B_3(\mu ,  (q_1 , q_2) ) 
\]
\end{enumerate}
Then $q_3$ is a draw from $\mu^{(3)}$.
\end{example} 

Ancestral sampling requires that PNM outputs can be sampled.
Such sampling methods were discussed in Section \ref{sec: MC for ND}.
For a more general approach, sequential Monte Carlo methods can be used to propagate a collection of particles through the pipeline $P$, similar to work on SMC for general graphical models \citep{Briers2005,Ihler2009,Lienart2015,Lindsten2016,Paige2016}.

\section{Numerical Experiments} \label{sec:results}

In this final section of the paper we present three numerical experiments.
The first is a linear PDE, the second is a nonlinear ODE and the third is an application to a problem in industrial process monitoring, described by a pipeline of PNM.
In each case we experiment with non-Gaussian belief distributions and, in doing so, go beyond previous work.

\subsection{Poisson Equation} \label{sec:poisson}


Our first illustration is an instance of the Poisson equation, a linear PDE with mixed Dirichlet-Neumann boundary conditions:
\begin{align}
	-\nabla^2 x(t) &= 0 				& \;  t&\in(0,1)^2 &&& \label{eq:elliptic_interior}\\
	x(t) &= t_1 						& \;  t_1 &\in [0,1] 	&\; t_2 &= 0 \label{eq:elliptic_dirichlet_1}\\
	x(t) &= 1-t_1 						& \;  t_1 &\in [0,1] 	&\; t_2 &= 1 \label{eq:elliptic_dirichlet_2}\\
	\partial x / \partial t_2 &= 0 	& \;  t_2 &\in (0,1) 	&\; t_1 &= 0, 1 \label{eq:elliptic_neumann}
\end{align}
A model solution to this system, generated with a finite-element method on a fine mesh, is shown in Figure~\ref{fig:elliptic_model_solution}.

\begin{figure}
	\centering
	\includegraphics[width=0.5\textwidth]{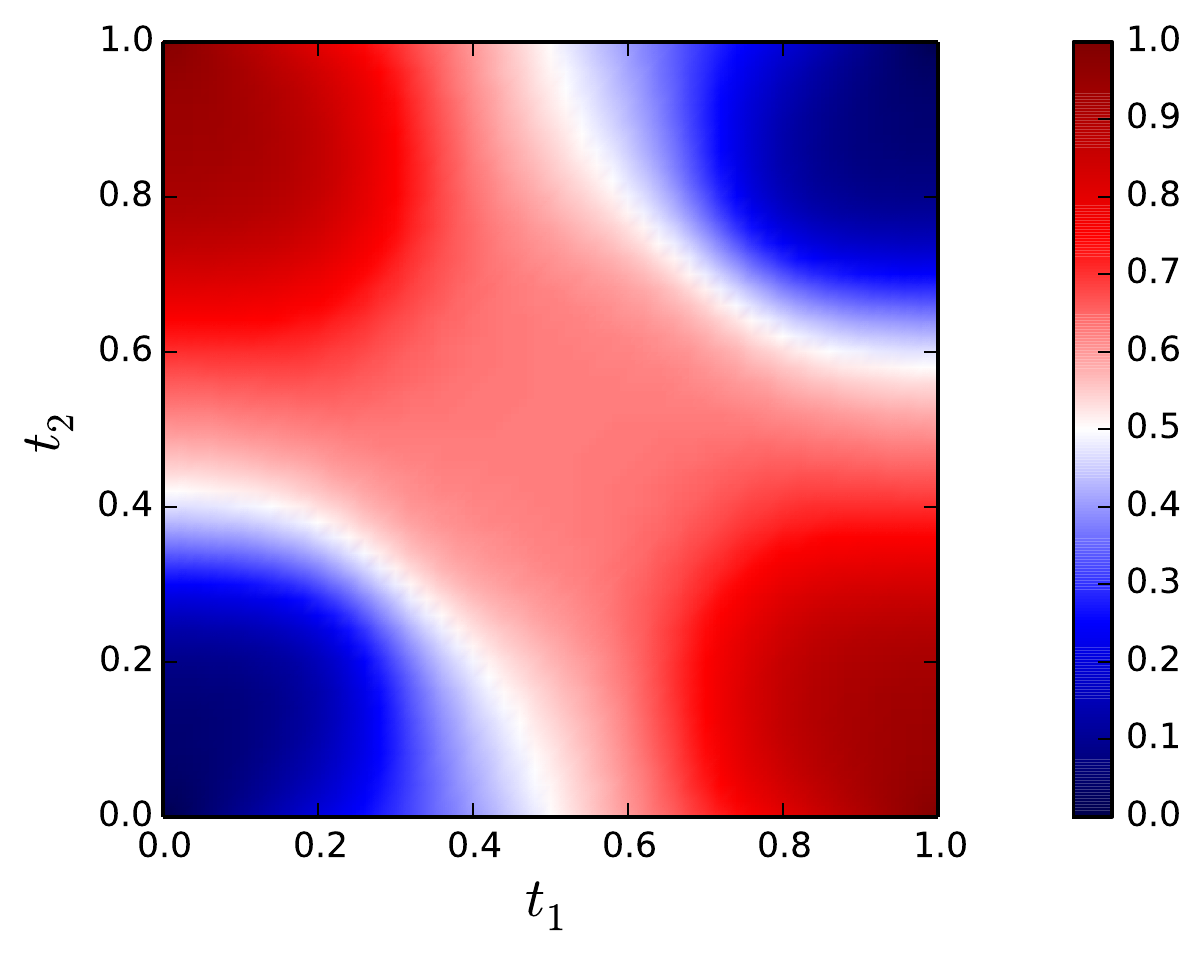}
	\caption{Model solution $x(t)$, $t = (t_1,t_2)$, generated by application of a finite element method based on a triangular mesh of $50\times50$ elements.} \label{fig:elliptic_model_solution}
\end{figure}

As the spatial domain for this problem is two-dimensional, the basis used for specification of the belief distribution is more complex. Here tensor products of orthogonal polynomials have been used:
$\phi_i(t) = C_j(2t_1-1)C_k(2t_2-1)$, $i+j \leq N_c$.
The polynomials $C_i$ were chosen to be normalised Chebyshev polynomials of the first kind. 
Prior specification then follows the formulation given in Section \ref{sec:rcp_prior}, where the remaining parameters were chosen to be $x_0\equiv 1$, and $\gamma_i = \alpha(i+1)^{-2}$. 
The random variables $\xi$ were taken to be either Gaussian or Cauchy and the polynomial basis was truncated to $N=45$ terms, corresponding to a maximum polynomial degree of $N_C=8$.
For both priors the parameter $\alpha$ was set to $\alpha=1$.
Note that closed-form expressions are available for analysis under the Gaussian prior \citep{Cockayne:2016ts} but, to simplify interpretation of empirical results, were not exploited.
Mathematical background on Cauchy priors can be found in \cite{Sullivan:2016wp}.

The information operator was defined by a set of locations $t_i \in [0,1]^2$, $i=0,\dots,N_t$, where either the interior condition or one of the boundary conditions was enforced. 
Denote by $\set{t^{I,i}}$ the set of interior points, $\set{t^{D,j}}$ the set of Dirichlet boundary points and $\set{t^{N,k}}$ the set of Neumann boundary points, where $i=1,\dots,N_I$, $j=1,\dots,N_D$ and $k=1,\dots,N_N$, with $n = N_I + N_D + N_N$. Then, the information operator is given by the concatenation of the conditions defined above:
\begin{equation*}
	A(x) = [A^I(x)^\top, A^D(x)^\top, A^N(x)^\top]^\top, \\
\end{equation*}
\begin{equation*}
	A^I(x) = \begin{bmatrix} -\nabla^2 x(t^{I,1}) \\ \vdots \\ -\nabla^2 x(t^{I,N_I}) \end{bmatrix} ,
	\quad A^D(x) = \begin{bmatrix} x(t^{D,1}) \\ \vdots \\ x(t^{D,N_D}) \end{bmatrix} ,
	\quad A^N(x) = \begin{bmatrix} \frac{\partial}{\partial t_1} x(t^{N,1}) \\\vdots \\ \frac{\partial}{\partial t_1} x(t^{N,N_N}) \end{bmatrix}
\end{equation*}

The Bayesian PNM output was approximated by numerical disintegration and sampled with a Monte Carlo method whose description is reserved for the Electronic Supplement.
In Figure~\ref{fig:elliptic_posteriors} the mean and pointwise standard-deviations of the posterior distributions are plotted for Gaussian and Cauchy priors with $n = 16$. 
There is little qualitative difference between the posterior distributions for the Gaussian and Cauchy priors. 
The mean functions match closely to the mean function from the model solution, as given in Figure~\ref{fig:elliptic_model_solution}.
The posterior variance is lowest near to the Dirichlet boundaries where the solution is known, and peaks where the Neumann condition is imposed. This is to be expected, as evaluations of the Neumann boundary condition provide less information about the solution itself.



\begin{figure}
	\begin{subfigure}{\textwidth}
	\centering
		\includegraphics[width=0.45\textwidth]{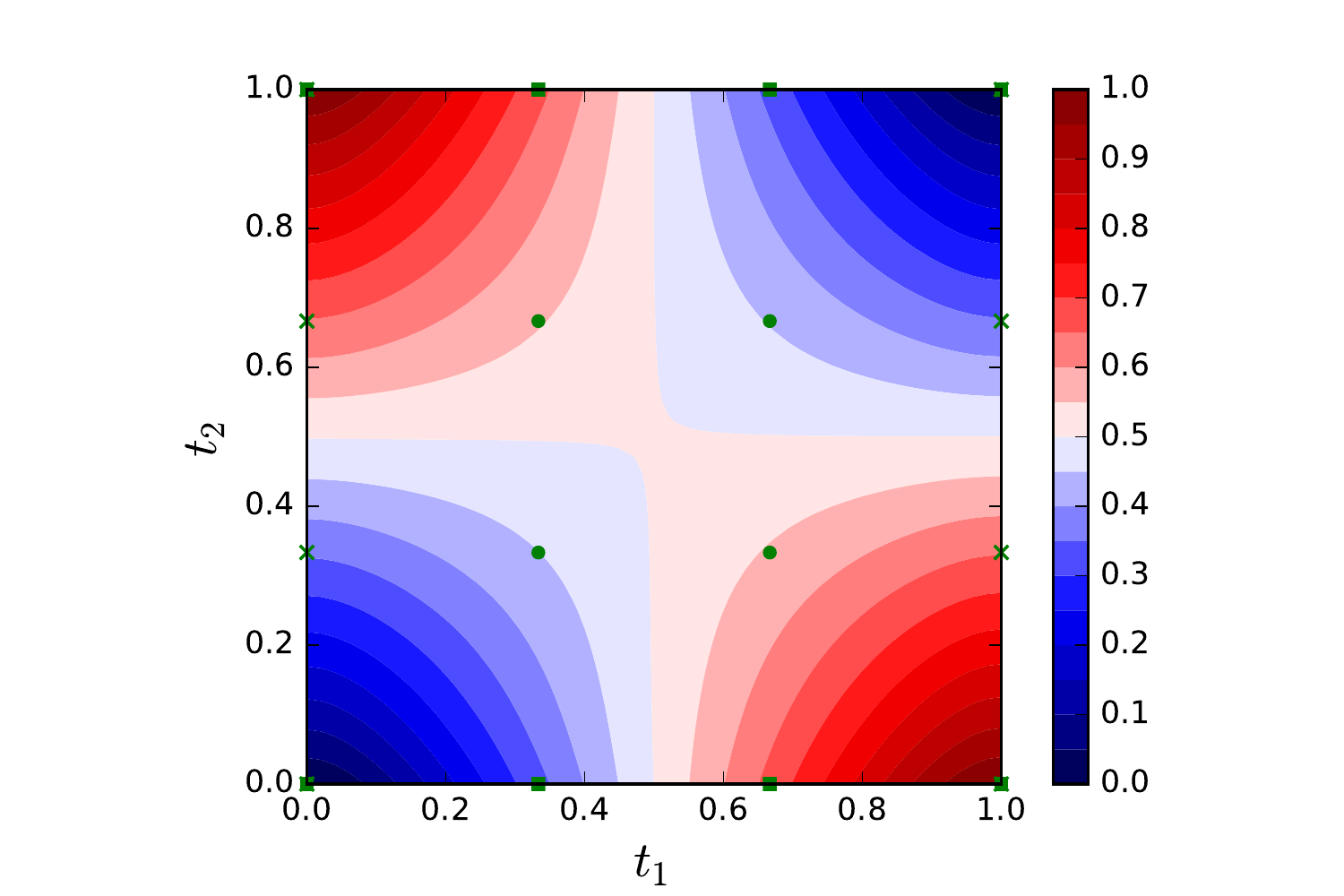}
		\includegraphics[width=0.45\textwidth]{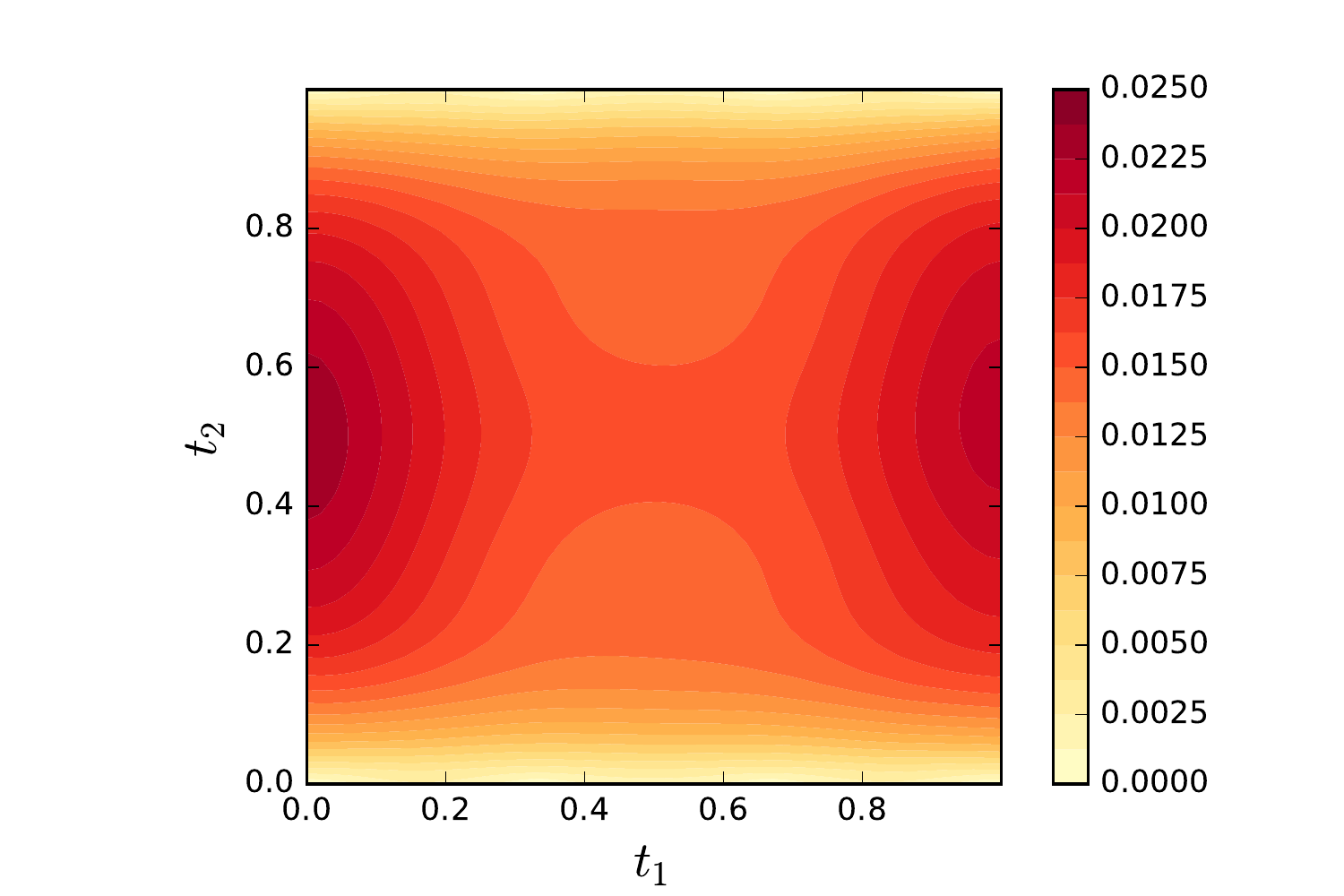}
		\caption{Gaussian prior}
	\end{subfigure}
	\begin{subfigure}{\textwidth}
	\centering
		\includegraphics[width=0.45\textwidth]{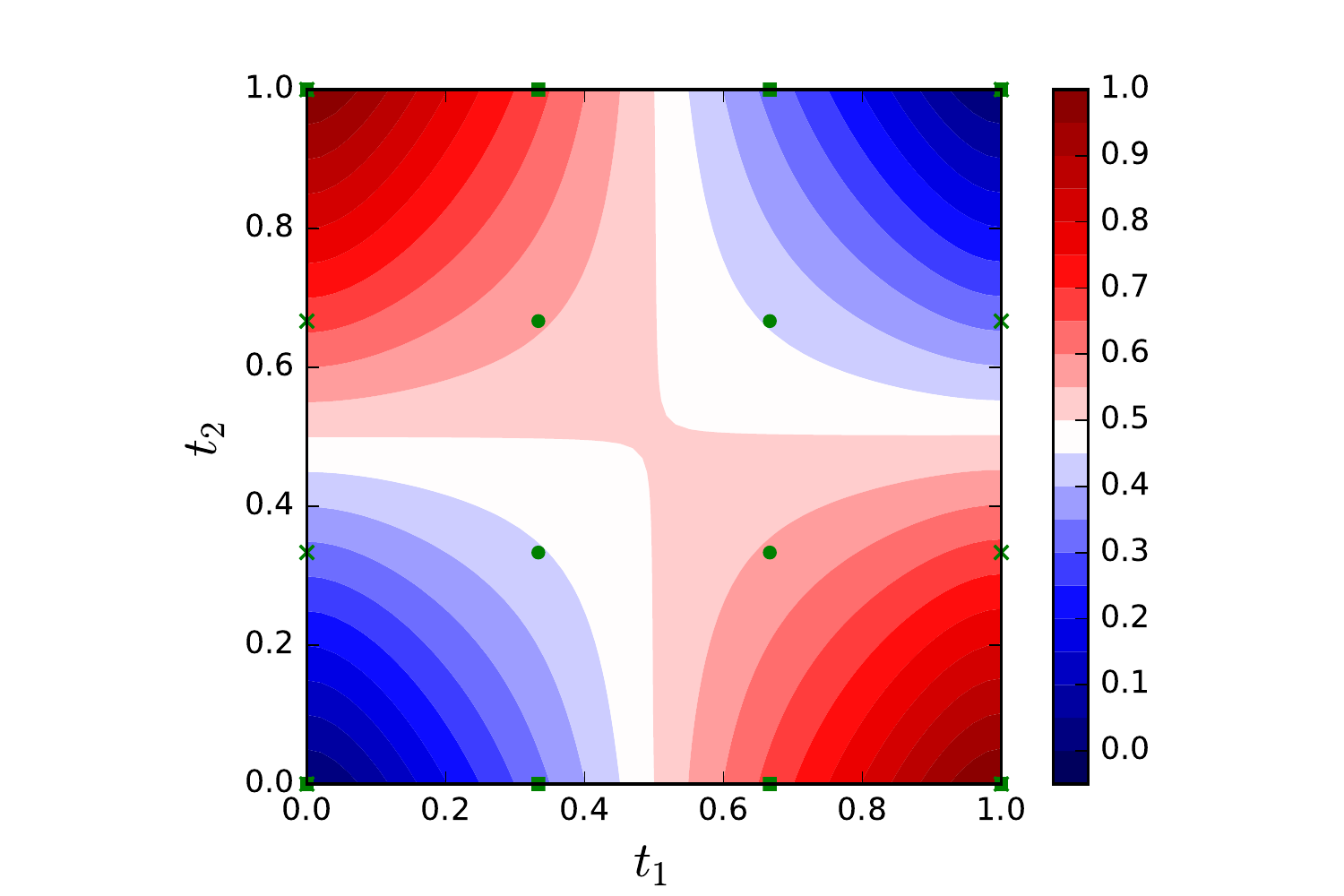}
		\includegraphics[width=0.45\textwidth]{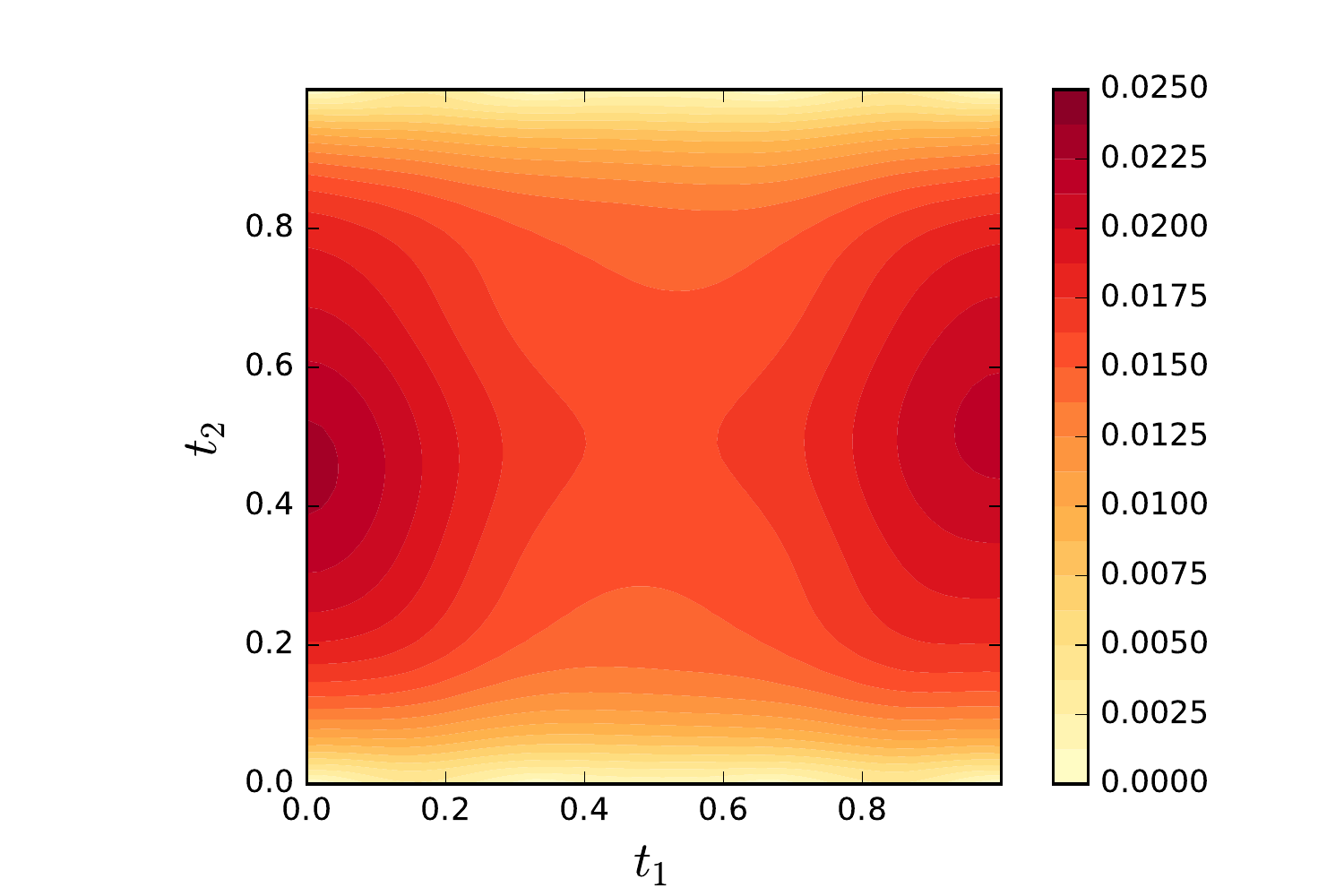}
		\caption{Cauchy prior}
	\end{subfigure}
	\caption{Posterior distributions for the solution $x$ of the Poisson equation, with $n=16$ and different choices of prior distribution. 
	Left: Posterior mean.
	Design points for the interior, Dirichlet and Neumann boundary conditions are indicated by green dots, green squares and green crosses, respectively.
	Right: Posterior standard deviation.
	}
	\label{fig:elliptic_posteriors}
\end{figure}

Next, the posterior distribution of the spectrum $\{u_i\}$ was investigated.
In Figure~\ref{fig:elliptic_spectrum} the posterior distribution over these coefficients is plotted and it is seen that the correlation structure between coefficients is non-trivial, c.f. the joint distribution between $u_0$ and $u_3$.

\begin{figure}
	\includegraphics[width = \textwidth,clip,trim = 0cm 2cm 0cm 1.5cm]{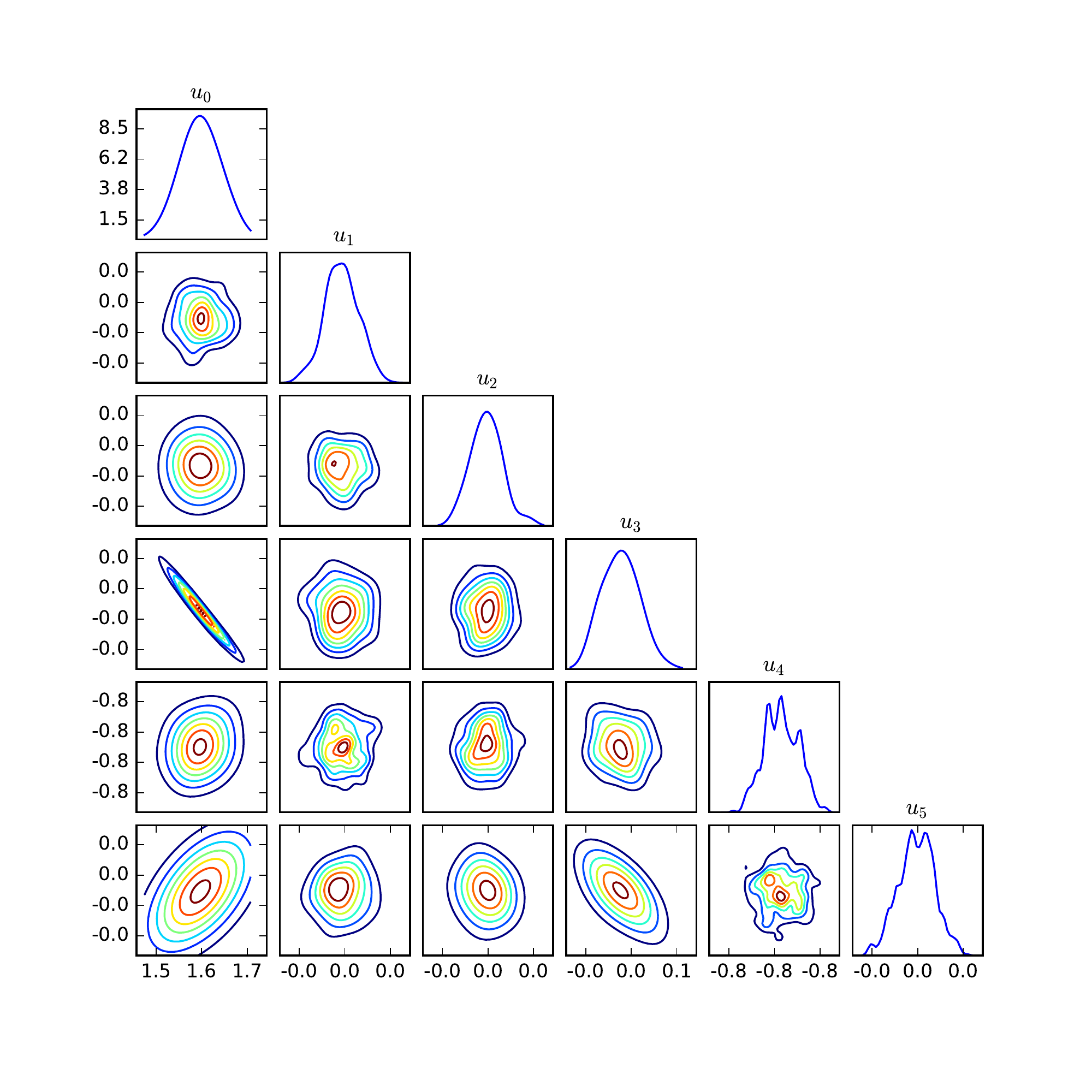}
	\caption{Posterior distributions for the first six coefficients of the spectrum for the solution $x$ of the Poisson equation, obtained with Monte Carlo methods and numerical disintegration, based on $\delta=0.0008$, $n=16$.
	(NB: The posterior is Gaussian and can be obtained in closed-form, but we opted to additionally illustrate the Monte Carlo method.)
	} \label{fig:elliptic_spectrum}
\end{figure}

Last, in Figure~\ref{fig:elliptic_variance} convergence of the posterior distribution is plotted as the number of design points is varied, for $n=16,25,36$. In each case a Gaussian prior was used. As expected, the standard deviation in the posterior distribution is seen to decrease as the number of design points is increased. At $n=36$, the shape of the region of highest uncertainty changes markedly, with the most uncertain region lying between the Dirichlet boundary and the first evaluation points on the Neumann boundary. This is likely due to the fact that the number of evaluation points is approaching the size of the polynomial basis; when the number of points equals the size of the basis the system is completely determined for a linear model. 
Thus, we need $N \gg n$ in order for discretisation error to be quantified.

\begin{figure}
\centering
	\begin{subfigure}{0.3\textwidth}
		\includegraphics[width=\textwidth,clip,trim = 1cm 0cm 1cm 0cm]{figures_elliptic/gaussian_n=4_variance.pdf}
		\caption{$n=16$}
	\end{subfigure}
	~
	\begin{subfigure}{0.3\textwidth}
		\includegraphics[width=\textwidth,clip,trim = 1cm 0cm 1cm 0cm]{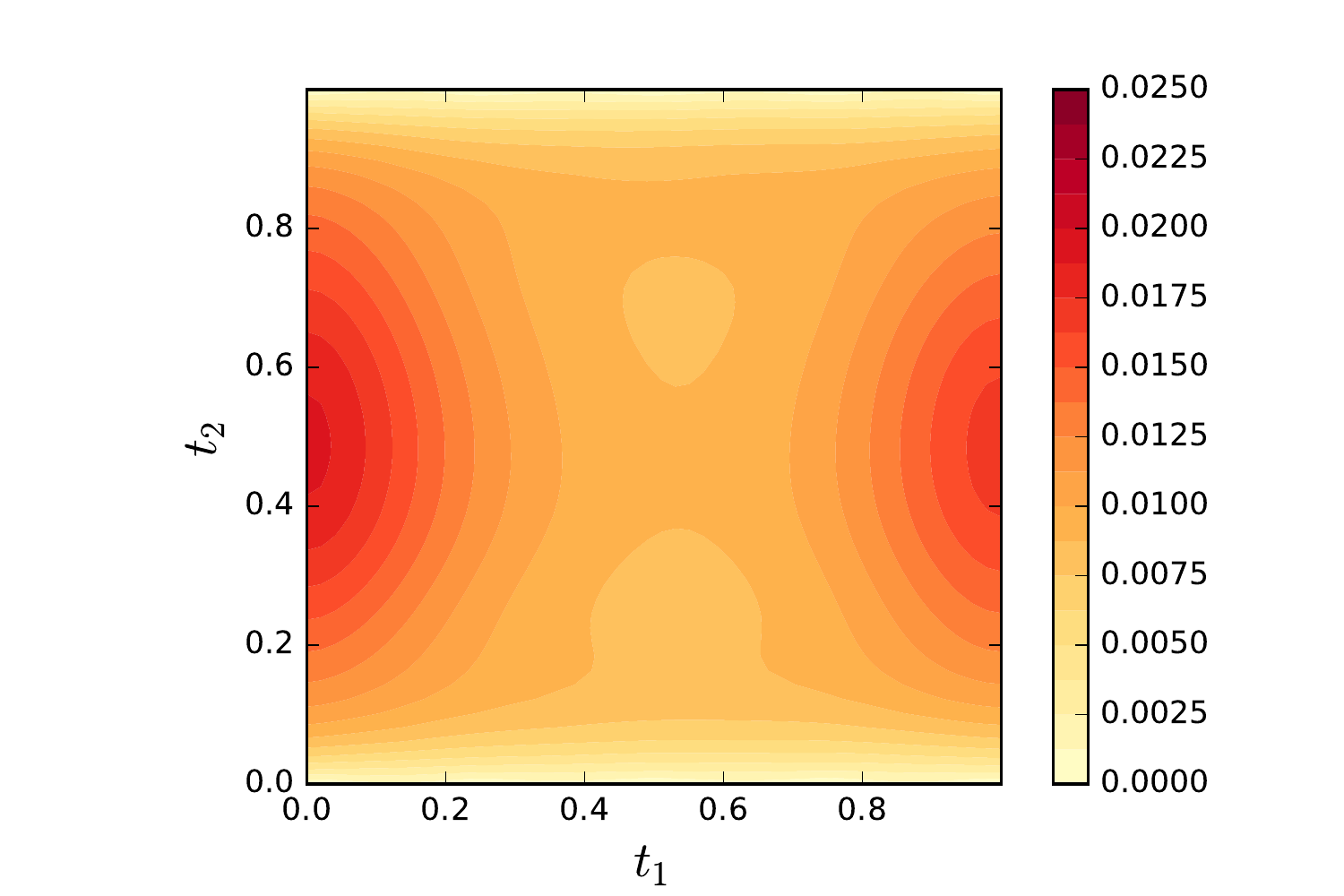}
		\caption{$n=25$}
	\end{subfigure}
	~
	\begin{subfigure}{0.3\textwidth}
		\includegraphics[width=\textwidth,clip,trim = 1cm 0cm 1cm 0cm]{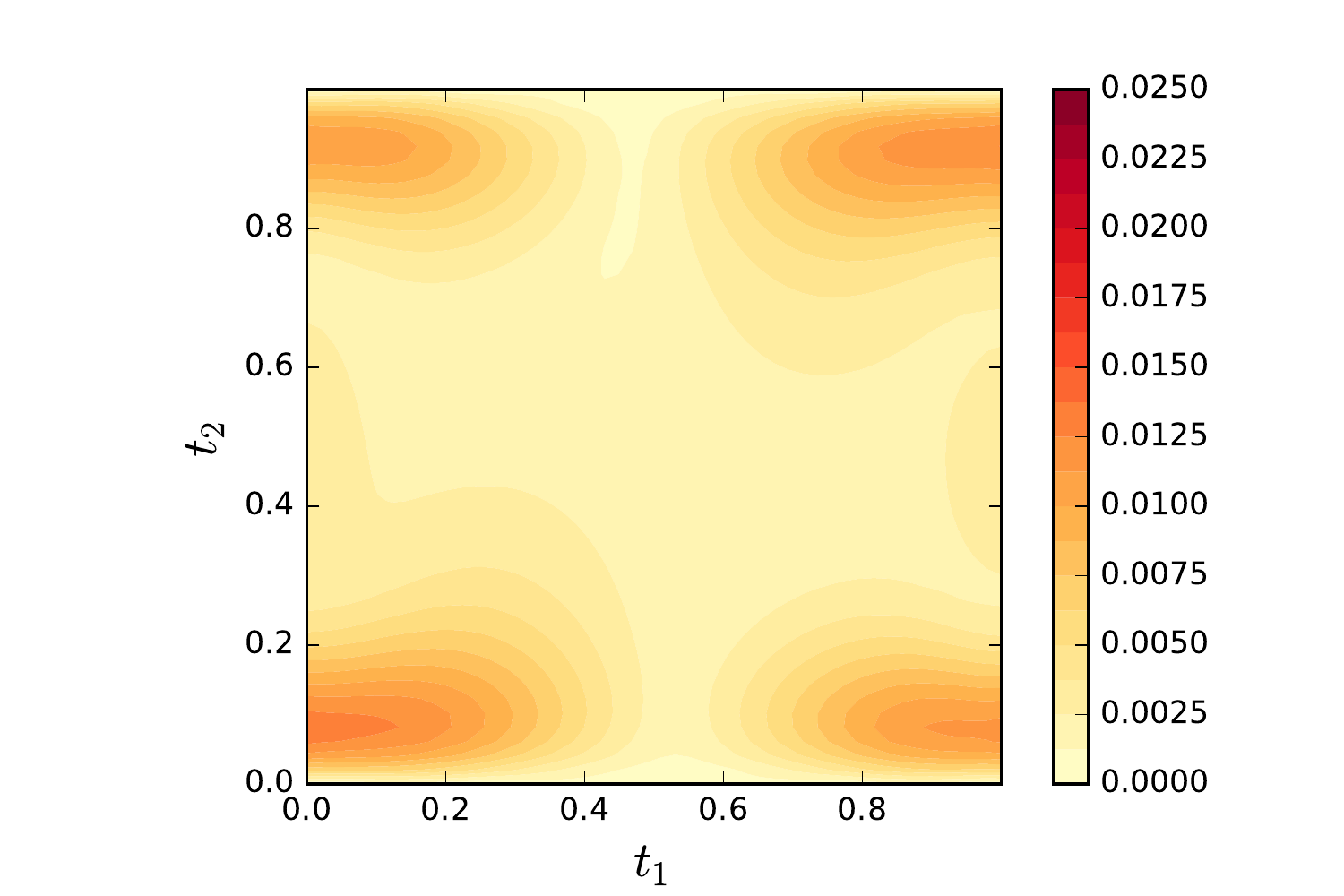}
		\caption{$n=36$}
	\end{subfigure}
	\caption{Heat map of the point-wise standard deviation for the solution $x$ to the Poisson equation as the number $n$ of design points is varied. In each case a Gaussian prior has been used.
	} \label{fig:elliptic_variance}
\end{figure}

\subsection{The Painlev\'{e} ODE}\label{sec:painleve}


In this section a Bayesian PNM is developed to solve a nonlinear ODE based on Painlev\'{e}'s first transcendental
\begin{align*}
	x'' &= x^2 - t, \hspace{20pt} t \in [0,\infty)
	x(0) &= 0 \\
	t^{-1/2} x(t) &\to 1 \hspace{20pt} \text{ as } t \to \infty \;.
\end{align*}
To permit computation, the right-boundary condition was relaxed by truncating the domain to $[0, 10]$ and using the modified condition $x(10) = \sqrt{10}$.

Two distinct solutions are known, illustrated in Figure~\ref{fig:painleve_solutions} (left). 
These model solutions were obtained using the deflation technique described in \cite{Funke2013}. 
The spectrum plot in Figure~\ref{fig:painleve_solutions} (right) represents the coefficients $\{u_i\}$ obtained when each solution is represented over a basis of normalised Chebyshev polynomials. 
As those polynomials are orthonormal with respect to the $L_2$-inner-product, the slower decay for the negative solution compared to the positive solution is equivalent to the negative solution having a larger $L_2$-norm. 
This explains the preference that optimisation-based numerical solvers have for returning the positive solution in general, and also explains some of the results now presented.

\begin{figure}
	\includegraphics[width=\textwidth]{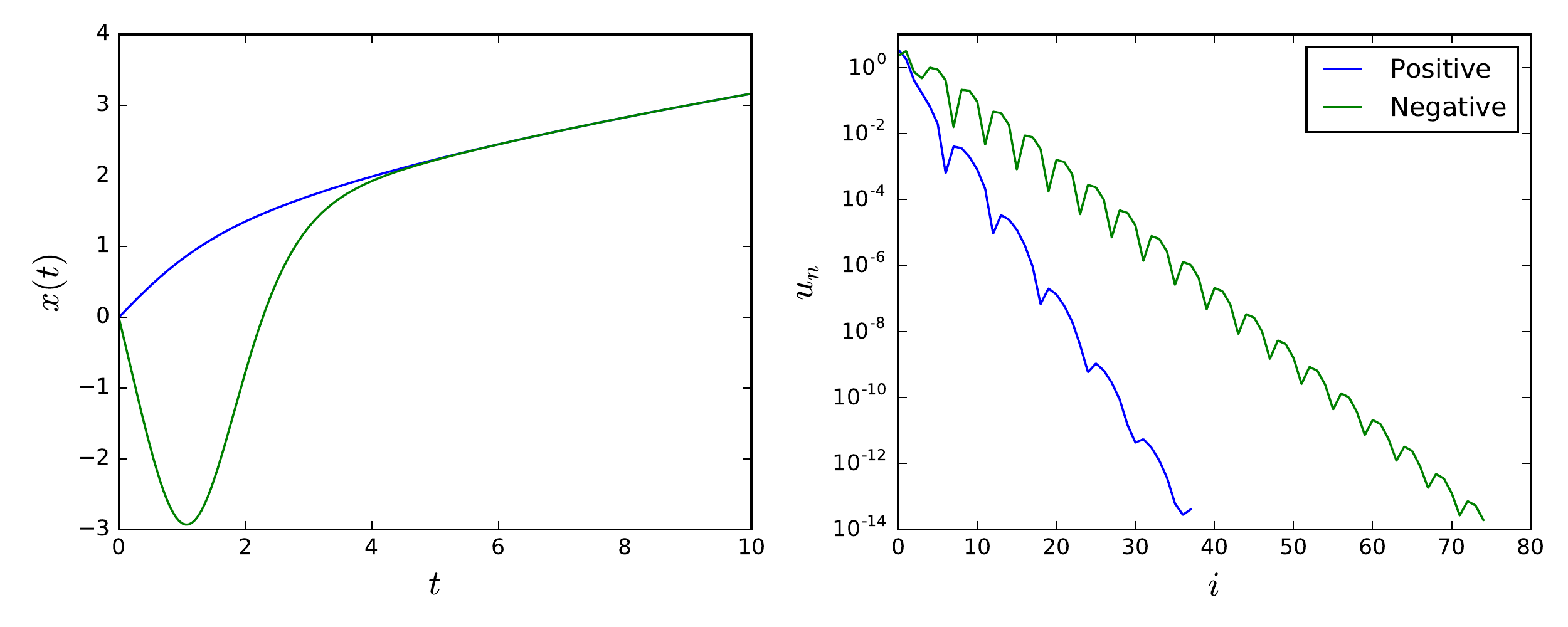}
	\caption{Two distinct solutions for the Painlev\'{e} ODE. The spectral plot on the right shows the true coefficients $\{u_i\}$, as determined by a model solver (the MatLab package \texttt{chebfun}).
	} \label{fig:painleve_solutions}
\end{figure}

Such systems for which multiple solutions exist have been studied before in the context of PNM, both in \cite{Chkrebtii:2013ux} and in \cite{Cockayne:2016ts}.
It was noted in both papers that existence of multiple solutions can present a substantial challenge to classical numerical methods.

To build a Bayesian PNM, a prior $\mu$ for this problem was defined by using a series expansion as in Eq. (\ref{eq:generic_random_field}).
The basis functions were $\phi_i(t) = C_i(\frac{1}{2}(t-5))$ where the $C_i$ were normalised Chebyshev polynomials of the first kind.
Both Gaussian and Cauchy priors were considered by taking $u_i := \gamma_i \xi_i$, where $\xi_i$ were taken to be either standard Gaussian or standard Cauchy and in in each case $x_0(t) \equiv 0$.
In accordance with the exponential convergence rate for spectral methods when the solution to the system is a smooth function, the sequence of scale parameters was set to $\gamma_i = \alpha \beta^{-i}$, where $\alpha=8$ and $\beta=1.5$. These values were chosen by inspection of the true spectra (obtained with Matlab's ``\texttt{chebfun}'' package) to ensure that both solutions were in the support of the prior.

The information operator $A$ was defined by the choice of locations $\set{t_j}$, $j=1,\dots, m$, which determine the locations at which the posterior will be constrained.
Analysis for several values of $m$ was performed. In each case $t_1 = 0$, $t_m = 10$ and the remaining $t_j$ were equally spaced on $[0,10]$. To be explicit, the information operator was
\begin{equation*}
	A(x) = \begin{bmatrix} 
		x''(t_1) - (x(t_1))^2 \\ 
		\vdots \\ 
		x''(t_m) - (x(t_m))^2 \\
		x(0) \\
		x(10)
	\end{bmatrix}
\end{equation*}
with the last two elements enforcing the boundary conditions. Thus our information was $a = [-t_1, \dots, -t_m, 0, \sqrt{10}]$, which is $n = m+2$ dimensional.

The Bayesian PNM output $B(\mu,a)$ was approximated via numerical disintegration with the first $N = 40$ terms of the series representation used.
This was sampled with Monte Carlo methods, the details of which are reserved for the Electronic Supplement.

Results for a selection of bandwidths $\delta$, with $n = 17$, are shown in Figure~\ref{fig:results_multiple_delta}. 
Note that a strong preference for the positive solution is expressed at the smallest $\delta$, with mass around both solutions at larger $\delta$. 
For the Gaussian prior, some mass remained around the negative solution at the smallest $\delta$, while this was not so for the Cauchy prior. 
This reflects the fact that, for a collection of independent univariate Cauchy random variables, one element is likely to be significantly larger in magnitude than the others, which favours faster decay for the remaining elements.

Using the calculation described in Section~\ref{sec:evidence_computation}, model evidence was computed for both the Gaussian and the Cauchy prior at $n=15$. The Bayes factor for the Cauchy, compared to the Gaussian prior, was found to be $20.26$, which constitutes strong evidence in favour of a Cauchy prior for this problem at the given level of discretisation.

\begin{figure}
\centering
	\begin{subfigure}{0.75\textwidth}
	\includegraphics[width=\textwidth]{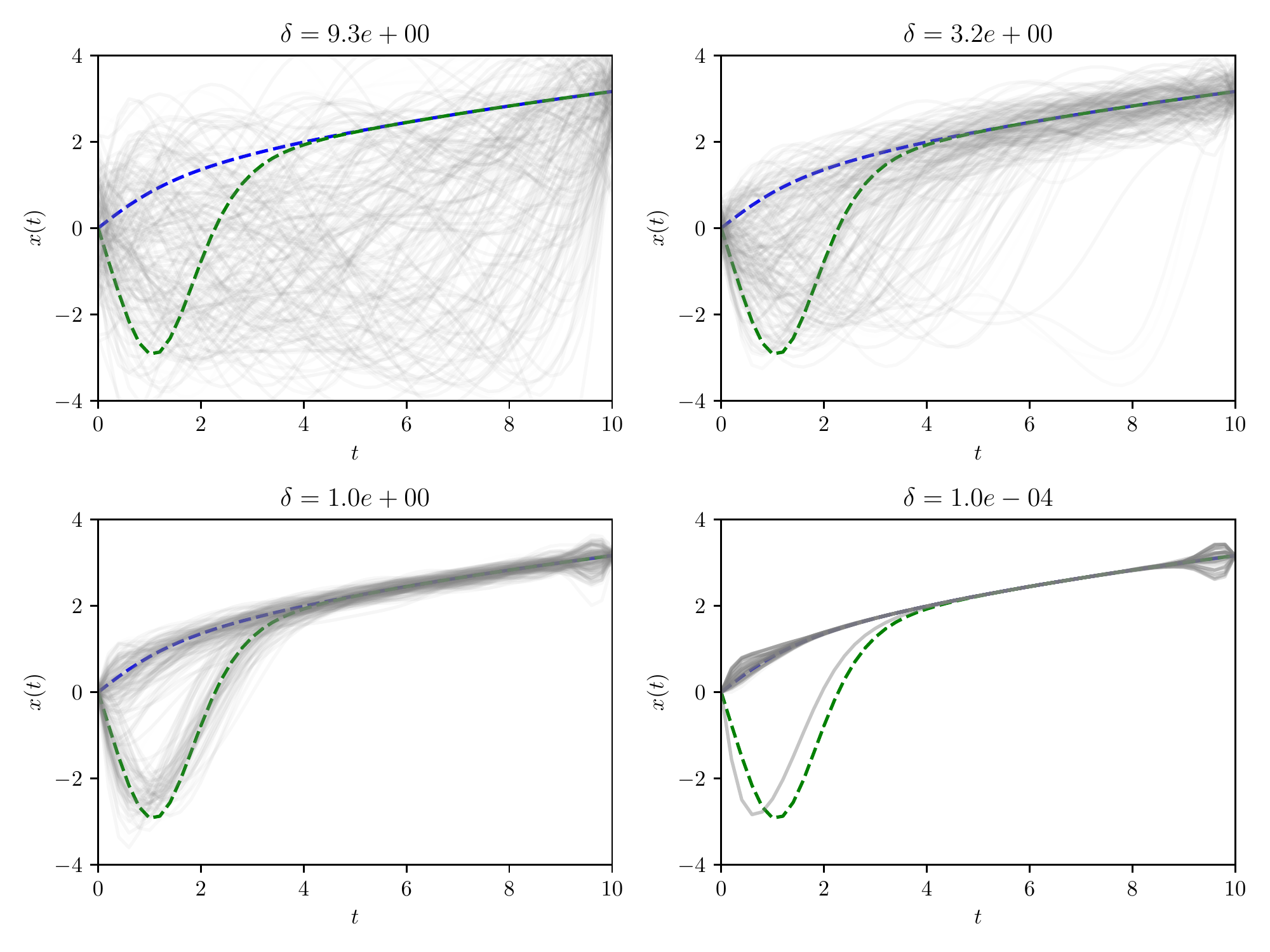}
	\caption{Gaussian Prior}
	\end{subfigure}
	\begin{subfigure}{0.75\textwidth}
	\includegraphics[width=\textwidth]{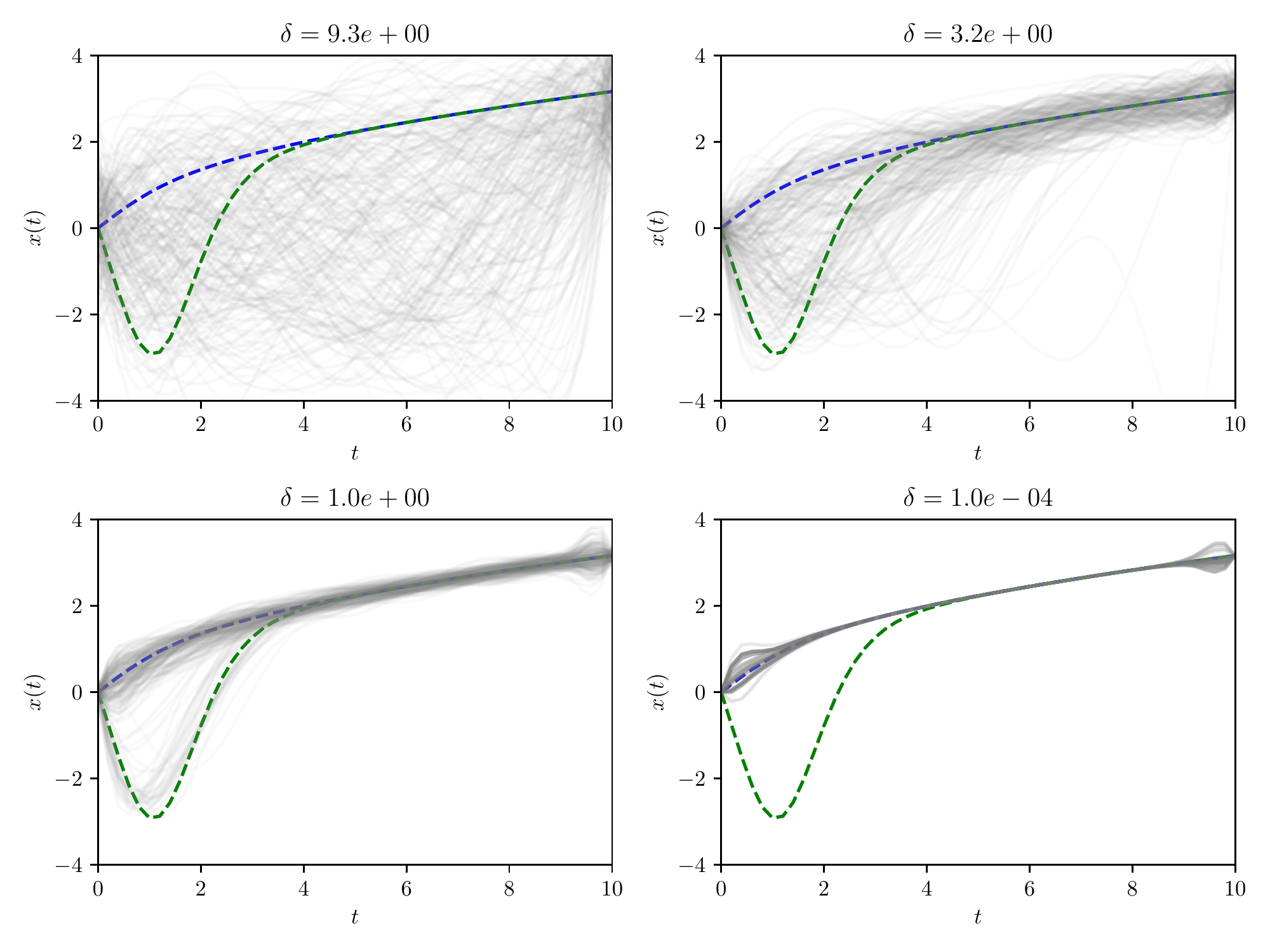}
	\caption{Cauchy Prior.}
	\end{subfigure}
	\caption{Posterior samples for the Painlev\'{e} system for $n = 17$. 
	Blue and green dashed lines represent the positive and negative solutions determined with \texttt{chebfun}.
	Grey lines are samples from an approximation to the posterior provided by numerical disintegration (bandwidth parameter $\delta$).} \label{fig:results_multiple_delta}
\end{figure}

In Figure~\ref{fig:painleve_coeffs} the posterior distributions for first six coefficients $u_i$ at $n = 17$ and $\delta=1$ are plotted. Strong multimodality is clear, as well as skewed correlation structure between the coefficients. 
Illustration of such posteriors for smaller $\delta$ is difficult as the posteriors become extremely peaked.

\begin{figure}
	\includegraphics[width=\textwidth,clip, trim = 0cm 2cm 0cm 1.5cm]{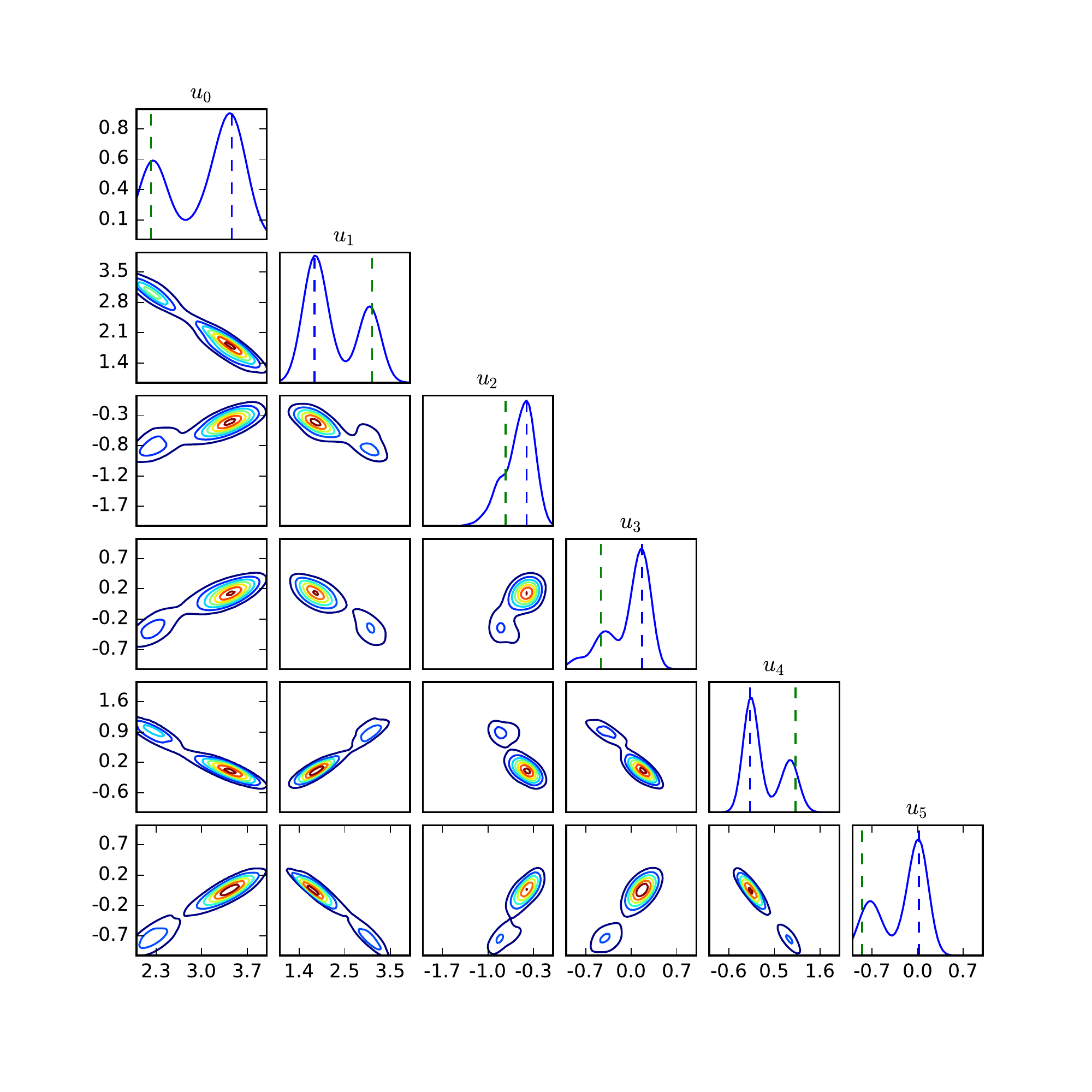}
	\caption{Posterior distributions for the first six coefficients obtained with numerical disintegration (bandwidth parameter $\delta = 1$), at $n = 17$. 
	Vertical dashed lines on the diagonal plots indicate the value of the coefficients for the positive (blue) and negative (green) solutions determined with \texttt{chebfun}.} \label{fig:painleve_coeffs}
\end{figure}

Figure~\ref{fig:painleve_convergence} displays convergence of the posterior distributions as $n$ is increased. Of particular interest is that for $n = 12$, the posterior distribution based on a Gaussian prior becomes trimodal.
For each prior, the posterior mass settles on the positive solution to the system at $n = 22$. This is in accordance with the fact that this solution has smaller $L_2$-norm. 
This perhaps reflects the fact that, while in the limiting case both solutions should have an equal likelihood, the curvature of the likelihood at each mode may differ. 
Prior truncation may also be influential; in Figure~\ref{fig:painleve_truncation_likelihoods} the log-likelihood of the negative solution increases at a slower rate than that of the positive solution.
Thus, while in the setting of an infinite prior series neither solution should be preferred, in practice truncation might bias one solution over the other. Lastly, it is clear that the parameters $\alpha$ and $\beta$ may also have a significant effect on which solution is preferred.
Further theoretical work will be required to understand many of the phenomena that we have just described.

\begin{figure}
	\includegraphics[width=\textwidth]{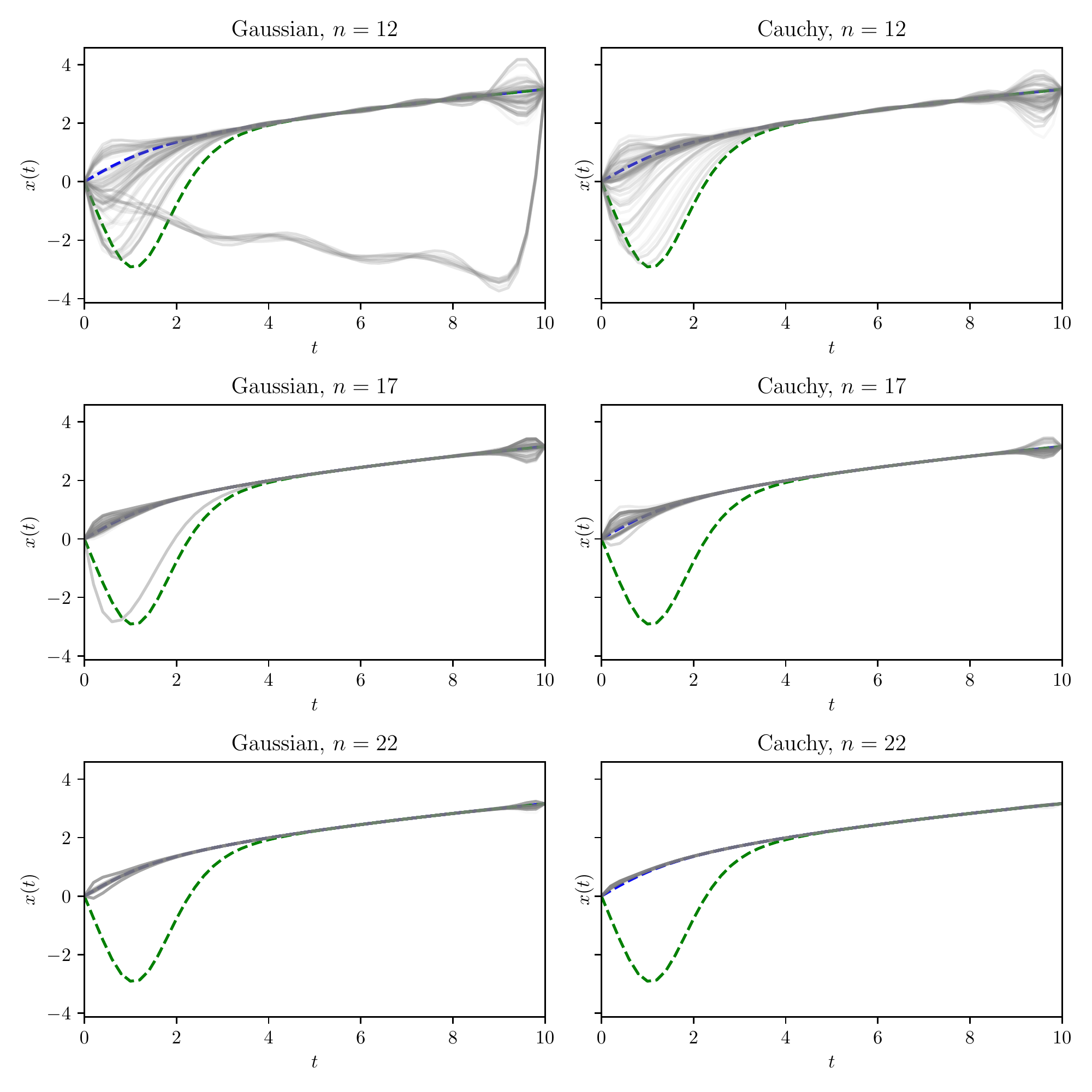}
	\caption{Convergence for the numerical disintegration scheme as $n$ is increased.
	Left: Gaussian prior. Right: Cauchy prior.
	 In all cases $\delta=10^{-4}$.}
	\label{fig:painleve_convergence}
\end{figure}

\begin{figure}
	\centering 
	\includegraphics[width=0.6\textwidth]{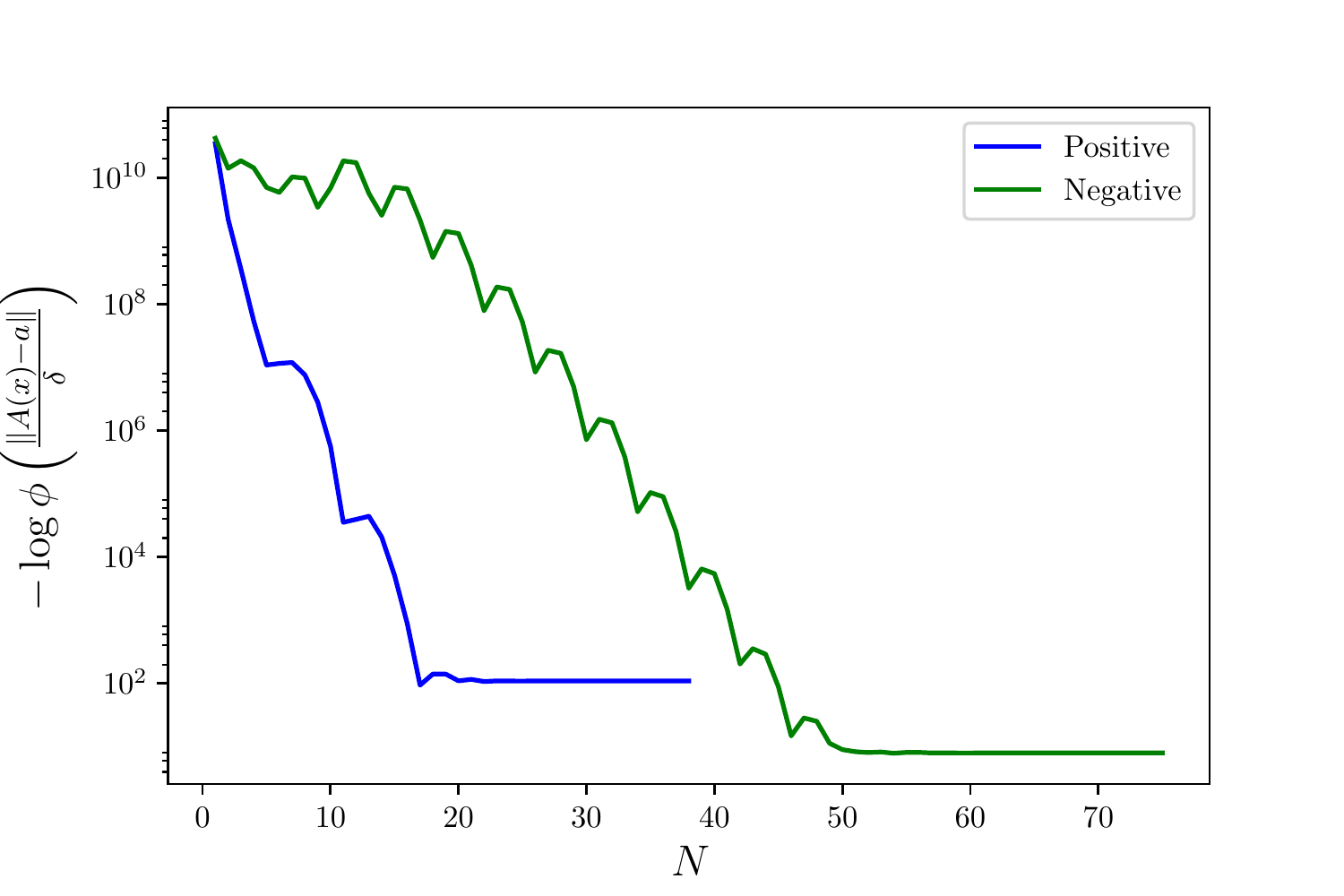}
	\caption{Negative-log-likelihoods for the point-estimates of coefficients for the postive and negative solutions given by \texttt{chebfun}, as the truncation level $N$ is varied. The fact that the likelihood for the positive solution decreases more rapidly than that of the negative solution suggests indicates that the posterior may have a preference for that solution over the other, though the level $N=40$ has been selected in an attempt to minimise the impact.} \label{fig:painleve_truncation_likelihoods}
\end{figure} 

Of particular interest is how a preference for the negative solution could be encoded into a PNM. 
Owing to the flexible specification the information operator, there is considerable choice in this matter. 
An elegant approach is the introduction of additional, inequality-based information
\begin{equation}
	x'(0) \leq 0 \;.
	\label{eq:negative_bc} 
\end{equation}
Such information can be difficult to incorporate in standard numerical algorithms, but is of interest in many physical problems \citep{Kinderlehrer:2000we}.
For Bayesian PNM we can extend the information operator to include $1[x'(0) \leq 0]$.
Posterior distributions for the Gaussian prior at $n = 17$ are shown in Figure~\ref{fig:painleve_negative_convergence}. Note that posterior mass has settled close to the negative solution. This highlights the simplicity with which Bayesian PNMs can encode a preference for a particular solution when a multiplicity of solutions exist.

\begin{figure}
	\includegraphics[width=\textwidth]{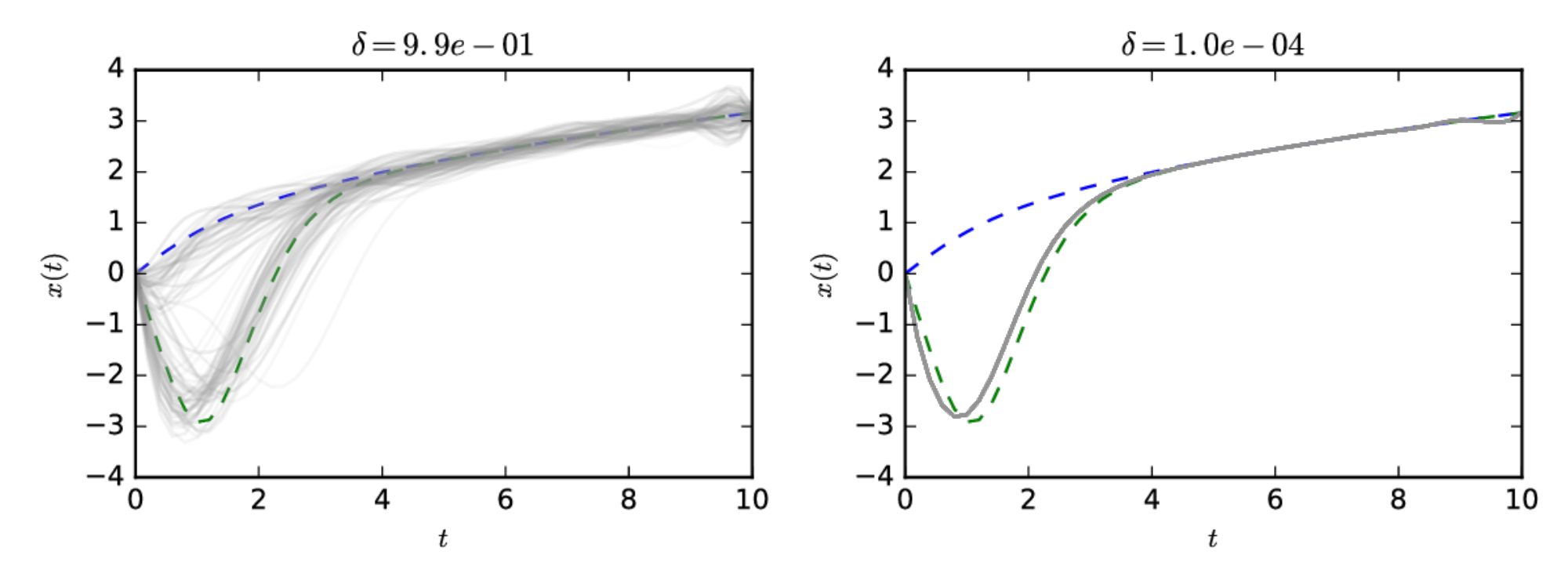}
	\caption{Posterior distribution at $n=17$, based on a Gaussian prior, with the negative boundary condition given by Eqn.~\eqref{eq:negative_bc} enforced. 
Left: $\delta = 0.99$.
Right: $\delta = 0.0001$.} \label{fig:painleve_negative_convergence}
\end{figure}

\subsection{Application to Industrial Process Monitoring}


This final application illustrates how statistical models for discretisation error can be propagated through a pipeline of computation to model how these errors are accumulated.

Hydrocyclones are machines used to separate solid particles from a liquid in which they are suspended, or two liquids of different densities, using centrifugal forces. High pressure fluid is injected into the top of a tank to create a vortex. The induced centrifugal force causes denser material to move to the wall of the tank while lighter material concentrates in the centre, where it can be extracted. They have widespread applications, including in areas such as environmental engineering and the petrochemical industry \citep{Sripriya2007}. An illustration of the operation is given in Figure~\ref{fig:hydrocyclone_sketch}.

\begin{figure}
\begin{subfigure}[b]{0.45\textwidth}
\centering
\begin{tikzpicture}
  \coordinate (O) at (0,-1);

  \begin{scope}
    \def\rx{0.31}
    \def\ry{0.15}
    \def\z{1.45}

    \path [name path = ellipse]    (0,\z) ellipse ({\rx} and {\ry});
    \path [name path = horizontal] (-\rx,\z-\ry*\ry/\z) -- (\rx,\z-\ry*\ry/\z);
    \path [name intersections = {of = ellipse and horizontal}];

    \draw[fill = gray!50, gray!50] (intersection-1) -- (0,0.5)
      -- (intersection-2) -- cycle;
    
    \draw[fill = gray!30, densely dashed] (0,\z) ellipse ({\rx} and {\ry});
  \end{scope}

  \draw (0.25,0.4) -- (0.9,0.1) node at (1.8,0.0) {more dense};
  \draw (0,0.9) -- (-0.9,0.1) node at (-1.8,0.0) {less dense};

  \filldraw (O) circle (1pt) node[below] {underflow};

  \draw[] (O) to (-1.33,1.33);
  \draw[] (O) -- (1.33,1.33);

  \draw[black, densely dashed] (-1.36,1.46) arc [start angle = 170, end angle = 10,
    x radius = 13.8mm, y radius = 3.6mm];
  \draw[black] (-1.29,1.52) arc [start angle=-200, end angle = 20,
    x radius = 13.75mm, y radius = 3.15mm];

  \draw (-1.2,2.2) -- (-0.1,1.5) node at (-1.37,2.37) {overflow};
\end{tikzpicture}
\caption{Hydrocyclone tank schematic}
\end{subfigure}
\begin{subfigure}[b]{0.45\textwidth}
\centering
\begin{tikzpicture}

\draw[ball color=blue,shading=ball, opacity = 0.2,line width=4pt] (2,0) to [out = 180, in = 0] (0,0) to [out = 180, in = 90] (-1,-1) to [out = -90, in = 180] (0,-2) to [out = 0, in = -90] (1,-1) to [out = 90, in = -45] (0.7071,-0.2929) to [out = 0, in = 180] (2,-0.2929);

\node (n1) at (2,-0.15) {};
\node (n2) at (3.5,-0.15) {input flow};
\path[->] (n2) edge (n1);

\node (n3) at (-0.5,-1) {};
\node (n4) at (0,-1.5) {};
\path[->] (n3) edge [bend right = 45] (n4);

\node (n5) at (0,-1.5) {};
\node (n6) at (0.5,-1) {};
\path[->] (n5) edge [bend right = 45] (n6);

\node (n7) at (0.5,-1) {};
\node (n8) at (0,-0.5) {};
\path[->] (n7) edge [bend right = 45] (n8);

\node (n9) at (0,-0.5) {};
\node (n10) at (-0.5,-1) {};
\path[->] (n9) edge [bend right = 45] (n10);

\end{tikzpicture} \vspace{20pt}
\caption{Cross-section (top of tank)}
\end{subfigure}
\caption{A schematic description of hydrocyclone equipment.
(a) The tank is cone-shaped with overflow and underflow pipes positioned to extract the separated contents.
(b) Fluid, a mixture to be separated, is injected at high pressure at the top of the tank to create a vortex.
Under correct operation, denser materials are directed toward the centre of the tank and less-dense materials are forced to the peripheries of the tank.}
\label{fig:hydrocyclone_sketch}
\end{figure}
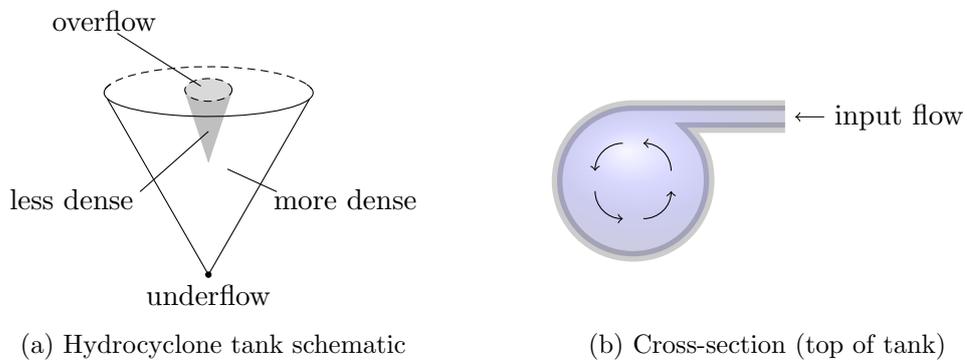 

To ensure the materials are well-separated the hydrocyclone must be moitored to allow adjustment of the input flow-rate. This is also important for safe operation, owing to the high pressures involved \citep{Bradley2013}. However, direct monitoring is impossible owing to the opaque walls of the equipment and the high interior pressure. For this purpose electrical impedance tomography (EIT) has been proposed to allow monitoring of the contents \citep{Gutierrez2000}. 

EIT is a technique which allows recovery of an interior conductivity field based upon measurements of voltage obtained from applying a stimulating current on the boundary. It is suited to this problem, as the two materials in the hydrocyclone will generally be of different conductivities. In its simplified form due to \cite{Calderon1980}, EIT is described by a linear partial differential equation similar to that in Section~\ref{sec:poisson}, but with modified boundary conditions to incorporate the stimulating currents and measured voltages:
\begin{eqnarray}
	-\nabla \cdot \left( a(t) \nabla x(t) \right) = 0  & & t \in D \nonumber \\
	a(t) \frac{\partial x}{\partial n}(t) = \left\{ \begin{array}{c} c_e \\ 0 \end{array} \right. & & \begin{array}{l} t = t^e \\ t \in \partial D \setminus \{t^e\}_{e=1}^{N_e} \end{array} \label{eq:sys_eit}
\end{eqnarray}
where $D$ denotes the domain, modelling the hydrocyclone tank, $e$ indexes the stimulating electrodes, $t_e \in \partial D$ are the corresponding locations of the electrodes on $\partial D$, $a$ is the unknown conductivity field to be determined and $\frac{\partial}{\partial n}$ denotes the derivative with respect to the outward pointing normal vector. 
The electrode $t^1$ is referred to as the \emph{reference} electrode. 
The vector $c = (c_1,\dots,c_{N_e})$ denotes the stimulation current pattern.
Several stimulation patterns were considered, denoted $c^j$, $j=1,\dots,N_j$.

The experimental data described in \cite{West2005} were considered. 
In the experiment, a cylindrical perspex tank was used with a single ring of eight electrodes.
Translation invariance in the vertical direction means that the contents are effectively a single 2D region and electrical conductivity can be modelled as a 2D field.
At the start of the experiment, a mixing impeller was used to create a rotational flow.
This was then removed and, after a few seconds, concentrated potassium chloride solution was carefully injected into the tap water initially filling the tank.
Data, denoted $y_\tau$, were collected at regular time intervals by application of several stimulation patterns $c^1,\dots,c^M$.

To formulate the statistical problem, consider parameterising the conductivity field as $a(\tau, t)$, where $\tau \in [0,T]$ is a temporal index while $t \in D$ is the spatial coordinate and $D$ is the circular domain representing the perspex tank in the experiment. 
A log-Gaussian prior was placed over the conductivity field so that $\log a$ is a Gaussian process with separable covariance function $k_a((\tau,t),(\tau',t')) := \lambda \min(\tau, \tau') \exp\left( - \frac{\norm{t - t'}^2}{2 \ell^2} \right)$ where $\ell$ is a length-scale parameter representing the anticipated spatial variation of the conductivity field and $\lambda$ is a parameter controlling the amplitude of the field. 
Here $\ell$ was fixed to $\ell=0.3$, while $\lambda = 10^{-3}$.
The problem of estimating $a$ based on data can be well-posed in the Bayesian framework \citep{Dunlop2015}.
Full details of this experiment can be found in the accompanying report \cite{Oates2017}.

Our aim is to use a PNM to account for the effect of discretisation on inferences that are made on the conductivity field.
For fixed $\tau$, a Gaussian prior was posited for $x$, with covariance $	k_x(t, t') :=  \exp\left( - \frac{\norm{t - t'}^2}{2 \ell_x^2} \right)$ where $\ell_x$ was fixed to $\ell_x = 0.3$.
The associated Bayesian PNM, a probabilistic meshless method (PMM), was described in Example \ref{ex:PMM}.

The statistical inference procedure is formulated in a pipeline of computations in Figure~\ref{fig:hydrocyclone_pipeline}. It is assumed that the desired outcome is to monitor the contents of the tank while the current contents are being mixed. 
This suggests a particle filter approach where a PMM $M_\tau$ is employed to handle the intractable likelihood $p(y_\tau | a_\tau)$ that involves the exact solution of a PDE.
The distribution of $a_\tau$ given $y_1,\dots,y_\tau$ is denoted $\pi_\tau$ an the computation $P(M_1,\dots,M_\tau)$ is Bayesian only if the particle approximation error due to the use of a particle filter is overlooked.


\begin{figure}
\centering
\resizebox{0.6\textwidth}{!}{
\begin{tikzpicture}

\tikzstyle{square}=[regular polygon,regular polygon sides=4];
\tikzstyle{info}=[draw,square,fill = black!0,minimum width=1.2cm];
\tikzstyle{meth}=[draw,square,fill = black!100,minimum width=1.2cm,text=white];
\tikzstyle{arrow}=[very thick,->];

\node at (0,0) (D1) {\dots};
\node[info] at (2,0) (I1) {};
\node[meth] at (4,0) (M1) {$\tau$};
\node[info] at (6,0) (I2) {};
\node at (8,0) (D2) {\dots};
\node[info] at (4,2) (data) {};

\node at (2,-1) (d1) {$\frac{\wrt \pi_{\tau - 1}}{\wrt \pi_0}$};
\node at (6,-1) (d2) {$\frac{\wrt \pi_{\tau}}{\wrt \pi_0}$};
\node at (5,2) (d3) {$y_\tau$};

\path[arrow] (D1) edge (I1);
\path[arrow] (I1) edge node [above] {$1$} (M1);
\path[arrow] (M1) edge (I2);
\path[arrow] (I2) edge (D2);
\path[arrow] (data) edge node [right] {$2$} (M1);

\end{tikzpicture}}

\caption{Pipeline for hydrocyclone application: The method node (black) represents the use of PMM solvers, which are incorporated into the likelihood for evolving the particles according to a Markov transition kernel.}
\label{fig:hydrocyclone_pipeline}
\end{figure}
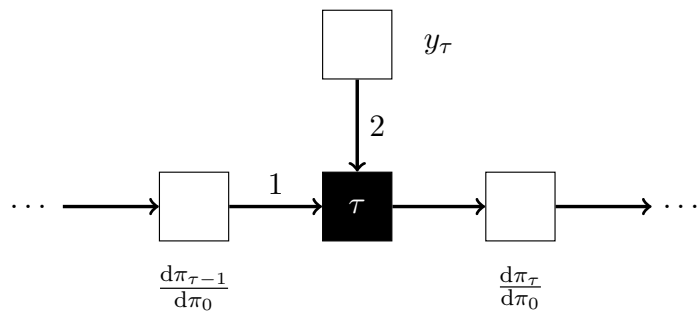

To briefly illustrate the method, Figure~\ref{fig:hydrocyclone_recovery} presents posterior means for the field $a(\tau,\cdot)$, for each post-injection time point $\tau = 1,\dots,8$. 
These are based on a particle approximation of size $P=500$, with method nodes based upon a Bayesian PNM, as in Example~\ref{ex:PMM}, with $n=119$ design points. 
The high conductivity region representing the potassium chloride solution can be seen rotating through the domain in the frames after injection, with its conductivity reducing as it mixes with the water. 
The full posterior distribution over the conductivity field is inflated as a result of explicitly modelling the discretisation error; an extensive analysis of these results will be reported in the upcoming \cite{Oates2017}.

\begin{figure}
	\includegraphics[width=\textwidth]{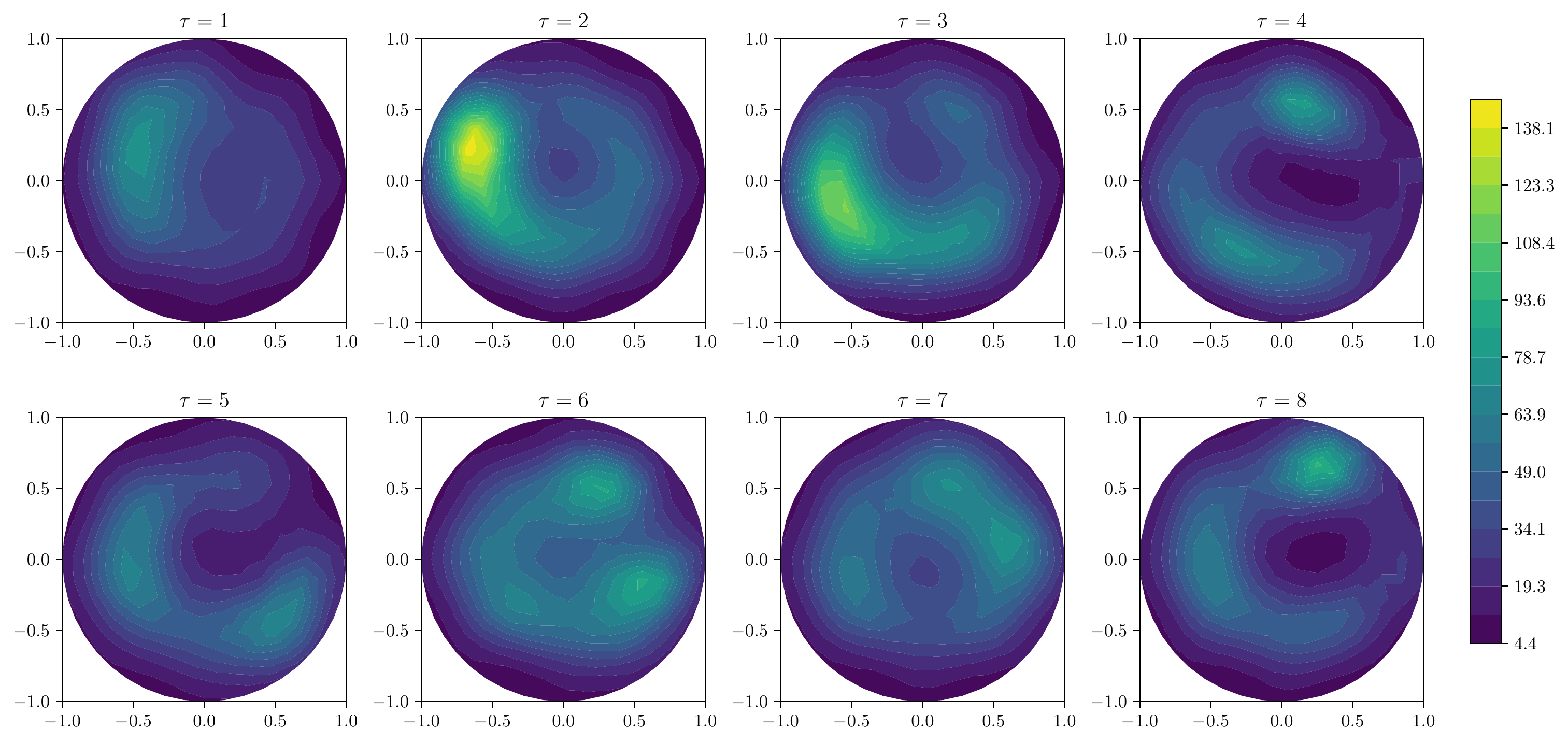}
	\caption{Mean conductivity fields recovered in the hydrocyclone experiment, for the first 8 frames post-injection.}
	\label{fig:hydrocyclone_recovery}
\end{figure}

In Figure~\ref{fig:hydrocyclone_variance}, the integrated standard-deviation $\int_D \sigma(t)\:\wrt t$ is shown for $\tau=1,\dots,8$ for both the “pipeline”, as described above, and a ``static'' approach in which no uncertainty was propagated. 
In this static approach a symmetric collocation PDE solver\footnote{Recall that the PMM has a corresponding symmetric collocation solution to the PDE as its mean function.} was used to solve the forward problem, and a separate Bayesian inversion problem was solved at each time point.
The parameters of the symmetric collocation solver were identical to those used in the PMM. 
In the left panel we observe some structural periodicity, present in both the pipeline and the static approach. 
We speculate that this may be due to the rotation of the medium causing the area of high conductivity to periodically reach an area of the domain, relative to the 8 sensors, in which it is particularly easy to recover. 
With this periodicity subtracted in the right panel, there was a clear increase in posterior uncertainty in the pipeline compared to the static approach, which is depicted. 
Temporal regularisation would usually be expected to reduce uncertainty; thus, the fact that the overall uncertainty increased with $\tau$, relative to the static formulation, demonstrates that we have quantified and propagated uncertainty due to successive discretisation of the PDE at each time point. 

\begin{figure}
	\centering
	\includegraphics[width=0.8\textwidth]{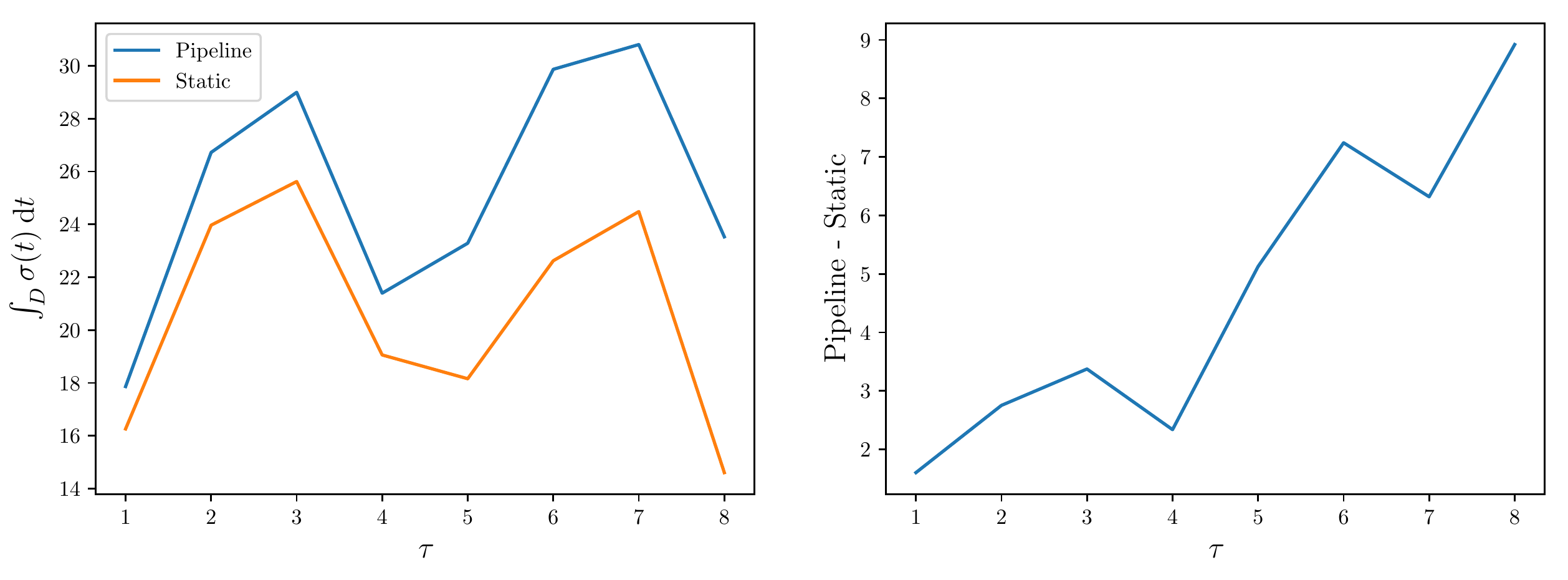}
	\caption{Left: Integrated standard-deviation over the domain, for the first 8 frames post-injection, for both the pipeline and the static approaches described in the text. Right: The difference between these two quantities.}
	\label{fig:hydrocyclone_variance}
\end{figure}

\section{Discussion}

This paper has established statistical foundations for PNMs and investigated the Bayesian case in detail.
Through connection to Bayesian inverse problems \citep{Stuart:2010gja}, we have established when Bayesian PNM can be well-defined and when the output can be considered meaningful.
The presentation touched on several important issues and a brief discussion of the most salient points is now provided.

\paragraph{Bayesian vs Non-Bayesian PNMs}

The decision to focus on Bayesian PNMs was motivated by the observation that the output of a pipeline of PNMs can only be guaranteed to admit a valid Bayesian interpretation if the constituent PNMs are each Bayesian and the prior distribution is coherent. 
Indeed, Theorem \ref{thm:markov} demonstrated that prior coherence can be established at a local level, essentially via a local Markov condition, so that Bayesian PNMs provide a extensible modelling framework as required to solve more challenging numerical tasks.
These results support a research strategy that focuses on Bayesian PNMs, so that error can be propagated in a manner that is meaningful. 

On the other hand, there are pragmatic reasons why either approximations to Bayesian PNMs, or indeed, non-Bayesian PNMs might be useful.
The predominant reason would be to circumvent the off-line computational costs that can be associated with Bayesian PNMs, such as the use of numerical disintegration developed in this research.
Recent research efforts, such as \cite{Schober:2014wt,Schober:2016uh} and \cite{Kersting2016} for the solution of ODEs, have aimed for computational costs that are competitive with classical methods, at the expense of fully Bayesian estimation for the solution of the ODE.
Such methods are of interest as non-Bayesian PNMs, but their role in pipelines of PNMs is unclear. Our contribution serves to make this explicit.

\paragraph{Computational Cost}

The present research focused on the more fundamental cost of access to the information $A(x)$, rather than the additional CPU time required to obtain the PNM output.
Indeed, numerical disintegration constituted the predominant computational cost in the applications that were reported.
However, we stress that in many challenging applications gated by discretisation error, such as occur with climate models, the fundamental cost of the information $A(x)$ will be dominant. Furthermore, the Monte Carlo methods that were employed for numerical disintegration admit substantial improvements \citep[e.g. in a similar vein to][]{Botev2012,Koskela2016}.
The objective of this paper was to establish statistical foundations that will permit the development of more sophisticated and efficient Bayesian PNMs.

\paragraph{Prior Elicitation}

Throughout this work we assumed that a belief distribution $\mu$ was provided. 
The question of \emph{whose} belief is represented in $\mu$ has been discussed by several authors and a chronology is included in the Electronic Supplement.
Of these perspectives we mention in particular \cite{Hennig:2015jf}, wherein $\mu$ is the belief of an agent that ``we get to design''.
This offers a connection to frequentist statistics, in that an agent can be designed to ensure favourable frequentist properties hold.

A robust statistics perspective is also relevant and one such approach would be to consider a generalised Bayes risk (Eq. \eqref{eq:bayes_risk}) wherein the state variable $X$ used for assessment is assumed to be drawn from a distribution $\tilde{\mu} \neq \mu$.
This offers an opportunity to derive Bayesian PNMs that are robust to certain forms of prior mis-specification.
This direction was not considered in the present paper, but has been pursued in the ACA literature for classical numerical methods \citep[see Chapter IV, Section 4 of][]{Ritter2000}.

In general, the specification of prior distributions for robust inference on an infinite-dimensional state space can be difficult.
The consistency and robustness of Bayesian inference procedures --- particularly with respect to perturbations of the prior such as those arising from numerical approximations --- in such settings is a subtle topic, with both positive \citep{CastilloNickl:2014, Doob:1949, Kleijnvanderv:2012, LeCam:1953} and negative \citep{DiaconisFreedman:1986, Freedman:1963, Owhadi2015b} results depending upon fine topological and geometric details.

In the context of computational pipelines, the challenge of eliciting a coherent prior is closely connected to the challenge of eliciting a single unified prior based on the conflicting input of multiple experts \citep{French2011,Albert2012}.

\paragraph{Consistent Estimation}

The present paper focused on foundations.
Further methodological work will be required to establish sufficient conditions for when $B(\mu,A_n(x^\dagger))$ collapses to an atom on a single element $q^\dagger = Q(x^\dagger)$ representing the data-generating QoI in the limit as the amount of information, $n$, is increased.
There are two questions here; (i) when is $q^\dagger$ identifiable from the given information, and (ii) at what rate does $B(\mu,A_n(x^\dagger))$ concentrate on $q^\dagger$.

\paragraph{Generalisation and Extensions}

Two more directions are highlighted for extension of this work.
First, note that in this paper the information operator $A : \mathcal{X} \rightarrow \mathcal{A}$ was treated as a deterministic object.
However, in some applications there is auxiliary randomness in the acquisition of information.
For our integration example, nodes $t_i$ might arise as random samples from a reference distribution on $[0,1]$.
Or, observations $x(t_i)$ themselves might occur with measurement error, for example due to finite precision arithmetic.
Then a more elaborate model $A \colon \mathcal{X} \times \Omega \rightarrow \mathcal{A}$ would be required, where $\Omega$ is a probability space that injects randomness into the information operator.
This is the setting of, for instance, randomised quasi-Monte Carlo methods.
Future work will extend the framework of PNMs to include randomised information operators of this kind.

As a second direction, recall that in an adaptive algorithm the choice of the information is made in an iterative procedure that is informed by the information observed up to that point.
For the canonical illustration in Example \ref{ex:optimalinfo} and its generalisations discussed there, it can be proven that adaptive algorithms do not out-perform non-adaptive algorithms in average case error \citep{Lee1986}.
However, outside this setting adaptation can be beneficial and should be investigated in the context of Bayesian PNM.

\paragraph{Connection with Probabilistic Programming}

The central goal of \emph{probabilistic programming} (PP) is to automate statistical computation, through symbolic representation of statistical objects and operations on those objects.
The formalism of pipelines as graphical models presented in this work can be compared to similar efforts to establish PP languages \citep{Goodman2012}. 
For instance, a method node in a pipeline can be related to a \emph{monad} aggregating several distributions into a single output distribution \citep{Scibior:2015}.
An important challenge in PP is the automation of computing conditional distributions \citep{Shan2017}.
Numerical disintegration and extensions thereof might be of independent interest to this field \citep[e.g.\ extending][]{Wood2014}.

\paragraph{Acknowledgements}

CJO was supported by the Australian Research Council (ARC) Centre of Excellence for Mathematical and Statistical Frontiers.
TJS was supported by the Excellence Initiative of the German Research Foundation (DFG) through the Free University of Berlin.
MG was supported by the Engineering and Physical Sciences (EPSRC) grants EP/J016934/1, EP/K034154/1, an EPSRC Mathematical Sciences Established Career Research Fellowship and a Lloyds Register Foundation grant for Programme on Data-Centric Engineering.
This material was based upon work partially supported by the National Science Foundation under Grant DMS-1127914 to the Statistical and Applied Mathematical Sciences Institute. Any opinions, findings, and conclusions or recommendations expressed in this material are those of the author(s) and do not necessarily reflect the views of the National Science Foundation.

The authors are grateful to Amazon for the provision of AWS credits and to the authors of the \verb+Eigen+ and \verb+Eigency+ libraries in Python.

\begin{appendices}

\section*{Appendices}
\section{Proofs} \label{appendix:proofs}

\begin{proof}[Proof of Theorem \ref{thm:optimal_information}]
The following observation will be required; the joint density of $X$ and $A = A(X)$ can be expressed in two ways:
\begin{equation}
	\label{eq:prod_eqn}
	\delta(A(x))(\rd a) \mu(\rd x) = \mu^a(\rd x) A_{\#} \mu(\rd a)
\end{equation}
which holds almost everywhere from the definition of a disintegration $\{\mu^a\}_{a \in \mathcal{A}}$.
Note that our integrability assumption justifies the interchange of integrals from Fubini's theorem.

The Bayes risk for a Bayesian PNM $M_{\text{BPNM}} = (A,B_{\text{BPNM}})$, $B_{\text{BPNM}}(\mu,a) = Q_{\#} \mu^a$, can be expressed as:
\begin{align*}
	R(\mu,M_{\text{BPNM}})
	& = \int r(x,B(\mu,A(x))) \mu(\rd x) \\
	& = \iint L( Q(x) , q ) Q_{\#} \mu^{A(x)}(\rd q) \mu(\rd x) \quad \text{(since $M$ Bayesian)} \\
	& = \iiint L( Q(x) , q ) Q_{\#} \mu^a(\rd q) \delta(A(x))(\rd a) \mu(\rd x) \\
	& = \iiint L( Q(x) , Q(x') ) \mu^a (\rd x') \mu^a (\rd x) A_{\#} \mu(\rd a) \quad \text{(from Eq.~\eqref{eq:prod_eqn})}
\end{align*}
On the other hand, let 
$$
b(a) \in \argmin_{q \in \mathcal{Q}} \int L(Q(x) , q) \mu^a(\wrt x)
$$
be a Bayes act.
Then the Bayes risk associated with such a method $M_{\text{BR}} = (A,B_{\text{BR}})$, $B_{\text{BR}}(\mu,a) = \delta(b(a))$, can be expressed as:
\begin{align*}
	R(\mu,M_{\text{BR}})
	& = \int L( Q(x) , b(A(x)) ) \mu(\rd x) \\
	& = \iint L( Q(x) , b(a) ) \delta(A(x))(\rd a) \mu(\rd x) \\
	& = \iint L( Q(x) , b(a) ) \mu^a(\rd x) A_{\#} \mu(\rd a)  \quad \text{(from Eq.~\eqref{eq:prod_eqn})}  
\end{align*}
Next we use the inner product structure on $\mathcal{Q}$ and the form of the loss function as $L(q,q') = \|q - q'\|_{\mathcal{Q}}^2$ to argue that $R(\mu,M_{\text{BPNM}}) = 2R(\mu,M_{\text{BR}})$, which in turn implies that the optimal information $A_\mu$ for Bayesian PNM and $A_\mu^*$ for ACA are identical.

For this final step, fix $a \in \mathcal{A}$ and denote the random variables $Q^a(X) = Q(X) - b(a)$ that are induced according to $X \sim \mu^a$.
Denote by $\tilde{Q}^a$ an independent copy of $Q^a$ generated from $\tilde{X} \sim \mu^a$.
The notation $\mathbb{E}$ will be used to refer to the expectation taken over $X,\tilde{X}$.
Then we have 
\begin{align*}
	Q(X) - Q(\tilde{X})
	& = (Q(X) - b(a)) - (Q(\tilde{X}) - b(a)) \\
	& = Q^a(X) - \tilde{Q}^a(\tilde{X})
\end{align*}
and moreover, from Theorem \ref{thm:bayes_mean} the posterior mean of $Q(X)$ is $b(a)$ and thus $\mathbb{E}[Q^a] = \mathbb{E}[\tilde{Q}^a] = 0$.
Then
\begin{align*}
R(\mu,M_{\text{BPNM}}) & = \int \mathbb{E} [\|Q^a - \tilde{Q}^a\|_{\mathcal{Q}}^2] A_{\#} \mu (\rd a) \\
	& = \int \mathbb{E} [\|Q^a\|_{\mathcal{Q}}^2 - 2 \langle Q^a , \tilde{Q}^a \rangle_{\mathcal{Q}} + \|\tilde{Q}_A^a\|_{\mathcal{Q}}^2] A_{\#} \mu (\rd a) \\
	& = 2 \int \mathbb{E} [\|Q^a\|_{\mathcal{Q}}^2] A_{\#} \mu (\rd a)  \quad \text{(since $\mathbb{E}[Q^a] = 0$ and $Q^a \ci \tilde{Q}^a$)} \\
	& = 2 R(\mu,M_{\text{BR}})
\end{align*}
as required.
\end{proof}

\begin{proof}[Proof of Theorem \ref{thm:rcp_contraction}]
Fix $f \in \mathcal{F}$ and $a \in \mathcal{A}$.
Then:
\begin{align*}
	\mu_\delta^a(f)
	& = \frac{1}{Z_\delta^a} \int f(x) \phi\left( \frac{\|A(x) - a\|_{\mathcal{A}}}{\delta} \right) \mu(\mathrm{d}x) \\
	& = \frac{1}{Z_\delta^a} \iint f(x) \phi\left( \frac{\|\tilde{a} - a\|_{\mathcal{A}}}{\delta} \right) \mu^{\tilde{a}}(\mathrm{d}x) A_\# \mu(\mathrm{d}\tilde{a}) \quad \text{(from Eq.~\eqref{eq:prod_eqn})} \\
	& = \frac{1}{Z_\delta^a} \int \phi\left( \frac{\|\tilde{a} - a\|_{\mathcal{A}}}{\delta} \right) \mu^{\tilde{a}}(f) A_\# \mu(\mathrm{d}\tilde{a}) \\
	& = \int \mu^{\tilde{a}}(f) A_\# \mu_\delta^a(\mathrm{d}\tilde{a}).
\end{align*}
Thus
\begin{align}
	|\mu_\delta^a(f) - \mu^a(f)|
	& = \left| \int [\mu^{\tilde{a}}(f) - \mu^a(f)] A_\# \mu_\delta^a(\mathrm{d}\tilde{a}) \right| \notag \\
	& \leq C_\mu^\alpha \|f\|_{\mathcal{F}} \int \|\tilde{a} - a\|_{\mathcal{A}}^\alpha A_\# \mu_\delta^a(\mathrm{d}\tilde{a}) \quad \text{(Assumption \ref{assumption:lipschitz_rcp})} . \label{eqn: intermediate step}
\end{align}
Now consider the random variable 
\begin{align}
R & \defeq \frac{\|A(X) - a\|_{\mathcal{A}}}{\delta} \label{eq: R def}
\end{align} 
induced from $X \sim \mu$.
The existence of a continuous and positive density $p_A$ implies that $R$ also admits a density on $[0,\infty)$, denoted $p_{R,\delta}$.
The fact that $p_A$ is uniform on an infinitesimal neighbourhood of $a$ implies that $p_{R,\delta}(r)$ is proportional to the surface area of a hypersphere of radius $\delta r$ centred on $a \in \mathcal{A}$:
\begin{equation}
	p_{R.\delta}(r) = \frac{2 \pi^{n/2}}{\Gamma(\frac{n}{2})} (\delta r)^{n-1} (p_A(a) + o(1)) \label{pR eqn}
\end{equation} 
This is valid since $\mathcal{A}$ is open and the hypersphere will be contained in $\mathcal{A}$ for $r$ sufficiently small.
Eq.~\eqref{eqn: intermediate step} can then be evaluated:
\begin{align}
	\int \|\tilde{a} - a\|_{\mathcal{A}}^\alpha A_\# \mu_\delta^a(\mathrm{d}\tilde{a})
	& = \frac{\int \|\tilde{a} - a\|_{\mathcal{A}}^\alpha \phi\left( \frac{\|\tilde{a} - a\|_{\mathcal{A}}}{\delta} \right) A_\# \mu(\mathrm{d}\tilde{a})}{\int \phi\left( \frac{\|\tilde{a} - a\|_{\mathcal{A}}}{\delta} \right) A_\# \mu(\mathrm{d}\tilde{a})} \notag \\
	& = \delta^\alpha \frac{\int r^\alpha \phi(r) p_{R,\delta}(r) \mathrm{d}r}{\int \phi(r) p_{R,\delta}(r) \mathrm{d}r} \quad \text{(change of variables; Eq. \ref{eq: R def}).} \label{eqn: ratio of ints} \\
	& \xrightarrow[\delta \downarrow 0]{} \frac{\int r^{\alpha + n - 1} \phi(r) \mathrm{d}r}{\int r^{n - 1} \phi(r) \mathrm{d}r} \quad \text{(from Eq. \ref{pR eqn})} \nonumber \\
	& = \frac{C_\phi^\alpha}{C_\phi^0} \quad \text{($< \infty$ from Assumption \ref{varphi_assumption}).} \nonumber
\end{align}
Thus, for $\delta$ sufficiently small, Eq.~\eqref{eqn: ratio of ints} can be bounded above by $\delta^\alpha (1 + \bar{C}_\phi^\alpha)$ where $\bar{C}_\phi^\alpha \defeq C_\phi^\alpha / C_\phi^0$ and ``$1$'' is in this case an arbitrary positive constant.
This establishes the upper bound
\begin{align*}
	|\mu_\delta^a(f) - \mu^a(f)| & \leq C_\mu^\alpha (1 + \bar{C}_\phi^\alpha) \|f\|_{\mathcal{F}} \; \delta^\alpha 
\end{align*}
for $\delta$ sufficiently small and completes the proof.
\end{proof}

\begin{proof}[Proof of Theorem \ref{thm:markov}]
To reduce the notation, suppose that the random variables $Y_1,\dots,Y_J$ admit a joint density $p(y_1,\dots,y_J)$,
However, we emphasise that existence of a density is not required for the proof to hold.
To further reduce notation, denote $y_{a:b} = (y_a,\dots,y_b)$.

The output of the computation $P(M_1,\dots,M_n)$ was defined algorithmically in Definition \ref{def:bayes_computation} and illustrated in Example \ref{ex:propaga}.
Our aim is to show that this algorithmic output coincides with the distribution $(Q_n)_\# \mu^a$ on $\mathcal{Q}_n$, which is identified in the present notation with $p(y_J | y_{1:I})$.

For $j \in \{I+1,\dots,J\}$, the coherence condition on $Y_1,\dots,Y_J$ translates into the present notation as $p(y_j | y_{1:j-1}) = p(y_j | y_{\pi(j)})$.
This allows us to deduce that:
\begin{align*}
p(y_J | y_{1:I})
& = \int \dots \int p(y_{I+1:J} | y_{1:I}) \rd y_{I+1:J-1} \\
& = \int \dots \int \prod_{j=I+1}^J p(y_j | y_{1:j-1}) \rd y_{I+1:J-1} \\
& = \int \dots \int \prod_{j=I+1}^J p(y_j | y_{\pi(j)}) \rd y_{I+1:J-1} .
\end{align*}
The right hand side is recognised as the output of the computation $P(M_1,\dots,M_n)$, as defined in Definition \ref{def:bayes_computation}.
This completes the proof.
\end{proof}

\end{appendices}

\bibliographystyle{abbrvnat}
\bibliography{bibliography}

\setcounter{section}{0}
\renewcommand{\thesection}{S\arabic{section}}
\renewcommand*{\theHsection}{S\thesection}


\title{Electronic Supplement}
\author{to the paper \emph{Bayesian Probabilistic Numerical Methods}}
\date{}
\maketitle


\section{Philosophical Status of the Belief Distribution}

The aim of this section is to discuss in detail the semantic status of the belief distribution $\mu$ in a probabilistic numerical method (PNM).
In Section \ref{sec:history} we survey historical work on this topic, while in Section \ref{sec:contemporary} more recent literature is covered.
Then in Section \ref{sec:philosophy} we highlight some philosophical objections and their counter-arguments.

\subsection{Historical Precedent} \label{sec:history}

The use of probabilistic and statistical methods to model a deterministic mathematical object can be traced back to \cite{Poincare1912}, who used a stochastic model to construct interpolation formulae. 
In brief, Poincar\'{e} formulated a polynomial 
\begin{equation*}
	f(x) = a_0 + a_1 x + \dots + a_m x^m
\end{equation*}
whose coefficients $a_i$ were modelled as independent Gaussian random variables.
Thus Poincar\'{e} in effect constructed a Gaussian measure over the Hilbert space with basis $\{1,x,\dots,x^m\}$.
This pre-empted \cite{Kimeldorf1970,Kimeldorf1970a} and others, which associated spline interpolation formulae to the means of Gaussian measures over Hilbert spaces.

The first explicit statistical model for numerical error (of which we are aware) was in the literature on rounding error in the numerical solution of ordinary differential equations (ODE), as summarised in \cite{Hull1966}.
Therein it was supposed that rounding, by which we mean representation of a real number
\begin{equation*}
	x = 0.a_1 a_2 a_3 a_4 \dots \quad \in [0,1]
\end{equation*}
in a truncated form
\begin{equation*}
	\hat{x} = 0.a_1 a_2 a_3 a_4 \dots a_n ,
\end{equation*}
is such that the error $e = x - \hat{x}$ can be reasonably modelled by a uniform random variable on $[-5 \times 10^{-(n+1)} , 5 \times 10^{-(n+1)}]$.
This implies a distribution $\mu$ over the unknown value of $x$ given $\hat{x}$.
The contribution of \cite{Hull1966} and others was to replace the last digit $a_n$, in each stored number that arises in the numerical solution of an ODE, with a uniformly chosen element of $\{0,\dots,9\}$.
This performs approximate propagation of the \emph{numerical uncertainty} due to rounding error through further computation and, in their case, induces a distribution over the solution space of the ODE.
Note that this work focused on rounding error, rather than the (time) discretisation error that is intrinsic to numerical ODE solvers; this could reflect the limited precision arithmetic that was available from the computer hardware of the period.

\cite{Larkin1972} was an important historical paper for PNMs, being the first to set out the modern statistical agenda for PNMs:
\begin{quote}
In any particular problem situation we are given certain specific properties of the solution, e.g.\ a finite number of ordinate or derivative values at fixed abscissae.
If we can assume no more than this basic information we can conclude only that our required solution is a member of that class of functions which possesses the given properties - a tautology which is unlikely to appeal to an experimental scientist!
Clearly, we need to be given, or to assume, extra information in order to make more definite statements about the required function.

Typically, we shall assume general properties, such as continuity or non-negativity of the solution and/or its derivatives, and use the given specific properties in order to assist in making a selection from the class $K$ of all functions possessing the assumed general properties.
We shall choose $K$ either to be a Hilbert space or to be simply related to one. 
\end{quote}
This description defines a set $K$ of permissible functions, rather than an explicit distribution over $K$, but it is clear that Larkin envisaged numerical analysis as an instance of statistical estimation:
\begin{quote}
In the present approach, an \textit{a priori} localisation is achieved effectively by making an assumption about the relative likelihoods of elements of the Hilbert space of possible candidates for the solution to the original problem.
Among other things, this permits, at least in principle, the derivation of joint probability density functions for functionals on the space and also allows us to evaluate confidence limits on the estimate of a required functional (in terms of given values of other functionals) without any extra information about the norm of the function in question. 
\end{quote}

\noindent Later, \cite{Diaconis1988} re-iterated this argument for the construction of $K$ more explicitly, considering numerical integration of the function
\begin{equation*}
	f(x) = \exp\left\{ \cosh\left( \frac{x + x^2 + \cos(x)}{ 3 + \sin(x^3)} \right) \right\} \;.
\end{equation*}
over the unit interval.
In particular, Diaconis asked:
\begin{quote}
	``What does it mean to `know' a function?''
	The formula says some things (e.g.\ $f$ is smooth, positive and bounded by $20$ on $[0,1]$) but there are many other facts about $f$ that we don't know (e.g.\ is $f$ monotone, unimodal or convex?)
\end{quote}
This argument was provided as justification for belief distributions that encode certain basic features, such as the smoothness of the integrand.
The belief distributions that were then considered in Diaconis' paper were Gaussian distributions on $K$.
Diaconis, as well as \cite{Larkin1972,Kadane:1983ww}, observed that some classical numerical methods are Bayes rules in this context.

The arguments of these papers are intrinsic to modern PNMs.
However, the associated theoretical analysis of computation under finite information has proceeded outside of statistics, in the applied mathematical literature, where it is usually presented without a statistical context.
That research is reviewed next.

\subsection{Contemporary Outlook} \label{sec:contemporary}

The mathematical foundations of computation based on finite information are established in the field of \emph{information-based complexity} (IBC).
The monograph of \cite{Traub1988} presents the foundations of IBC.
In brief, the starting point for IBC is the mantra that 
\begin{quote}
To compute fast you need to compute with partial information ($\sim$ Houman Owhadi, SIAM UQ 2016)
\end{quote}
This motivates the search for optimal approximations based on finite information, in either the worst-case or average-case sense of optimal.
The particular development of PNMs that we presented in the main text is somewhat aligned to \emph{average-case analysis} (ACA) and we focus on that literature in what follows.

Among the earliest work on ACA, \cite{Suldin1959,Suldin1960} studied numerical integration and $L_2$ function approximation in the setting where $\mu$ was induced from the Weiner process, with a focus on optimal linear methods.
Later, \cite{Sacks1970} moved from analysis with fixed $\mu$ to analysis over a class of $\mu$ defined by the smoothness properties of their covariance kernels.
At the same time \cite{Kimeldorf1970,Kimeldorf1970a} established optimality properties of splines in reproducing kernel Hilbert spaces in the ACA context.
\cite{Kadane1985,Diaconis1988} discussed the connection between ACA and Bayesian statistics.
A general framework for ACA was formalised in the IBC monograph of \cite{Traub1988}, while \cite{Ritter2000} provides a more recent account.

Game theoretic arguments have recently been explored in \cite{Owhadi2015}, who argued that the optimal prior for probabilistic meshless methods \citep{Cockayne:2016ts} is a particular Gaussian measure under a game theoretic framework where the energy norm is the loss function. 
This provides one route to the specification of default or \emph{objective} priors for PNMs which deserves further exploration in general.

The question of  ``whose'' belief is captured in $\mu$ was addressed in \cite{Hennig:2015jf}, where it was argued that the prior information in $\mu$ represents that of a hypothetical agent (numerical analyst) which 
\begin{quote}
[$\dots$] we are allowed to design ($\sim$ Michael Osborne, personal correspondence, 2016).
\end{quote}
This represents a more pragmatic approach to the design of PNM.

\subsection{Paradise Lost?} \label{sec:philosophy}

Typical numerical algorithms contain several different sources of discretisation error.
Consider the solution of the wave equation:
A standard finite element method involves both spatial and temporal discretisations, a series of numerical quadrature problems, as well as the use of finite precision arithmetic for all numerical calculations.
Yet, decades of numerical analysis have led to highly optimised computer codes such that these methods can be routinely used.
To develop PNM for solution of the wave equation, which accounts for each separate source of discretisation error, is it required to unpick and reconstruct such established numerical algorithms?
This would be an unattractive prospect that would detract from further research into PNMs.

Our view is that there is a choice for which discretisation errors to model.
In practice the PNMs implemented in this work were run on floating point precision machines, yet we did not model rounding error in their output.
This was because, in our examples, floating point error is insignificant compared to discretisation error and so we chose not to model it.
This is in line with the view that a model is a useful simplification of the real world.

\section{Existence of Non-Randomised Bayes Rule}

In this section we recall an argument for the general existence of non-randomised Bayes rules, that was stated without proof in the main text.
Sufficient conditions for Fubini's theorem to hold are assumed.

\begin{proposition}
	Let $\mathfrak{B}(A)$ be non-empty. 
	Then $\mathfrak{B}(A)$ contains a classical numerical method of the form $B(\mu,a) = \delta \circ b(a)$ where $b(a)$ is a Bayes act for each $a \in \mathcal{A}$.
\end{proposition}
\begin{proof}
	Let $\mathfrak{C}$ be the set of belief update operators of the classical form $B(\mu,a) = \delta \circ b(a)$.
	Suppose there exists a belief update operator $B^* \in \mathfrak{B}(A) \setminus \mathfrak{C}$.
	Then $B^*$ can be characterised as a non-atomic distribution $\pi$ over the elements of $\mathfrak{C}$.
	Its risk can be computed as:
	\begin{align*}
		R(\mu,(A,B^*))
		& = \int r(Q(x) , B^*(\mu,A(x))) \mu(\mathrm{d}x) \\
		& = \iint L(Q(x) , b(A(x))) \pi(\mathrm{d}b) \mu(\mathrm{d}x) \\
		& = \int R(\mu,(A,\delta \circ b)) \pi(\mathrm{d}b) .
	\end{align*}
	If we had $R(\mu,(A,B^*)) < R(\mu,(A,\delta \circ b))$ for all $\delta \circ b \in \mathfrak{C}$ we would have a contradiction, so it follows that $\mathfrak{B}(A) \cap \mathfrak{C}$ is non-empty.
	This completes the proof.
\end{proof}

\section{Optimal Information: A Counterexample} \label{sec:optimal_counterexample}

In this section we demonstrate that the optimal information $A_\mu$ for Bayesian PNM and the optimal information $A_\mu^*$ from average case analysis are different in general.

Let $\mathcal{X} = \{\spadesuit,\diamondsuit,\heartsuit,\clubsuit\}$ be a discrete set, with quantity of interest $Q(x) = 1[x = \spadesuit]$ and information operator $A(x) = 1[x \in S]$ so that $\mathcal{Q} = \mathcal{A} = \set{0, 1}$. 
In particular, $\mathcal{Q}$ is not a vector space and hence not an inner product space as specified in Theorem~\ref{thm:optimal_information}.

Consider two possible choices, $S = \{\spadesuit,\diamondsuit\}$ and $S = \{\spadesuit,\diamondsuit,\heartsuit\}$. 
Assume a uniform prior over $\mathcal{X}$. 
Consider the 0-1 loss function $L(q,q') = 1[q \neq q']$.
It will be shown that ACA optimal information for this example can be based on either $S = \{\spadesuit,\diamondsuit\}$ or $S = \{\spadesuit,\diamondsuit,\heartsuit\}$ whereas PNM optimal information must be based on $S = \{\spadesuit,\diamondsuit,\heartsuit\}$.
Thus Bayesian PNM optimal information $A_\mu$ and ACA optimal information $A_\mu^*$ need not coincide in general.

The classical case considers a method of the form $M_{\text{BR}} = (A,B_{\text{BR}})$, $B_{\text{BR}} = \delta \circ b$, where
\begin{equation*}
b(a) = 1[a = 0] c_0 + 1[a = 1] c_1
\end{equation*}
for some $c_0,c_1 \in \{0,1\}$.
The Bayes risk is
\begin{align*}
R(\mu,M_{\text{BR}}) & = \frac{1}{4} \sum_{x \in \{\spadesuit,\diamondsuit,\heartsuit,\clubsuit\}} 1[x \notin S] \: L(c_0,1[x=\spadesuit]) + 1[x \in S] \: L(c_1,1[x=\spadesuit]) .
\end{align*}

\paragraph{Case of $S = \{\spadesuit,\diamondsuit\}$:} 
We have
\begin{align*}
	4 \: R(\mu,M_{\text{BR}}) & = L(c_1,1) + L(c_1,0) + L(c_0,0) + L(c_0,0) \\
	& = 1[c_1 = 0] + 1[c_1 = 1] + 2 \times 1[c_0 = 1]
\end{align*}
which is minimised by $c_1 \in \{0, 1\}$ and $c_0 = 0$ to obtain a minimum Bayes risk of $\frac{1}{4}$.

\paragraph{Case of $S = \{\spadesuit,\diamondsuit,\heartsuit\}$:} 
We have
\begin{align*}
4 \: R(\mu,M_{\text{BR}}) & = L(c_1,1) + L(c_1,0) + L(c_1,0) + L(c_0,0) \\
& = 1[c_1 = 0] + 2 \times 1[c_1 = 1] + 1[c_0 = 1]
\end{align*}
which is minimised by $c_0 = 0$ and $c_1 = 0$ to again obtain a minimum Bayes risk of $\frac{1}{4}$.
Thus the ACA optimal information can be based on either $S = \{\spadesuit,\diamondsuit\}$ or $S = \{\spadesuit,\diamondsuit,\heartsuit\}$.

On the other hand, for the Bayesian PNM we have that $M_{\text{BPNM}} = (A,B_{\text{BPNM}})$, $B_{\text{BPNM}} = Q_{\#} \mu^A$ and
\begin{align*}
R(\mu, M_{\text{BPNM}}) & = \frac{1}{4} \sum_{x \in \{\spadesuit,\diamondsuit,\heartsuit,\clubsuit\}} 1[x \notin S] L(0,0) \\
& \hspace{50pt} + 1[x \in S] \left\{ (1 - \frac{1}{|S|}) L(0,1[x=\spadesuit]) + \frac{1}{|S|} L(1,1[x=\spadesuit]) \right\} .
\end{align*}

\paragraph{Case of $S = \{\spadesuit,\diamondsuit\}$:} 
We have
\begin{equation*}
4 \: R(\mu, M_{\text{BPNM}}) = \frac{1}{2} + \frac{1}{2} + 0 + 0 \quad = \quad 1 .
\end{equation*}

\paragraph{Case of $S = \{\spadesuit,\diamondsuit,\heartsuit\}$:} 
We have
\begin{equation*}
4 \:R(\mu, M_{\text{BPNM}}) = \frac{2}{3} + \frac{1}{3} + \frac{1}{3} + 0 \quad = \quad \frac{4}{3} .
\end{equation*}

\noindent Thus the PNM optimal information is $S = \{\spadesuit,\diamondsuit\}$ and \textit{not} $S = \{\spadesuit,\diamondsuit,\heartsuit\}$. Hence, PNM and ACA optimal information differ in general.

\section{Monte Carlo Methods for Numerical Disintegration} 

In this section, Monte Carlo methods for sampling from the distribution $\mu_\delta^a$ (or $\mu_{\delta,N}^a$; the $N$ subscript will be suppressed to reduce notation in the sequel) are considered.
The Monte Carlo approximation of $\mu_\delta^a$ is, in effect, a problem in rare event simulation as most of the mass of $\mu_\delta^a$ will be confined to a set $S$ such that $\mu(S)$ is small. 
Rare events pose some difficulties for classical Monte Carlo, as an enormous number of draws can be required to study the rare event of interest.

In the literature there are two major solutions proposed.
\emph{Importance sampling} \citep{Robert2013} samples from a modified process, under which the event of interest is more likely, then re-weights these samples to compensate for the adjustment.
Conversely, in \emph{splitting} \citep{Botev2012} trajectories of the process are constructed in a genetic fashion, by retaining and duplicating those which approach the events of interest and discarding others. Splitting is closely related to SMC \citep{Cerou:2012ha} and Feynman--Kac models \citep{DelMoral:2004fe}.

The splitting approach is described in the following section, while in Section~\ref{sec:pt} a parallel tempering (PT) algorithm is described. In spirit these approaches are similar in that they employ a tempering approach to ease sampling the relaxed posterior distribution for a small value of $\delta$. The SMC method employs a particle approximation to accomplish this, while the PT algorithm uses coupled Markov chains.

\subsection{Sequential Monte Carlo Algorithms  for Numerical Disintegration} \label{sec:rcp_smc}

Let $\set{\delta_i}_{i=0}^m$ be such that $\delta_0 = \infty$, $\delta_m = \delta$ and $\delta_i > \delta_{i+1} > 0$ for all $i < m-1$.
Furthermore let $\set{K_i}_{i=1}^m$ be some set of Markov transition kernels that leave $\mu^a_{\delta_{i}}$ invariant, for which $K_i(\cdot, S)$ is measurable for all $S \in \Sigma_\mathcal{X}$ and $K_i(x, \cdot)$ is an element of $\mathcal{P}_{\mathcal{X}}$ for all $x \in \mathcal{X}$.
Then our SMC for numerical disintegration (SMC-ND) algorithm, based on $P$ particles, is given in Algorithm~\ref{alg:rcp_smc}. 
Here we have used $\text{Discrete}(\{x_j\}_{j=1}^P ; \{w_j\}_{j=1}^P)$ to denote the discrete distribution which puts mass proportional to $w_j$ on the state $x_j \in \mathcal{X}$.

\begin{algorithm}
	Sample $x_j^0 \sim \mu$ for $j=1,\dots,P$ [Initialise] \\
	\For{$i=1,\dots,m$} {
		Sample $x_j^{i-1} \sim K_i(x_j^{i-1}, \cdot)$ for $j = 1,\dots,P$ [Move] \\
		Set $w_j^i \gets \frac{ \phi\left( \delta_i^{-1} \| A(x_j^{i-1}) - a \|_{\mathcal{A}} \right) }{ \phi\left( \delta_{i-1}^{-1} \| A(x_j^{i-1}) - a \|_{\mathcal{A}} \right) }$ for $j = 1,\dots,P$ [Re-weight] \\
		Sample $x_j^i \sim \text{Discrete}(\{x_j^{i-1}\}_{j=1}^P ; \{w_j^i\}_{j=1}^P)$  for $j = 1,\dots,P$ [Re-sample] \\
	}
	\caption{Sequential Monte Carlo for Numerical Disintegration (SMC-ND).} \label{alg:rcp_smc}
\end{algorithm}

The output of the SMC-ND algorithm is an empirical approximation\footnote{The bandwidth parameter $\delta$ and the use of $\delta$ to denote an atomic distribution should not be confused.}
\begin{equation*}
\mu^a_{\delta_m, P} = \frac{1}{P} \sum_{j=1}^P \delta(x_j^m)
\end{equation*} 
to $\mu^a_{\delta_m}$ based on a population of $P$ particles $\{x_j^m\}_{j=1}^P$.
There is substantial room to extend and improve the SMC-ND algorithm based on the wide body of literature available on this subject \citep[e.g.][]{Doucet2001a,DelMoral2006,Beskos2016,Ellam2016}, but we defer all such improvements for future work.
Our aim in the remainder is to establish the approximation properties of the SMC-ND output.
This will be based on theoretical results in \cite{DelMoral2006}.

\begin{assumption} \label{assumption:positive}
$\phi > 0$ on $\mathbb{R}_+$.
\end{assumption}

\begin{assumption} \label{assumption:composition}
	For all $i = 0,\dots,m-1$ and all $x, y \in \mathcal{X}$, it holds that $K_{i+1}(x, \cdot) \ll K_{i+1}(y, \cdot)$. 
Furthermore there exist constants $\epsilon_i > 0$ such that the \radonnikodym derivative
	\begin{equation*}
		\frac{\wrt K_{i+1}(x, \cdot)}{\wrt K_{i+1}(y, \cdot)} \geq \epsilon_i .
	\end{equation*}
\end{assumption}

\noindent
Assumption \ref{assumption:positive} ensures that Algorithm \ref{alg:rcp_smc} is well-defined, else it can happen that all particles are assigned zero weight and re-sampling will fail.
However, the result that we obtain in Theorem \ref{thm:smc_error} below can also be established in the special case of an indicator function $\phi(r) = 1[r < 1]$.
The details for this variation of the results are also included in the sequel.

The interpretation of Assumption \ref{assumption:composition} is that, for fixed $i$, transition kernels do not allocate arbitrarily large or small amounts of mass to different areas of the state space, as a function of their first argument.
This poses a constraint on the choice of Markov kernels for the SMC-ND algorithm.

\begin{theorem} \label{thm:smc_error}
	For all $\delta \in \{\delta_i\}_{i=0}^m$ and fixed $p \geq 1$ it holds that
	\begin{equation*}
		\expected \left( 
			\left[\mu^a_{\delta, P}(f) - \mu^a_{\delta}(f) \right]^p
		\right)^{\frac{1}{p}}
		\leq \frac{C_p \norm{f}_{\mathcal{F}}}{\sqrt{P}}
	\end{equation*}
	for some constant $C_p$ independent of $P$ but dependent on $\{\delta_i\}_{i=0}^m$, $p$ and $\{\epsilon_i\}_{i=0}^{m-1}$.
\end{theorem}

\noindent
The proof of Theorem \ref{thm:smc_error} is presented next. 
Note that the established bound is independent of $\delta \in \{\delta_i\}_{i=0}^m$; this is therefore a uniform convergence result.
The assumptions and the conclusion of Theorem \ref{thm:smc_error} can be weakened in several directions, as discussed in detail in \citep{DelMoral2006}.
Development of SMC methods in the context of high-dimensional and infinite-dimensional state spaces has also been considered in \cite{Beskos2014,Beskos2015}.

\subsection{Proof of Theorem \ref{thm:smc_error}}

In this section we establish the uniform convergence of the SMC-ND algorithm as claimed in Theorem \ref{thm:smc_error}.
This relies on a powerful technical result from \cite{DelMoral:2004fe}, whose context is now established.

\subsubsection{Feynman--Kac Models} \label{feynman-kac setup}

Let $(E_i,\mathcal{E}_i)$ for $i=0,\dots,m$ be a collection of measurable spaces.
Let $\eta_0$ be a measure on $E_0$ and let $\Gamma_i$ index a collection of Markov transition kernels from $E_{i-1}$ to $E_i$.
Let $G_i \colon E_i \rightarrow (0,1]$ be a collection of functions, which are referred to as \emph{potentials}.
The triplets $(\eta_0,G_i,\Gamma_i)$ are associated with \emph{Feynman--Kac} measures $\eta_i$ on $E_i$ defined as, for bounded and measurable functions $f_i$ on $E_i$;
\begin{align*}
	\eta_i(f_i) & = \frac{\gamma_i(f_i)}{\gamma_i(1)} \\
	\gamma_i(f_i) & = \mathbb{E}_{\eta_0} \left[ f_i(X^i) \prod_{j=0}^{i-1} G_j(X^j) \right]
\end{align*}
where the expectation is taken with respect to the Markov process $X^i$ defined by $X^0 \sim \eta_0$ and $X^i | X^{i-1} \sim \Gamma_i(X^{i-1} , \cdot)$.

The Feynman--Kac measures can be associated with a (non-unique) \emph{McKean interpretation} of the form $\eta_{i+1} = \eta_i \Lambda_{i+1,\eta_i}$ where the $\Lambda_{i+1,\eta}$ are a collection of Markov transitions for which the following compatibility condition holds:
\begin{equation*}
	\eta \Lambda_{i+1,\eta} = \frac{G_i}{\eta(G_i)} \eta \Gamma_{i+1}
\end{equation*}
Then the $\eta_i$ can be interpreted as the $i$th step marginal distribution of the non-homogeneous Markov chain defined by $X^0 \sim \eta_0$ and $X^{i+1} | X^i \sim \Lambda_{i+1,\eta_i}(X^i , \cdot)$.
The corresponding $P$-particle model is defined on $E_i^P = E_i \times \dots \times E_i$ and has
\begin{align*}
	\mathbf{X}^0 & \sim \eta_0^P \\
	\mathbb{P}(\mathbf{X}^i \in \mathrm{d}\mathbf{x}^i | \mathbf{X}^i) & = \prod_{j=1}^P \Lambda_{i,\eta_{i-1}^P} (X_j^{i-1} , \mathrm{d}x_j^i)
\end{align*}
where $\eta_i^P = \frac{1}{P} \sum_{j=1}^P \delta(X_j^i)$ is an empirical (random) measure on $E_i$.
The SMC-ND algorithm can be cast as an instance of such a $P$-particle model, as is made clear later.

The result that we require from \cite{DelMoral:2004fe} is given next.
Denote by $\text{Osc}_1(E_i)$ the set of measurable functions $f_i$ on $E_i$ for which $\sup\{|f_i(x^i) - f_i(y^i)| \; : \; x^i,y^i \in E_i\} \leq 1$.

\begin{theorem*}[Theorem 7.4.4 in \cite{DelMoral:2004fe}]
	Suppose that:
	\begin{itemize}
		\item[$(G)$] There exist $\epsilon_i^G \in (0,1]$ such that $G_i(x^i) \geq \epsilon_i^G G_i(y^i) > 0$ for all $x^i,y^i \in E_i$.
		\item[$(M_1)$] There exist $\epsilon_i^\Gamma \in(0,1)$ such that $\Gamma_{i+1}(x^i,\cdot) \geq \epsilon_i^\Gamma \Gamma_{i+1}(y^i,\cdot)$ for all $x^i,y^i \in E_i$.
	\end{itemize}
	Then for $p \geq 1$ and any valid McKean interpretation $\Lambda_{i,\eta}$, the associated $P$-particle model $\eta_i^P$ satisfies the uniform (in $i$) bound
	\begin{equation*}
		\sup_{0 \leq i \leq m} \sup_{f_i \in \text{Osc}_1(E_i)} \sqrt{P} \mathbb{E}[ |\eta_i^P(f_i) - \eta_i(f_i)|^p ]^{1/p} \leq C_p
	\end{equation*}
	for some constant $C_p$ independent of $P$ but dependent on $\{\epsilon_i^G\}_{i=0}^m$ and $\{\epsilon_i^\Gamma\}_{i=0}^{m-1}$.
\end{theorem*}

The actual statement in \cite{DelMoral:2004fe} contains a more general version of $(M_1)$ and a more explicit decomposition of the constant $C_p$; however the simpler version presented here is sufficient for the purposes of the present paper.

\subsubsection{Case A: Positive Function \texorpdfstring{$\phi(r) > 0$}{}}

First we prove Theorem \ref{thm:smc_error} as it is stated.
Later the assumption of $\phi > 0$ will be relaxed.

\paragraph{SMC-ND as a Feynman--Kac Model}

The aim here is to demonstrate that the SMC-ND algorithm fits into the framework of Section \ref{feynman-kac setup} for a specific McKean interpretation.
This connection will then be used to establish uniform convergence for the SMC-ND algorithm as a consequence of Theorem 7.4.4 in \cite{DelMoral:2004fe}.

For the state spaces we associate each $E_i = \mathcal{X}$ and $\mathcal{E}_i = \Sigma_{\mathcal{X}}$.
For the potentials we associate 
\begin{equation*}
G_i(x^i) = \frac{\phi\left( \frac{1}{\delta_{i+1}} \|A(x^i) - a\|_{\mathcal{A}} \right)}{\phi\left( \frac{1}{\delta_i} \|A(x^i) - a\|_{\mathcal{A}} \right)}
\end{equation*}
which clearly does not vanish and takes values in $(0,1]$ since $\delta_i > \delta_{i+1}$ and $\phi$ is decreasing.
For the Markov transitions we associate $\Gamma_{i+1}$ with $K_{i+1}$.

The Feynman--Kac measures associated with the SMC-ND algorithm can be cast as a non-homogeneous Markov chain with transitions $\Lambda_{i+1,\eta}$.
Here $\Lambda_{i+1,\eta_i}$ acts on the current measure $\eta_i$ on $\mathcal{X}$ by first propagating as $\eta_i K_{i+1}$ and then ``warping'' this measure with the potential $G_i$; i.e.
\begin{equation*}
\eta \Lambda_{i+1,\eta} = \frac{G_i}{\eta(G_i)} \eta \Gamma_{i+1} .
\end{equation*}
This demonstrates that the SMC-ND algorithm is the $P$-particle model corresponding to the McKean interpretation $\Lambda_{i+1,\eta}$ of the Feynman--Kac triplet $(\eta_0,G_i,\Gamma_i)$.
Thus the SMC-ND algorithm can be studied in the context of Section \ref{feynman-kac setup}, which we report next.

Note that it is common in applications of SMC to perform the ``Re-sample'' step before the ``Move'' step - our choice of order was required for the McKean framework that is the basis of the theoretical results in \cite{DelMoral2006}.
It is known in the SMC ``folk lore'' that the order of these steps can be interchanged.

\paragraph{Proof of Uniform Convergence Result for SMC-ND}

It remains to verify the hypotheses of Theorem 7.4.4 in \cite{DelMoral:2004fe}.
Condition $(G)$ is satisfied if and only if 
\begin{equation*}
\phi\left( \frac{1}{\delta_{i+1}} \|A(x^i) - a\|_{\mathcal{A}} \right)
\end{equation*} 
is bounded below, since 
\begin{equation*}
\phi\left( \frac{1}{\delta_i} \|A(x^i) - a\|_{\mathcal{A}} \right)
\end{equation*} 
is bounded above by 1.
Since $\phi$ is continuous, decreasing and satisfies $\phi > 0$ (Assumption \ref{assumption:positive}), it suffices to show that its argument $\frac{1}{\delta_{i+1}}\|A(x) - a\|_{\mathcal{A}}$ is upper-bounded.
This is the content of Assumption \ref{assumption:bounded_info} in the main text, which shows that 
\begin{align*}
	\frac{1}{\delta_i}\|A(x) - a\|_{\mathcal{A}}
	& \leq \frac{1}{\delta_i} \sup_{x \in \mathcal{X}} \|A(x)\|_{\mathcal{A}} + \|a\|_{\mathcal{A}} \\
	& =: \frac{1}{\epsilon_i^G} \; < \; \infty .
\end{align*}
Condition $(M_1)$ requires that
\begin{equation*}
	\Gamma_{i+1}(x^i , S) \geq \epsilon_i^\Gamma \Gamma_{i+1}(y^i, S)
\end{equation*}
for all $x^i,y^i \in E_i$ and $S \in \mathcal{E}_{i+1}$.
From construction this is equivalent to
\begin{equation*}
	K_{i+1}(x^i , S) \geq \epsilon_i^\Gamma K_{i+1}(y^i, S)
\end{equation*}
for all $x^i,y^i \in \mathcal{X}$ and $S \in \Sigma_{\mathcal{X}}$.
This is the content of Assumption \ref{assumption:composition}.

Thus we have established the hypotheses of Theorem 7.4.4 in \cite{DelMoral:2004fe} for the SMC-ND algorithm.
Theorem \ref{thm:smc_error} is a re-statement of this result.
For the statement of the result we used the $\|f\|_{\mathcal{F}}$ norm, based on the fact that (from Assumption \ref{assumption:norm_bounds}) $\|f_i\|_{\text{Osc}(E_i)} \leq 2 \|f\|_{\infty} \leq 2 C_{\mathcal{F}} \|f\|_{\mathcal{F}}$.

\subsubsection{Case B: Indicator Function \texorpdfstring{$\phi(r) = 1[r < 1]$}{}}

The previous analysis required that $\phi > 0$ on $\mathbb{R}_+$.
However, the most basic choice for $\phi$ is the indicator function $\phi(r) = 1[r < 1]$ which can take the value 0.
The case of an indicator function demands special attention, since Algorithm \ref{alg:rcp_smc} can fail in this case if all particles are assigned zero weight.
If this occurs, then we just define $\mu_{\delta,P}^\alpha(f) = 0$.
To be specific, the SMC-ND algorithm associated to the indicator function $\phi$ for approximation of the integral $\mu_\delta^a(f)$ is stated as Algorithm \ref{alg:rcp_smc2} next.

\begin{algorithm}
	Sample $x_j^0 \sim \mu$ for $j=1,\dots,P$ [Initialise] \\
	\For{$i=1,\dots,n$} {
		Sample $x^i_j \sim K_i(x^{i-1}_{j}, \cdot)$ for $j = 1,\dots,P$ [Sample] \\
		$E_i \gets \{x_j^i : x_j^i \in \mathcal{X}_{\delta_i}^a\}$ \\
		\If{$E_i  = \emptyset$} {
			Return $\mu_{\delta,P}^a(f) \gets 0$ \\
		}
		\For{$j = 1,\dots,P$} {
			\If{$x_j^i \notin E_i$} {
				$x_j^i \sim \text{Uniform}(E_i)$ [Re-sample] \\
			}
		}
	}
	Return $\mu_{\delta,P}^a(f) \gets \frac{1}{P}\sum_{j=1}^P f(x_j^n)$.
	\caption{Sequential Monte Carlo for Numerical Disintegration (SMC-ND), for the case where $\phi(r) = 1[r < 1]$.} \label{alg:rcp_smc2}
\end{algorithm}

Let $\mathcal{X}_\delta^a = \{x \in \mathcal{X} : \|A(x) - a\|_{\mathcal{A}} < \delta \}$.
If there is some iteration $i$ at which, after applying the kernel $K_i$ to each particle, no particle lies within $\mathcal{X}_{\delta_i}^a$, the algorithm fails. 
As a result it is critical to ensure that the distance between successive $\delta_i$ is small so that the probability of failure is controlled. 
This requirement is made formal next.
To establish the approximation properties of the random measure $\mu^a_{\delta_m,P}$, two assumptions are required.
These are intended to replace Assumptions \ref{assumption:positive}, \ref{assumption:composition} and Assumption \ref{assumption:bounded_info} from the main text:

\begin{assumption} \label{assumption:reachable}
For all $i = 0,\dots,m-1$ and all $x^i \in \mathcal{X}^a_{\delta_i}$, it holds that $K_{i+1}(x^i, \mathcal{X}^a_{\delta_{i+1}}) > 0$.
\end{assumption}

\begin{assumption} \label{assumption:composition2}
	For all $i = 0,\dots,m-1$ and all $x^i, y^i \in \mathcal{X}^a_{\delta_i}$, $K_{i+1}(x^i, \cdot) \ll K_{i+1}(y^i, \cdot)$. Furthermore there exist constants $\epsilon_i > 0$ such that the \radonnikodym derivative
	\begin{equation*}
		\frac{\wrt K_{i+1}(x^i, \cdot)}{\wrt K_{i+1}(y^i, \cdot)} \geq \epsilon_i .
	\end{equation*}
\end{assumption}

\noindent
Assumption~\ref{assumption:reachable} requires that the probability of reaching $\mathcal{X}^a_{\delta_{i+1}}$ when starting in $\mathcal{X}^a_{\delta_i}$ and applying the transition kernel $K_{i+1}$, is bounded away from zero.
Assumption~\ref{assumption:composition2} ensures that, for fixed $i$, transition kernels do not allocate arbitrarily large or small amounts of mass to different areas of the state space, as a function of their first argument.

\begin{theorem} \label{thm:smc_error2}
	For the alternative situation of an indicator function, it holds that for all $\delta \in \{\delta_i\}_{i=0}^m$ and fixed $p \geq 1$,
	\begin{equation*}
		\expected \left( 
			\left[\mu^a_{\delta, P}(f) - \mu^a_{\delta}(f) \right]^p
		\right)^{\frac{1}{p}}
		\leq \frac{C_p \norm{f}_{\mathcal{F}}}{\sqrt{P}}
	\end{equation*}
	for some constant $C_p$ independent of $P$ but dependent on $p$ and $\{\epsilon_i\}_{i=0}^{m-1}$.
\end{theorem}

\cite{Cerou:2012ha} proposed an algorithm similar to the one herein but focussed on approximation of the \emph{probability} of a rare event rather than sampling from the rare event itself.
In particular the theoretical results provided are in terms of these probabilities rather than how well the measure restricted to the rare event is approximated. 
Furthermore, many of the results therein focused upon an idealised version of the problem, in which it was assumed that the intermediate restricted measures can be sampled directly;
this avoids the issues with vanishing potentials indicated in \cite{DelMoral:2004fe}.
A similar algorithm was discussed in \cite{Scibior:2015} but was not shown to be theoretically sound.

The remainder of this Section establishes Theorem \ref{thm:smc_error2}.

\paragraph{SMC-ND as a Feynman--Kac Model}

The aim here is to demonstrate that Algorithm \ref{alg:rcp_smc2} fits into the framework of Section \ref{feynman-kac setup} for a specific McKean interpretation.
This is analogous to the proof of Theorem \ref{thm:smc_error}.

A technical complication is that the potentials $G_i$ must take values in $(0,1]$, which precludes the ``obvious'' choice of $E_i = \mathcal{X}$ and $G_i(x^i)$ as indicator functions for the sets $\mathcal{X}_{\delta_i}^a$.
Instead, we associate $E_i = \mathcal{X}_{\delta_i}^a$ and $\mathcal{E}_i$ with the corresponding restriction of $\Sigma_{\mathcal{X}}$.
For the potentials we then take $G_i(x^i) = 1$ for all $x_i \in E_i$, which clearly does not vanish and takes values in $(0,1]$.
For the Markov transitions $\Gamma_{i+1}$ from $E_i$ to $E_{i+1}$ we consider
\begin{equation*}
	\Gamma_{i+1}(x^i , \mathrm{d}x^{i+1}) \propto K_{i+1}(x^i,x^{i+1}) 
\end{equation*}
which is the restriction of $K_{i+1}$ to $E_{i+1}$.
For the latter to be well-defined it is required that the normalisation constant
\begin{equation*}
	\int_{E_{i+1}} K_{i+1}(x^i,x^{i+1}) \mathrm{d}x^{i+1} > 0
\end{equation*}
for all $x^i \in E_i$, so that there is a positive probability of reaching $E_{i+1}$ from $E_i$.
This is the content of Assumption \ref{assumption:reachable}.

The Feynman--Kac measures associated with Algorithm \ref{alg:rcp_smc2} can be cast as a non-homogeneous Markov chain with transitions $\Lambda_{i+1,\eta}$.
Here $\Lambda_{i+1,\eta_i}$ acts on the current measure $\eta_i$ on $E_i$ by first propagating as $\eta_i K_{i+1}$ and then restricting this measure to $E_{i+1}$.
This procedure is seen to be identical to the Markov transition $\Gamma_{i+1}$ defined above and, since the potentials $G_i \equiv 1$, it follows that
\begin{align*}
	\eta \Lambda_{i+1,\eta}
	& = \eta \Gamma_{i+1} \\
	& = \frac{G_i}{\eta(G_i)} \eta \Gamma_{i+1} .
\end{align*}
This demonstrates that Algorithm \ref{alg:rcp_smc2} is the $P$-particle model corresponding to the McKean interpretation $\Lambda_{i+1,\eta}$ of the Feynman--Kac triplet $(\eta_0,G_i,\Gamma_i)$.
Thus the SMC-ND algorithm can be studied in the context of Section \ref{feynman-kac setup}, which we report next.

\paragraph{Proof of Uniform Convergence Result for SMC-ND}

It remains to verify the hypotheses of Theorem 7.4.4 in \cite{DelMoral:2004fe}.
Condition $(G)$ is satisfied with no further assumption, since $G_i \equiv 1$ and we can take $\epsilon_i^G = 1$.
Condition $(M_1)$ requires that
\begin{equation*}
	\Gamma_{i+1}(x^i , S) \geq \epsilon_i^\Gamma \Gamma_{i+1}(y^i, S)
\end{equation*}
for all $x^i,y^i \in E_i$ and $S \in \mathcal{E}_{i+1}$.
From construction this is equivalent to
\begin{equation*}
	K_{i+1}(x^i , S) \geq \epsilon_i^\Gamma K_{i+1}(y^i, S)
\end{equation*}
for all $x^i,y^i \in E_i$ and $S \in \mathcal{E}_{i+1}$.
This is the content of Assumption \ref{assumption:composition2}.

Thus we have established the hypotheses of Theorem 7.4.4 in \cite{DelMoral:2004fe} for Algorithm \ref{alg:rcp_smc2} and in doing so have established Theorem \ref{thm:smc_error2}. 

\subsection{Parallel Tempering for Numerical Disintegration} \label{sec:pt}

Let $K_i$, $\set{\delta_i}_{i=1}^m$ be as in Section~\ref{sec:rcp_smc}. 
The PT algorithm \citep{Geyer:1991ws} for sampling from $\mu_{\delta_m}^a$ runs $m$ Markov chains in parallel, one for each temperature, by alternately applying $K_i$, then randomly proposing to ``swap'' the current state of two of the chains. 
Commonly only swaps of adjacent chains are considered; to this end suppose at iteration $j$ an index $q \in \{0,\dots,m-1\}$ has been selected. 
Denote by $x^q$ the state of the chain with $\mu^a_{\delta_q}$ as its invariant measure. 
Then to ensure the correct invariant distribution of all chains is maintained, the swap of state $x^q$ and $x^{q+1}$ is accepted with probability
\begin{equation}
	\alpha(x^q, x^{q+1}) = \frac{ \pi_q(x^{q+1}) \pi_{q+1}(x^q) }{ \pi_q(x^q) \pi_{q+1}(x^{q+1}) }
	\label{eq:pt_accept}
\end{equation}
where $\pi_q$ denotes the density of the target distribution $\mu^{a}_{\delta_q}$ with respect to a suitable reference measure. 
The density notation can be justified since in our experiments the sampler was applied to the finite-dimensional distributions $\mu^{a}_{\delta_q, N}$ and so the reference measure can be taken to be the Lebesgue measure on $\mathbb{R}^N$.

The PT algorithm for numerical disintegration is described in Algorithm~\ref{alg:rcp_pt}.
The samples $\{x_j^m\}_{j=1}^P$ are approximate draws from the distribution $\mu_{\delta_m}^a$.

\begin{algorithm}
	Given some initial $x_0^i$ for $i=1,\dots,m$ [Initialise] \\
	\For{$j=1,\dots,P$} {
		Sample $\hat{x}_j^i \sim K_i(x_{j-1}^i, \cdot)$ for $i = 1,\dots,m$ [Move] \\
		Sample $q \sim \text{Uniform}(0, m-1)$ \\
		\uIf{$U(0,1) < \alpha(x_j^q, x_j^{q+1})$} {
			Set $x_j^q = \hat{x}_j^{q+1}$ and $x_j^{q+1} = \hat{x}_j^q$ [Accept Swap]
		} 
		\uElse {
			Set $x_j^q = \hat{x}_j^{q}$ and $x_j^{q+1} = \hat{x}_j^{q+1}$ [Reject Swap]
		} 
		For $i \neq q, q+1$, set $x_j^i = \hat{x}_j^i$ [Update]
	}
\caption{Parallel Tempering for Numerical Disintegration}\label{alg:rcp_pt}
\end{algorithm}

Algorithms \ref{alg:rcp_smc} and \ref{alg:rcp_pt} are each valid for sampling from a target measure $\mu_\delta^a$. The choice of which algorithm to use is problem dependent, and each algorithm has been applied in the experiments in Section~\ref{sec:results}.

\subsection{Estimation of Model Evidence} \label{sec:evidence_computation}

The model evidence $p_A(a)$ was estimated as a by-product of the numerical disintegration algorithm developed.
Attention is restricted to the specific relaxation function $\phi(r) = \exp(-r^2)$.
Then the thermodynamic integral identity \citep{Gelman1998} can be exploited to calculate the model evidence:
\begin{align*}
\log p_A(a) & = - \lim_{\delta \downarrow 0} \frac{1}{\delta^2} \int_0^1 \int \|A(x) - a\|_{\mathcal{A}}^2 \; \wrt \mu_{\delta / \sqrt{t}}^a \; \wrt t 
\end{align*}
where the parameterisation $\delta \mapsto \delta/\sqrt{t}$ is such that $t=0$ corresponds to the prior, while $t=1$ corresponds to the distribution $\mu^a_\delta$.

To approximate this integral, the outer integral is first discretised.
To this end, fix a sequence $\infty = \delta_0 < \delta_1 < \dots < \delta_m$ of relaxation parameters. For convenience this may be the same sequence as used to apply numerical disintegration.
Then for $\delta_m$ small, and letting $\sqrt{t_i} = \delta_m / \delta_i$:
\begin{align*}
\log p_A(a) & \approx - \frac{1}{\delta_m^2} \sum_{i=1}^m (t_i - t_{i-1}) \int \|A(x) - a\|_{\mathcal{A}}^2 \; \wrt \mu_{\delta_m / \sqrt{t_i}}^a
\end{align*}
Thus we obtain a consistent approximation
\begin{align*}
\log p_A(a) & \approx - \sum_{i=1}^m \left(\frac{1}{\delta_i^2} - \frac{1}{\delta_{i-1}^2}\right) \underbrace{\int \|A(x) - a\|_{\mathcal{A}}^2 \; \wrt \mu_{\delta_i}^a}_{(*)}
\end{align*}
The terms $(*)$ were estimated via Monte Carlo, based on samples from the distributions $\mu_{\delta_i}^a$ obtained through numerical disintegration.
Higher-order quadrature rules and variance reduction techniques can be used, but were not implemented for this work \citep{Oates2016}.

\subsection{Monte Carlo Details for Painlev\'{e} Transcendental}

Sampling of the posterior was performed for a temperature schedule of $m=1600$ steps, equally spaced on a logarithmic scale from $10$ to $10^{-4}$, for an ensemble of $P=200$ particles. 

Specification of appropriate transition kernels $K_i$ for this problem was challenging due both to the high dimension and the empirical observation that, for small $\delta$, mixing of the chains tends to be poor. This is likely due to the nonlinearity of the information operator which leads to highly a complex posterior structure. 
For this reason a gradient-based sampler was used to construct the transition kernel; the Metropolis-adjusted Langevin algorithm (MALA) \citep{Roberts:1996}. 

Denote by $u^k$ the coefficients $[u^k_j]_{j=1}^N$ at iteration $k$ of MALA. Then, recall that MALA has proposals given by
\begin{equation*}
	u^{k+1} = u^k + \tau_i \Gamma \nabla \log \pi_i(u^k) + \sqrt{2\tau_i \Gamma} W
\end{equation*}
where $W$ is a standard Gaussian distribution and $\Gamma \in \reals^{N\times N}$ is a positive definite preconditioning matrix. 
The $\tau_i$ were taken to be fixed for each kernel $K_i$ to a value found empirically to provide a reasonable acceptance rate. $\pi_i$ denotes the unnormalised target distribution for $K_i$, here given by
\begin{equation*}
	\pi_i(u^k) = \phi\left(\frac{\norm{A x^N - a}}{\delta_i}\right) q^N(u^k)
\end{equation*}
where $x^N = \sum_{i=0}^N u_i \phi_i$ and $q^N(\cdot)$ denotes the prior density of the coefficients $[u_j]_{j=1}^N$.

To ensure proposals were scaled to match the decay of the prior for the coefficients, we took $\Gamma=\textrm{diag}(\gamma)$, the diagonal matrix which has the coefficients $\gamma_i$ on its diagonal. 
Even with such a transition kernel, mixing is generally poor. To compensate $k$ was taken to be large; for $n = 12,17$ we took $k=10,000$, while for $n=22$ we took $k=40,000$. We note that such a large number of temperature levels and transitions makes computation expensive, highlighting the importance of future work toward methods for approximating the Bayesian posterior in a more computationally efficient manner.

\subsection{Monte Carlo Details for Poisson Equation}

The posterior distribution was obtained by use of the PT algorithm, for $m=20$ temperatures equally spaced on a logarithmic scale between $10^{-2}$ and $10^{-4}$. 
The transition kernels $K_i$ were given by 10 iterations of a MALA sampler, with preconditioner as described earlier and parameter $\tau$ chosen to achieve a good acceptance rate.
The number of iterations $P$ was taken to be $10^6$ when $n = 25$ and $10^7$ when $n=25$ or $n=36$.

\section{Truncation of the Prior Distribution (Proof of Theorem \ref{thm:truncation error})}

In this section we present the proof of Theorem \ref{thm:truncation error} in the main text.
We use a general result on the well-posedness of Bayesian inverse problems:

\begin{theorem}[Theorem 4.6 in \cite{Sullivan:2016wp}] \label{thm:tim}
	Let $\mathcal{X}$ and $\mathcal{A}$ be separable quasi-Banach spaces over $\mathbb{R}$.
	Suppose that
	\begin{equation}
		\frac{\mathrm{d} \mu_\delta^a}{\mathrm{d} \mu} = \frac{\exp(-\Phi_\delta(x;a))}{Z_\delta^a} 
	\end{equation}
	where the potential function $\Phi_\delta$ satisfies:
	\begin{enumerate}
		\item[S0] $\Phi_\delta(x; \cdot)$ is continuous for each $x \in \mathcal{X}$, $\Phi_\delta(\cdot; a)$ is measurable for each $a \in \mathcal{A}$, and for every $r > 0$, there exists $M_{0, r , \delta} \in \mathbb{R}$ such that, for all $(x, a) \in \mathcal{X} \times \mathcal{A}$ with $\| x \|_{\mathcal{X}} < r$ and $\| a \|_{\mathcal{A}} < r$,
		\[
			| \Phi_\delta(x; a) | \leq M_{0, r , \delta}.
		\]
		
		\item[S1] For every $r > 0$, there exists a measurable $M_{1, r , \delta} \colon \mathbb{R}_{+} \to \mathbb{R}$ such that, for all $(x, a) \in \mathcal{X} \times \mathcal{A}$ with $\| a \|_{\mathcal{A}} < r$,
		\[
			\Phi_\delta(x; a) \geq M_{1, r , \delta} \bigl( \| x \|_{\mathcal{X}} \bigr) .
		\]

		\item[S2] For every $r > 0$, there exists a measurable $M_{2, r , \delta} \colon \mathbb{R}_{+} \to \mathbb{R}_{+}$ such that, for all $(x, a, \tilde{a}) \in \mathcal{X} \times \mathcal{A} \times \mathcal{A}$ with $\| a \|_{\mathcal{A}} < r$, $\| \tilde{a} \|_{\mathcal{A}} < r$,
		\[
			| \Phi_\delta(x; a) - \Phi_\delta(x; \tilde{a}) | \leq \exp \bigl( M_{2, r , \delta} \bigl( \| x \|_{\mathcal{X}} \bigr) \bigr) \| a - \tilde{a} \|_{\mathcal{A}} .
		\]

	\end{enumerate}
	Let $\Phi_{\delta,N}$ be an approximation to $\Phi_\delta$ that satisfies (S1-S3) with $M_{i, r , \delta}$ independent of $N$, and such that
	\begin{enumerate}
		\item[S3] $\Psi \colon \mathbb{N} \to \mathbb{R}_{+}$ is such that, for every $r > 0$, there exists a measurable $M_{3, r , \delta} \colon \mathbb{R}_{+} \to \mathbb{R}_{+}$, such that, for all $(x, a) \in \mathcal{X} \times \mathcal{A}$ with $\| a \|_{\mathcal{A}} < r$,
		\[
			| \Phi_{\delta,N}(x; a) - \Phi_\delta(x; a) | \leq \exp \bigl( M_{3, r , \delta} \bigl( \| x \|_{\mathcal{X}} \bigr) \bigr) \Psi(N) .
		\]

		\item[S4] For some $r > 0$,
		\begin{equation}
			\mathbb{E}_{X \sim \mu} \bigl[ \exp ( 2 M_{3, r , \delta} ( \| X \|_{\mathcal{X}} ) - M_{1, r , \delta} ( \| X \|_{\mathcal{X}} ) ) \bigr] < \infty.
		\end{equation}
	\end{enumerate}
	
	Let $d_{\text{H}}$ denote the Hellinger distance on $\mathcal{P}_{\mathcal{X}}$.
	Then there exists a constant $C_\delta$, independent of $N$, such that 
	\[
		d_{\text{H}} \bigl( \mu_{\delta,N}^a, \mu_\delta^a \bigr) \leq C_\delta \Psi(N)
	\]
	where $\mu_{\delta,N}^a$ is the posterior distribution based on the potential function $\Phi_{\delta,N}$ instead of $\Phi_\delta$.
\end{theorem}

This allows us to establish conditions on $A$ and $\mu$ that guarantee stability under truncation of the prior:

\begin{proof}[Proof of Theorem \ref{thm:truncation error}]
	Let $\varphi$ be as in Section \ref{sec:rcp_approximation}, and let 
	\begin{align*}
		\Phi_\delta(x;a) & = \varphi\left( \frac{\|A(x) - a\|_{\mathcal{A}}}{\delta} \right) \\
		\Phi_{\delta,N}(x;a) & = \varphi\left( \frac{\|A \circ P_N (x) - a\|_{\mathcal{A}}}{\delta} \right) .
	\end{align*}
	Our task is to check the conditions of Theorem \ref{thm:tim} hold for $\Phi_\delta$ and $\Phi_{\delta,N}$.
	\begin{enumerate}
		\item[S0] First, note that $\Phi_\delta(x;\cdot)$ is continuous (since $\varphi$ is continuous from Assumption \ref{varphi_assumption} and $\Phi_\delta(x;\cdot)$ is a composition of continuous functions) and that $\Phi_\delta(\cdot;a)$ is measurable (since $\phi$ is measurable and $\Phi_\delta(\cdot;a)$ is a composition of measurable functions).
		Second, note that $\varphi$ is a continuous bijection from $(0,\infty)$ to itself with $\varphi(0) = 0$.
		Thus $\varphi^{-1}$ exists and we can consider
		\begin{align*}
			\delta \varphi^{-1} \sup \{ |\Phi_\delta(x;a)| : \|x\|_{\mathcal{X}} , \|a\|_{\mathcal{A}} < r \}
			& = \sup \{ \|A(x) - a\|_{\mathcal{A}} : \|x\|_{\mathcal{X}} , \|a\|_{\mathcal{A}} < r \} \\
			& \leq \sup_{x \in \mathcal{X}} \|A(x)\|_{\mathcal{A}} + r \\
			& \leq \infty \quad \text{(Assumption \ref{assumption:bounded_info})} .
		\end{align*}
		Thus we can take $M_{0,r,\delta} = \varphi(\frac{1}{\delta} \sup_{x \in \mathcal{X}} \|A(x)\|_{\mathcal{A}} + \frac{r}{\delta} )$.

		\item[S1] Since $\Phi_\delta(x;a) \geq 0$ we can take $M_{1,r,\delta} = 0$.

		\item[S2] Given $r > 0$ let $R = \frac{1}{\delta}\sup_{x \in \mathcal{X}} \|A(x)\|_{\mathcal{A}} + \frac{r}{\delta}$, which is finite by Assumption \ref{assumption:bounded_info}.
		The upper bound
		\begin{align*}
			|\Phi_\delta(x;a) - \Phi_\delta(x;\tilde{a})|
			& = \left|\varphi\left(\frac{\|A(x) - a\|_{\mathcal{A}}}{\delta} \right) - \varphi\left(\frac{\|A(x) - \tilde{a}\|_{\mathcal{A}}}{\delta} \right) \right| \\
			& \leq C_R \left|\frac{\|A(x) - a\|_{\mathcal{A}}}{\delta} - \frac{\|A(x) - \tilde{a}\|_{\mathcal{A}}}{\delta} \right| \quad \text{(Assumption \ref{assumption:local_Lipschitz})} \\
			& \leq \frac{C_R}{\delta} \|a - \tilde{a}\|_{\mathcal{A}} \quad \text{(reverse triangle inequality)} 
		\end{align*}
		demonstrates that we can take $M_{2,r,\delta} = \max\{0 , \log(\frac{C_R}{\delta})\}$.
	\end{enumerate}

	Minor variation on the above arguments show that S1-3 also hold for $\Phi_{\delta,N}$ with the same constants $M_{i,r,\delta}$.

	\begin{enumerate}
		\item[S3] Let $C_R$ be defined as in S2.
		The upper bound
		\begin{align*}
			|\Phi_{\delta,N}(x;a) - \Phi_\delta(x;a)|
			& = \left|\varphi\left(\frac{\|A \circ P_N(x) - a\|_{\mathcal{A}}}{\delta} \right) - \varphi\left(\frac{\|A(x) - a\|_{\mathcal{A}}}{\delta} \right) \right| \\
			& \leq C_R \left|\frac{\|A \circ P_N(x) - a\|_{\mathcal{A}}}{\delta} - \frac{\|A(x) - a\|_{\mathcal{A}}}{\delta} \right| \quad \text{(Assumption \ref{assumption:local_Lipschitz})} \\
			& \leq \frac{C_R}{\delta} \|A \circ P_N(x) - A(x)\|_{\mathcal{A}} \quad \text{(reverse triangle inequality)} \\
			& \leq \frac{C_R}{\delta} \exp(m(\|x\|_{\mathcal{X}})) \Psi(N) \quad \text{(Assumption \ref{assumption:prior_truncation})}
		\end{align*}
		demonstrates that we can take $M_{3,r,\delta}(\|x\|_{\mathcal{X}}) = \max\{0 , \log(\frac{C_R}{\delta}) + m(\|x\|_{\mathcal{X}})\}$.

		\item[S4] Let $C_R$ be defined as in S2.
		The upper bound
		\begin{align*}
			\mathbb{E}_{X \sim \mu} [\exp(2M_{3,r,\delta}(\|X\|_{\mathcal{X}}) - M_{1,r,\delta}(\|X\|_{\mathcal{X}}))]  & = \mathbb{E}_{X \sim \mu} [\exp(2\max\{0, \log(C_R / \delta) + m(\|X\|_{\mathcal{X}})\} ) ] \\
			& \leq 1 + \frac{C_R}{\delta} \mathbb{E}_{X \sim \mu}[\exp(2m(\|X\|_{\mathcal{X}}))] \\
			& < \infty \quad \text{(Assumption \ref{assumption:prior_truncation})}
		\end{align*}
		establishes the last of the conditions for Theorem \ref{thm:tim} to hold.
	\end{enumerate}
	Thus from Theorem \ref{thm:tim}, $d_{\text{H}} \bigl( \mu_{\delta,N}^a, \mu_\delta^a \bigr) \leq C_\delta \Psi(N)$.
	The proof is completed since Assumption \ref{assumption:norm_bounds} implies that $d_{\mathcal{F}} \leq C_{\mathcal{F}}^{-1} d_{\text{TV}}$ where $d_{\text{TV}}$ is the total variation distance based on $\mathcal{F} = \{f : \|f\|_\infty \leq 1\}$; in turn it is a standard fact that $d_{\text{TV}} \leq \sqrt{2} d_{\text{H}}$.
\end{proof}

\end{document}